\documentclass[aps,prx,superscriptaddress,twocolumn,nofootinbib,10pt,longbibliography,floatfix]{revtex4-2}
\makeatletter
\@ifundefined{pdfmapfile}{}{
    \pdfmapfile{+euler.map}%
}
\makeatother
\usepackage[utf8]{inputenc}
\usepackage[T1]{fontenc}

\makeatletter
\renewcommand{\bibinfo}[2]{%
  \ifstrequal{#1}{title}{\textit{#2}}{#2}}
\makeatother

\usepackage{mathpazo}

\usepackage[toc,title,page]{appendix}

\usepackage{csquotes}
\usepackage[italian,english]{babel}

\usepackage{bm}

\usepackage{dsfont} 
\usepackage{amsmath,amssymb,amsthm,thmtools}
\usepackage{mathtools}
\usepackage{cases}
\usepackage{calc}
\usepackage{mathrsfs} 
\usepackage[normalem]{ulem} 

\usepackage{parskip}
\usepackage[colorlinks=true]{hyperref}
\usepackage{bookmark}
\usepackage[nameinlink]{cleveref}

\usepackage{smtoc}
\usepackage{orcidlink}

\crefalias{appsec}{section}

\crefname{appsec}{}{}
\Crefname{appsec}{}{}
\crefname{conjecture}{conjecture}{conjectures}
\crefname{box}{Box}{Box}
\hypersetup{
  colorlinks   = true, 
  urlcolor     = green!80!black, 
  linkcolor    = blue, 
  citecolor    = red!80!black 
}

\usepackage{physics} 

\usepackage{graphicx}

\usepackage[dvipsnames]{xcolor}

\usepackage{placeins}
\usepackage{multirow,booktabs}
\newcommand{\proofsymcell}[1]{%
    \parbox[t]{0.22\linewidth}{\raggedright #1\strut}}
\newcommand{\proofdefcell}[1]{%
    \parbox[t]{0.72\linewidth}{\raggedright #1\strut}}
\newcommand{\proofgroup}[1]{%
    \addlinespace[0.8ex]
    \multicolumn{2}{@{}l@{}}{\bfseries #1}\\[-0.2ex]
}

\usepackage{tcolorbox}
\tcbuselibrary{theorems,breakable}

\graphicspath{{./figures/}}

\newcommand{\RR}{\mathbb{R}}
\newcommand{\CC}{\mathbb{C}}

\newcommand{\EE}{\mathbb{E}}
\newcommand{\PP}{\mathbb{P}}

\newcommand{\rmd}{\mathrm{d}}

\newcommand{\calD}{\mathcal{D}}
\newcommand{\calE}{\mathcal{E}}

\newcommand{\calH}{\mathcal{H}}
\newcommand{\calL}{\mathcal{L}}
\newcommand{\calM}{\mathcal{M}}
\newcommand{\calN}{\mathcal{N}}
\newcommand{\calO}{\mathcal{O}}
\newcommand{\calP}{\mathcal{P}}

\newcommand{\bs}[1]{\boldsymbol{#1}}
\newcommand{\on}[1]{\operatorname{#1}}

\newcommand{\parTitle}[1]{\textit{#1 ---}}

\NewCommandCopy{\rawcite}{\cite}
\NewCommandCopy{\rawcref}{\cref}
\NewDocumentCommand{\safecite}{om}{%
  \mbox{\IfValueTF{#1}{\rawcite[#1]{#2}}{\rawcite{#2}}}%
}
\NewDocumentCommand{\safecref}{m}{%
  \mbox{\rawcref{#1}}%
}
\NewDocumentCommand{\addLI}{+m +g}{%
  {\color{RoyalBlue}#1}%
  \IfNoValueF{#2}{%
    \begingroup
      \let\cite\safecite
      \let\cref\safecref
      \color{red}\small\sout{#2}%
    \endgroup
  }%
}

\DeclareMathOperator{\Herm}{Herm}

\newcommand{\MSE}{\text{MSE}}

\DeclareMathOperator*{\E}{\mathbb{E}}
\newcommand{\im}{{\rm im}}
\newcommand{\supp}{{\rm supp}}
\newcommand{\coker}{{\rm coker}}
\newcommand{\diag}{\on{diag}}
\newcommand{\bias}{\on{bias}}
\newcommand{\Var}{\on{Var}}
\newcommand{\Cov}{\on{Cov}}

\newcommand{\poly}{\on{p}}

\newcommand{\nout}{n_{\rm out}}

\newcommand{\ntr}{n_{\rm tr}}

\newcommand{\bsp}{\bs{p}}
\newcommand{\bsx}{\bs{x}}
\newcommand{\bsv}{\bs{v}}
\newcommand{\bsw}{\bs{w}}

\newcommand{\bsA}{\bs{A}}
\newcommand{\bsB}{\bs{B}}
\newcommand{\bsC}{\bs{C}}
\newcommand{\bsD}{\bs{D}}

\newcommand{\bsI}{\bs{I}}
\newcommand{\bsM}{\bs{M}}
\newcommand{\bsP}{\bs{P}}
\newcommand{\bsQ}{\bs{Q}}
\newcommand{\bsS}{\bs{S}}
\newcommand{\bsW}{\bs{w}}
\newcommand{\bsX}{\bs{X}}

\newcommand{\bsMrho}{\bsM_{\bs\rho}}
\newcommand{\bsMmu}{\bsM_{\bs\mu}}
\newcommand{\bsMtildemu}{\bsM_{\tilde{\bs\mu}}}

\newcommand{\Id}{\on{Id}}
\newcommand{\bstildemu}{\bs{\tilde{\mu}}}

\newcommand{\hatbs}[1]{\hat{\bs{#1}}}

\newcommand{\hatbsC}{\hat{\bs C}}

\newcommand{\hatbsM}{\hat{\bs M}}
\newcommand{\hatbsP}{\hat{\bs P}}
\newcommand{\hatbsQ}{\hat{\bs Q}}
\newcommand{\hatbsS}{\hat{\bs S}}

\newcommand{\hatbsZ}{\hat{\bs Z}}
\newcommand{\hatbsW}{{\hat{\bs w}}}
\newcommand{\hatbsXi}{\hatbs{\Xi}}
\newcommand{\hsXi}{\hatbsXi}
\newcommand{\hatbsp}{\hat{\bs p}}

\newcommand{\blueMSchurComplementTwo}{(\hatbsM/\hatbsM_{22})}
\newcommand{\blueMSchurComplementTwoInv}{(\hatbsM/\hatbsM_{22})^{-1}}

\newcommand{\bsmu}{\boldsymbol{\mu}}
\newcommand{\vect}[1]{\mathrm{vec}(#1)}

\newcommand{\frameOp}{\mathscr{F}}

\newcommand{\barSigma}{\overline{\bs{\Sigma}}}

\newtheorem{theorem}{Theorem}

\newtheorem{lemma}{Lemma}
\newtheorem{corollary}{Corollary}

\theoremstyle{definition}
\newtcbtheorem[crefname={example}{examples}]{example}{Example}%
{
  parbox=false,
  breakable,
  colback=red!0,          
  colframe=gray!10,         
  fonttitle=\bfseries,      
  arc=2mm,                  
  boxrule=0.2pt,           
  left=3mm,                
  right=3mm,               
  top=2mm,                 
  bottom=2mm,
  fonttitle=\bfseries\color{black}
}{}

\begin{document}

\title{
Who can sample forever? Shot noise effects in quantum linear regression
}

\author{Gabriele Lo Monaco\,\orcidlink{0000-0002-3594-3477}}
\email{gabriele.lomonaco@unipa.it}
\let\comma,
\affiliation{Universit\`a degli Studi di Palermo\comma{} Dipartimento di Fisica e Chimica - Emilio Segr\`e\comma{} via Archirafi 36\comma{} I-90123 Palermo\comma{} Italy}

\author{Salvatore Lorenzo\,\orcidlink{0000-0002-0827-5549}}
\let\comma,
\affiliation{Universit\`a degli Studi di Palermo\comma{} Dipartimento di Fisica e Chimica - Emilio Segr\`e\comma{} via Archirafi 36\comma{} I-90123 Palermo\comma{} Italy}

\author{Alessandro Ferraro\,\orcidlink{0000-0002-7579-6336}}
\affiliation{Quantum Technology Lab\comma{} Dipartimento di Fisica Aldo Pontremoli\comma{} Universit\`a degli Studi di Milano\comma{} I-20133 Milano\comma{} Italy}

\author{Mauro Paternostro\,\orcidlink{0000-0001-8870-9134}}
\affiliation{Universit\`a degli Studi di Palermo\comma{} Dipartimento di Fisica e Chimica - Emilio Segr\`e\comma{} via Archirafi 36\comma{} I-90123 Palermo\comma{} Italy}
\affiliation{Centre for Quantum Materials and Technologies\comma{} School of Mathematics and Physics\comma{} Queen's University Belfast\comma{} BT7 1NN\comma{} United Kingdom}

\author{G. Massimo Palma\,\orcidlink{0000-0001-7009-4573}}
\let\comma,
\affiliation{Universit\`a degli Studi di Palermo\comma{} Dipartimento di Fisica e Chimica - Emilio Segr\`e\comma{} via Archirafi 36\comma{} I-90123 Palermo\comma{} Italy}

\author{Luca Innocenti\,\orcidlink{0000-0002-7678-1128}}
\let\comma,
\affiliation{Universit\`a degli Studi di Palermo\comma{} Dipartimento di Fisica e Chimica - Emilio Segr\`e\comma{} via Archirafi 36\comma{} I-90123 Palermo\comma{} Italy}

\begin{abstract}
Quantum fidelity kernels and quantum extreme learning machines share the same basic architecture: a fixed quantum device generates features that are combined through a classically trained linear model. We refer to this broad non-variational setting as quantum linear regression. A defining aspect of these models is that the quantum features are not directly available, but must be estimated from a finite number of shots. We show that finite statistics makes quantum linear regression unavoidably ill-conditioned. Ideal feature vectors are confined to a subspace determined by the dimension of the Hilbert space, whereas finite-shot estimates generically acquire components outside this subspace. Training then attempts to fit directions that disappear in the infinite-statistics limit, making the learned model increasingly sensitive to sampling noise.
We analyze this mechanism for quantum extreme learning machines trained on \(n_{\rm tr}\) states with \(N\) shots per state. Despite the ill-conditioning, predictions remain controlled when the training targets are exact. Finite sampling produces a systematic prediction error that decreases quadratically with \(N\), while fluctuations induced by the random training measurements decrease inversely with the total training budget. Thus, unlike in standard linear regression with exactly known features, increasing the dataset size alone cannot remove the error caused by sampling noise.
If the training targets are also noisy, due for instance to state-preparation errors, the otherwise hidden directions become active. The test error can then increase with $N$, a behavior that we ascribe to a form of \emph{quantum overfitting}. Our results expose a fundamental tradeoff between dataset size and per-state statistics in non-variational quantum learning and provide a baseline for resource-aware training and regularization strategies.
\end{abstract}

\maketitle
\bookmarksetup{startatroot}

\section{Introduction}


\begin{figure*}[t]
    \centering
    \includegraphics[width=1.\linewidth]{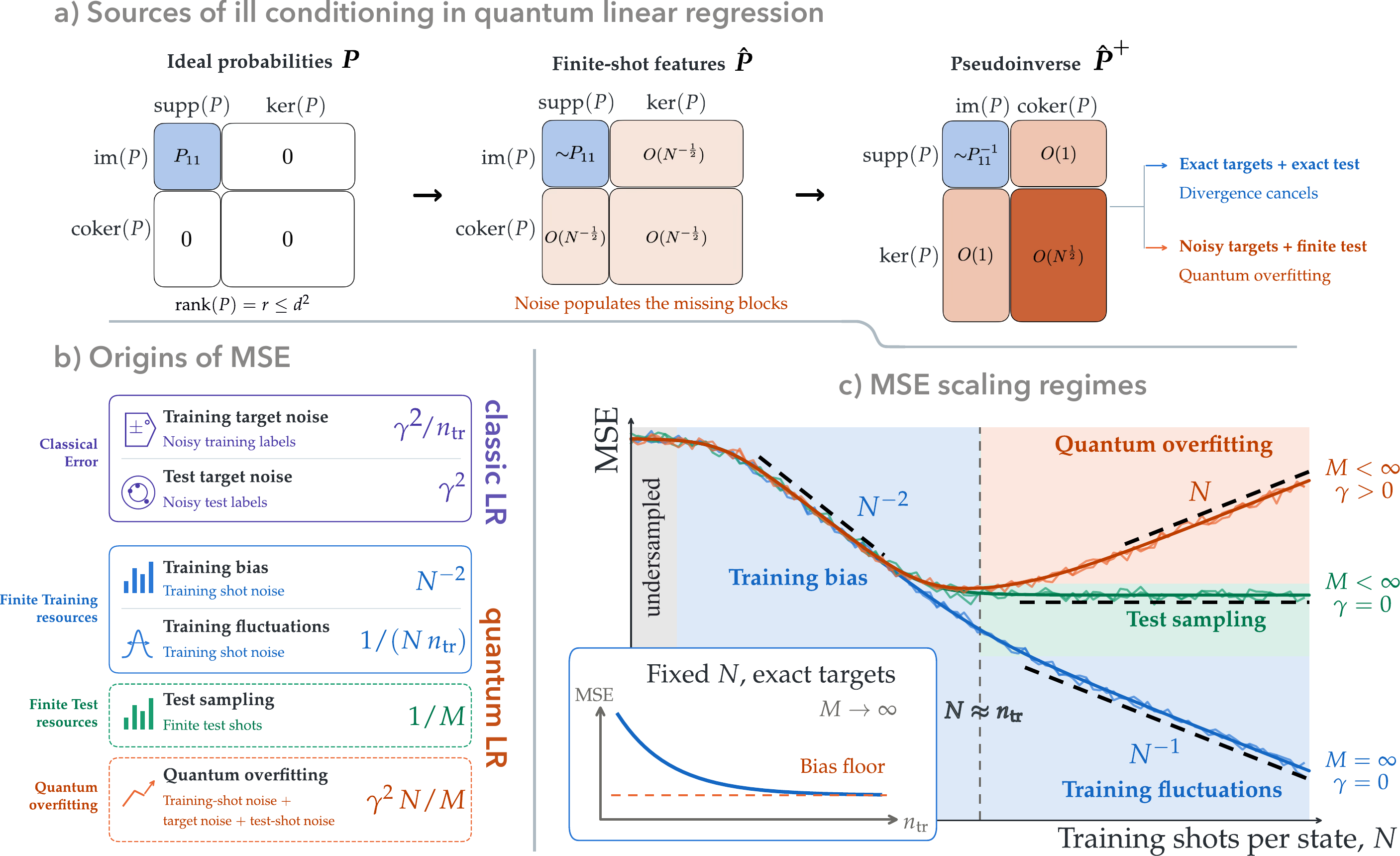}
\caption{
Panel \textbf{a)}: The ideal probability matrix has $\rank(\bsP)=d^2$, while the
empirical frequency matrix $\hatbsP_N=\bsP+N^{-1/2}\hatbsXi$, constructed by
measuring each state $N$ times, is full rank with high probability. The
spurious directions carry singular values at the noise scale $N^{-1/2}$, so
$\hatbsP_N^{+}$ develops a part that grows as $\sqrt N$ and maps $\coker(\bsP)$
into $\ker(\bsP)$ (\cref{eq:Delt_PN_small}). Exact targets
($\bs y\in\supp(\bsP)$) and exact test probabilities ($\bsp_\sigma\in\im(\bsP)$)
are orthogonal to these subspaces, so the divergence cancels from the
prediction $\bs y^T\hatbsP_N^{+}\bsp_\sigma$. Isotropic target noise and the
test covariance $\bs\Sigma_\sigma$ instead both have components in $\ker(\bsP)$:
only when the two occur together does the divergence survive and the error grow
with $N$ --- \emph{quantum overfitting}.
Panel \textbf{b)}: The error splits according to its origin. Training shot noise alone
produces a systematic bias $b^2/N^2$ and a fluctuation $V/(N\ntr)$
(\cref{eq:multinomial_b_V_main}), neither of which has a counterpart in
classical regression with exact features. Target noise of variance $\gamma^2$
acts classically: an irreducible $\gamma^2$ from the test labels, and an
$O(\gamma^2/\ntr)$ floor on the training variance. Finite test statistics add
$\bsW_\calO\bs\Sigma_\sigma\bsW_\calO^T/M$, a floor that no amount of training
data can remove (\cref{eq:fixed_W_noise_floor}). Only the product of all three
--- training shots, target noise and test shots --- reaches the divergence of
\textbf{a)}, giving $\gamma^2NC/M$ with $C=(\nout-d^2)/(\ntr-\nout-1)$, whose
numerator counts the directions populated by noise alone: the amplification is
a direct measure of the rank mismatch.
Panel \textbf{c)}: Schematic MSE versus $N$ at fixed $\ntr$. With exact
targets and exact test features (teal, $M\to\infty$) the bias $N^{-2}$ dominates
at low statistics and the fluctuation $(N\ntr)^{-1}$ at high statistics; the two
balance at $N\sim\ntr$, so it is the training-set size that fixes where the
decay slows down.
Finite test statistics arrest the decay at the $1/M$ floor (grey). With noisy
targets (orange) the term $\gamma^2NC/M$ eventually dominates and the curve
turns upward: the MSE is minimized at a finite $N^\star$, beyond which
additional shots per training state actively hurt, and $N^\star$ moves to larger
values only if $M$ is increased with it. Inset: at fixed $N$, more training
states suppress the fluctuation but leave the $N^{-2}$ bias floor --- unlike
classical regression, data alone cannot push the error below it.}
    \label{fig:scheme}
\end{figure*}


Linear models provide a widely used alternative to gradient-based learning methods in classical machine learning, owing in part to their favorable training properties and the availability of efficient optimization procedures \cite{scholkopf2001learning,hastie2009Elements,huang2006extreme,jaeger2004harnessing}. A similar distinction arises in quantum machine learning (QML)~\cite{cerezo2022ChallengesOpportunitiesQuantum,wang2024ComprehensiveReviewQuantum}. While much attention has focused on variational quantum algorithms~\cite{cerezo2021VariationalQuantumAlgorithms}, their trainability can be severely hindered by phenomena such as barren plateaus~\cite{mcclean2018barren,arrasmith2022equivalence,cerezo2025does,larocca2025barren} and, more generally, exponential concentration~\cite{holmes2022connecting,kairon2025equivalence,tangpanitanon2020expressibility,stilck2021limitations,li2022concentration}. Non-variational quantum learning architectures, in which the quantum processing is fixed and only a classical model is optimized during training, therefore provide a promising alternative. 
Many such architectures share a common structure that we refer to as \textit{quantum linear regression}: a fixed quantum device maps input data into a quantum feature space, while inference is performed through a linear combination of expectation values of observables measured at its output, with the corresponding weights learned classically from training data. Equivalently, these architectures implement linear models in quantum feature space~\cite{schuld2019QuantumMachineLearning,schuld2021SupervisedQuantumMachine,jerbi2023QuantumMachineLearning,innocenti2023potential}. Prominent examples include quantum reservoir computing~\cite{fujii2017HarnessingDisorderedQuantum,mujal2021OpportunitiesQuantumReservoir}, quantum extreme learning machines (QELMs)~\cite{innocenti2023potential,xiong2025fundamental,vetrano2025state,lomonaco2024quantum,suprano2024experimental,zia2025quantum}, and quantum kernel methods~\cite{schuld2021SupervisedQuantumMachine,havlicek2019SupervisedLearningQuantum,kubler2021inductive,jager2023universal,gentinetta2024complexity,thanasilp2024exponential,kempkes2026double}; the same structure also underlies the training of recently proposed quantum Gaussian processes~\cite{jager2026provable,jager2026QuantumGaussianProcesses}.

Despite inheriting the favorable training properties of classical linear models, quantum linear regression introduces a fundamental complication with no direct counterpart when classical features are exactly accessible: quantum features cannot be read out directly, but must be estimated from a finite number of measurement shots. Their intrinsic sampling noise, therefore, makes the number of shots $N$ per training example an additional computational resource, alongside the number of training examples $\ntr$ themselves. Understanding how these two resources jointly determine learning performance is thus key to assess the practical requirements of non-variational QML~\cite{sweke2020StochasticGradientDescent,moussa2023ResourceFrugalOptimizer,hu2024OvercomingCoherenceTime,ahmed2025OptimalTrainingFinitely,hu2023tackling,gentinetta2024complexity}. In this work, we characterize this interplay and show that finite-shot estimation fundamentally alters the behavior of quantum linear regression, giving rise to intrinsic ill-conditioning and to a nontrivial trade-off between dataset size and per-state measurement statistics.


More specifically, the features estimated using $N$ shots per state are collected into a finite resource feature matrix $\hatbsP_N$, which converges in the infinite-shot limit to the ideal feature matrix $\bsP$.
Collecting training labels in a matrix $\bs Y$, linear regression seeks $\bs W$ such that $\bs Y\simeq\bs W\hatbsP_N$, with minimum-norm least-squares solution $\bs W=\bs Y\hatbsP_N^+$.
However, ideal and empirical feature matrices have different rank structures: $\bsP$ is rank deficient, while stochastic fluctuations make $\hatbsP_N$ full rank with high probability.
This rank discontinuity manifests itself as divergences in $\hatbsP_N^+$~\cite{wedin1973perturbation,stewart1977perturbation,stewart1990Matrix,gonzalez2015laurent}, posing a central technical obstacle to training these models.
As we show in \cref{sec:qlr_acuteness}, the argument only relies on the dimension of the space in which the features live, and is therefore independent of how they are generated: whether they are measurement probabilities, as in QELMs, or overlaps between states, as in quantum fidelity kernels, the ill-conditioning is a generic feature of quantum linear regression.

While the statistical properties of classical linear regression have been studied extensively \cite{bishop1995training,dobriban2018high,belkin2018understand,bibas2019new,liang2020just,bartlett2020benign,derezinski2020exact,muthukumar2020harmless,hastie2022surprises,mei2022generalization}, particularly in high-dimensional regimes where both the feature dimension and the dataset size grow large, these analyses typically assume exactly known features.
Noise in the features is instead the subject of errors-in-variables models~\cite{carroll2006MeasurementErrorNonlinear,golub1980AnalysisTotalLeast,vanhuffel1991TotalLeastSquares}, where it is known to bias least-squares estimators by an amount that does not vanish as the dataset grows.
Much less is known rigorously when the feature noise has the structure typical of quantum models~\cite{hu2023tackling,hu2024GeneralizationErrorQuantum,khan2024practical,ahmed2025OptimalTrainingFinitely,wang2021UnderstandingPowerQuantum,hubregtsen2022TrainingQuantumEmbedding}.
We address this gap by initiating a systematic study of the effects of shot noise on quantum linear regression, focusing in particular on the QELM paradigm.
We provide two separate approaches to analytically control the rank discontinuity, isolate the singular part of the empirical pseudoinverse, and identify how it affects predictions.
Although the pseudoinverse contains terms that diverge as $N\to\infty$, we show that these singular contributions do not affect physical predictions when the training targets are known exactly.

We derive the scaling of the average mean-squared prediction error (MSE) on $N$ and $\ntr$, identifying two main regimes:
in the high statistics regime $N\gtrsim\ntr$, the MSE is dominated by stochastic fluctuations due to finite training statistics and scales as $(N\ntr)^{-1}$;
in the low statistics regime $N\lesssim\ntr$, these fluctuations become subleading and the MSE is dominated by a quantum-induced systematic bias, arising from the rank increase and associated ill-conditioning produced by finite-shot noise, and scales as $N^{-2}$.
The low-statistics regime is where the departure from regression with exactly known features is the sharpest: as the sampling noise sits in the features, it produces a bias that is insensitive to the size of the training set.
Consequently, at fixed $N$, increasing the number of training states eventually ceases to improve the prediction error: beyond $\ntr\sim N$, the MSE saturates at a floor set by the finite measurement statistics available for each training state.
Performance is thus governed by two resources rather than one, and increasing the size of the training dataset does not help unless the per-state measurement budget is increased as well.
This bias is the quantum counterpart of the attenuation bias of errors-in-variables regression~\cite{fuller1987MeasurementErrorModels,carroll2006MeasurementErrorNonlinear}, with two novel features: the noise strength is set by the per-state measurement budget, which is under the experimenter's control, and the noise also populates directions orthogonal to the span of the ideal features, which reshapes the bias through a Schur complement of the noise covariance.

We then proceed to also study the effect of noise on the training targets themselves, modelling scenarios in which preparation errors cause the training states to not be precisely characterized.
The singular directions, previously harmless, now couple to this noise and resurface in the prediction error.
While label noise amplification caused by minimum-norm least squares along poorly conditioned directions is familiar from the variance peak of double descent~\cite{belkin2019ReconcilingModernMachine,hastie2022surprises},
what is specific to the quantum setting is that these directions exist only because shot noise lifts the rank of the ideal feature matrix: the effect is a direct signature of the fact that quantum features are confined to a subspace fixed by the dimension of the Hilbert space, while their estimates are not. The phenomenon causes an unexpected interplay between training statistics $N$ and test statistics $M$ that contributes to the MSE a term scaling as $N/M$.
This means that increasing the training statistics without correspondingly increasing the statistics used at inference time can counterintuitively worsen performances, a behavior we identify as a form of \emph{quantum overfitting.} This phenomenon contrasts with the benign overfitting studied in quantum models in Ref.~\cite{peters2023GeneralizationOverfittingQuantum}, as in our case increasing the training statistics degrades test performances.

\Cref{fig:scheme} summarizes the resulting picture.
All the contributions to the prediction error are represented schematically, grouped according to the bias--variance splittings at training and at test time, and each is tagged with its physical origin: shot noise in the estimated features, or noise in the training targets.
Crucially, our analysis is not limited to the qualitative structure of this decomposition.
We determine explicitly the coefficients multiplying each scaling, in terms of the POVM, the training ensemble, and the target observable, and we find them to be in excellent agreement with numerical simulations across a broad set of regimes and model parameters.

Overall, by combining ideas from quantum machine learning and random matrix theory, our results indicate that familiar intuitions from classical learning must be revised to account for constraints imposed by the geometry of quantum state space.
More broadly, our work highlights that finite statistics can qualitatively alter conclusions drawn from idealized, infinite-sample analyses, motivating a more statistics-aware approach to quantum machine learning.

The remainder of the manuscript is organized as follows.
\Cref{sec:classical} reviews classical linear regression, and
\cref{sec:quantum_setup} introduces its quantum counterpart.
\Cref{sec:qlr_acuteness} identifies the rank-changing instability, and
\cref{sec:perturbation} gives a complementary structural analysis of its Schur-complement mechanism.
\Cref{sec:gaussian_approximation} recovers the resource laws in a homogeneous Gaussian reference model.
The main bias-variance result and its resource scaling are presented in \cref{sec:asymptotic_bias}, and finally \cref{sec:asymptotic_var} discusses the amplification mechanism that appears when the training targets are noisy, before we conclude with the resulting resource-allocation perspective in \cref{sec:conclusions}.

\section{Classical linear regression theory}%
\label{sec:classical}

To understand what changes when we account for the shot-noise in measurement probabilities, we first review the \textit{classical} linear regression setting where the input features are assumed to be known exactly, and the noise may only be in the targets. In this setting, noise affects performances in a relatively straightforward way, creating an irreducible error due to the fundamental unsolvability of the problem, and increasing the expected variance of the corresponding estimator. Notably, it does not enter its bias. We will see in \cref{sec:quantum_setup} how this picture changes significantly when we account for the structure of the experimental noise in both training and inference phases.

In the standard treatment of linear regression~\cite{hastie2009Elements,nakkiran2019MoreDataCan,hastie2022surprises}, the relation between an input $\bsx\in \mathbb{R}^{d}$ and a response $y\in \mathbb{R}$ is assumed to be linear: $y=\langle\bs w,\bsx\rangle+\eta$, where $\bs w\in\mathbb{R}^d$ is the regression vector we want to estimate, and $\eta$ is stochastic noise, typically assumed to be conditionally independent and
homoscedastic, i.e. $\EE[\eta|\bsx]=0$ and
$\EE[\eta^2|\bsx]=\gamma^2$.
Notably, this setting accounts for noise in $y$ but not in $\bsx$.
In the training phase, the input--output pairs $(\bsx,y)$ are drawn from a distribution $\calD$, and the training objective is to minimize the \textit{risk}, that is, some suitable function quantifying how well $\bs w$ reproduces the functional relation in the training data.
A standard choice is to quantify the quality of a candidate regressor $\hatbs w$ is via the population mean-squared error (MSE):
\begin{equation}%
\label{eq:risk_MSE0}
    \MSE(\hatbs w)=\E_{(\bsx,y)\sim \calD}\left[
    \left(
    y-
    \hatbs  w^T\bsx
    \right)^2
    \right].
\end{equation}
This decomposes as
\begin{equation}%
\label{eq:risk_MSE}
    \MSE({\hatbs{ w}})
    =
    \gamma^2+\|\bs{ w}-\hatbs{ w}\|^2_{\bs\Sigma},
\end{equation}
where $\bs\Sigma=\EE_{\calD}[\bs{x}\bs{x}^T]$ is the second moment matrix, and
$\|\bs{v}\|_{\bs\Sigma}^2\equiv \bs{v}^T\bs\Sigma\bs{v}$.
The $\gamma^2$ term in~\cref{eq:risk_MSE} is referred to as the \textit{irreducible noise}, and corresponds to white noise intrinsic to the data that cannot be reduced by any choice of $\hatbs w$.
On the other hand, $\|\bs w-\hatbs w\|_{\bs\Sigma}^2$ is the \textit{excess risk}, and represents error terms due to the training not producing the ideal regressor $\bs w$.

In practice, in the training phase one has only access to a finite training set $\calE_{\rm tr}\equiv\{(\bs x_i,y_i)\}_{i=1}^{\ntr}$, and the objective is to minimize the \textit{empirical} risk --- in this case the empirical MSE --- estimated on these examples, that is,
\begin{equation}
    R_{\rm tr}(\bs w)
    =
    \frac{1}{\ntr}
    \sum_{i=1}^{\ntr}
    \left(
    y_i-\bs w^T\bsx_i
    \right)^2 .
\end{equation}
This least-squares problem can be solved explicitly writing it as $R_{\rm tr}(\bs w)=\ntr^{-1}\|\bs Y-\bs w\bs X\|^2$, denoting with $\bs X\in\mathbb{R}^{d\times \ntr}$ the training matrix whose $i$-th column is $\bsx_i$, and with $\bs Y$ the row vector of training labels.
The $\bs w$ minimizing $R_{\rm tr}(\bs w)$ is then given by
\begin{equation}
\label{eq:pseudo_classical_estimator}
    \hatbs{ w}_{\rm tr}^T
    =
    \bs y^T \bs X^+
    =
    \bs w^T \bs{X}\bs{X}^+
    +
    \bs\eta^T \bs{X}^+,
\end{equation}
where $\bs{X}^+$ denotes the Moore-Penrose pseudoinverse of $\bs{X}$.

The MSE in \cref{eq:risk_MSE} depends on the training set
$\calE_{\rm tr}$ and is therefore random with respect to both $\bs X$ and $\bs\eta$ through its excess risk term.
Averaging this excess risk with respect to $\bs\eta$ we get a further decomposition:
\begin{equation}
\begin{gathered}
    \EE_{\bs\eta}
    \|\hatbs{ w}_{\rm tr}-\bs  w\|_{\bs \Sigma}^2
    =
    \bias^2(\bs X) + \Var(\bs X),
\end{gathered}
\end{equation}
where
$\Var(\bs X) \equiv
    \EE_{\bs\eta}
    \|\hatbs{ w}_{\rm tr}-\EE_{\bs \eta}[\hatbs{ w}_{\rm tr}]
    \|_{\bs\Sigma}^2$
measures the sensitivity of the estimator with respect to noise realizations, and
$\bias^2(\bs X)\equiv
    \|\bs{ w}-\EE_{\bs \eta}[\hatbs{ w}_{\rm tr}]
    \|_{\bs\Sigma}^2$
measures the average error for a given $\bs X$.
For the empirical least-squares estimator in~\cref{eq:pseudo_classical_estimator}, these equal
\begin{equation}
\begin{aligned}
    \bias^2(\bs{X})
    &=
    \|(\bs X\bs X^+-\bsI)\bs w\|_{\bs\Sigma}^2 ,
    \\
    \Var(\bs{X})
    &=
    \EE_{\bs\eta}\|(\bs{X}^+)^T \bs\eta\|_{\bs\Sigma}^2
    =
    \gamma^2
    \Trace[
    \bs{X}^+
    \bs\Sigma
    (\bs{X}^+)^T
    ].
\end{aligned}
\end{equation}
The bias is exactly zero in the \textit{underparametrized regime}, $\ntr\ge d$, provided $\rank(\bs X)=d$, since then
$\bs X\bs X^+=\bsI_d$.

This conclusion hinges on the exact knowledge of the features. When the inputs are themselves observed with noise, as in errors-in-variables models, the least-squares estimator acquires the so-called attenuation, or regression-dilution, bias, which does not vanish as $\ntr\to\infty$~\cite{fuller1987MeasurementErrorModels,carroll2006MeasurementErrorNonlinear}. In \cref{sec:asymptotic_bias} we show that the finite-statistics bias of quantum linear regression has the same structure, reshaped by the rank deficiency of the ideal feature matrix.

\section{Quantum linear regression
}%
\label{sec:quantum_setup}

We use the term \textit{quantum linear regression} to refer to linear regression tasks in which the input features $\bsx$ are estimates of probabilities obtained by measuring quantum states with finite resources.
The training data thus consist of pairs $(\rho,y)\sim\calD$, where each state $\rho$ is evolved under fixed quantum dynamics and subsequently measured.
We describe the overall measurement procedure by an effective POVM $\bs\mu=\{\mu_a\}_{a=1}^{\nout}$ over $\nout$ measurement outcomes; this can be understood as the Heisenberg-evolved measurement: if the input states evolve through a quantum map $\Phi$ and are then measured with a physical POVM $\bs\mu^{\rm phys}$, then $\mu_a\equiv \Phi^\dagger(\mu_a^{\rm phys})$~\cite{innocenti2023potential}, where $\Phi^\dagger$ denotes the adjoint of $\Phi$.
The probability vector corresponding to measuring an input state $\rho$ is then $\bsp(\rho)\equiv \langle\bs\mu^{\rm phys},\Phi(\rho)\rangle=\langle \bs\mu,\rho\rangle\in\RR^{\nout}$, where we use the notation $\langle A,B\rangle\equiv\trace(A^\dagger B )$.
The learning problem is to find a matrix $\bs W$ such that $\bs y\simeq \bs W\langle\bs\mu,\rho\rangle$.

We will focus on the cases where the target is a scalar and is achievable, meaning that it can be expressed as the expectation value of some observable in the real linear span of the effective POVM elements~\cite{innocenti2023potential}.
In these cases, following the conventions introduced in \cref{app:notation}, $\bs W$ becomes a row vector $\bsW_\calO$, and we write $y=\langle\calO,\rho\rangle$ for some observable $\calO$.
The condition $\langle\calO,\rho\rangle=\bsW_\calO\bsp(\rho)$ for all $\rho$ is then 
equivalent to $\calO=\bsW_\calO\bs\mu\equiv\sum_a w_a\mu_a$, with $w_a$ the elements of $\bsW_\calO$.

We first assume that the training targets $\bs y=\{y_i\}_{i=1}^{\ntr}$ are known exactly, so that any training error comes from the measured input features.
The consequences of also having noise in the targets will be discussed in~\cref{sec:asymptotic_var}.
By contrast, the input features are affected by statistical noise because the probabilities $\langle\bs\mu,\rho\rangle$ must be estimated on a quantum device using finitely many shots.
More precisely, we assume that each training state $\rho$ is measured $N$ times, and let $\hatbs n(\rho)=\{\hat n_a(\rho)\}_{a=1}^{\nout}$ be the count vector, whose $a$-th component records how many times outcome $a$ is observed.
The empirical probability vector is then
\begin{equation}
    \hatbs p_N(\rho)=\hatbs n(\rho)/N,
    \qquad
    \hatbs n(\rho)\sim \operatorname{Mult}(N,\bsp(\rho)),
\end{equation}
where $\operatorname{Mult}(N,\bsp(\rho))$ denotes the multinomial distribution
with $N$ trials and probability vector $\bsp(\rho)$. By construction,
\begin{equation}
\label{eq:distribution_p}
\begin{aligned}
    \EE[\hatbs p_N(\rho)]&=\bsp(\rho),\\
    \Cov[\hatbs p_N(\rho)]&=N^{-1}\bs\Sigma(\bsp(\rho)),
\end{aligned}
\end{equation}
where $\bs\Sigma(\bsp)\equiv\diag(\bsp)-\bsp\bsp^T$ is the categorical
covariance matrix. In what follows, we use $\hatbs p_\sigma$ as shorthand for the empirical frequency $\hatbs p_N(\sigma)$, leaving its dependence on the statistics implicit.
Similarly, we use $\bs p_\sigma\equiv \bs p(\sigma)$ for the exact probability vector and $\bs\Sigma_\sigma$ for its categorical covariance matrix.

To study how finite measurement resources affect training performances, we consider the \textit{finite-statistics MSE}, defined as the squared error for a fixed choice of measurement, test state, and target observable, averaged over the possible noise realizations at the inference stage with statistics $M$:
\begin{equation}
\label{eq:def_MSE_finite}
    \MSE(\hatbsW,M)
    =
    \EE_{\hatbs p_\sigma}
    \left[(\hatbsW\hatbs p_\sigma-\bsW_\calO\bsp_\sigma)^2\right],
\end{equation}
where the dependence on $M$ is hidden in the frequency vector $\hatbsp_\sigma\equiv \hatbsp_M(\sigma)$.
This MSE is defined with respect to an arbitrary estimator $\hatbsW$ and does not depend on the actual training method used to derive it; we will discuss later on how $\hatbsW$ is derived from the training dataset and how this choice affects the scaling of the MSE.
By contrast, $\bsW_\calO$ is the \textit{ideal} estimator, the one that would be obtained if measuring each training states with infinite precision; by construction, this satisfies $\langle\calO,\rho\rangle=\bsW_\calO\bsp(\rho)$ for every state $\rho$, hence $\hatbsW\hatbs p_\sigma-\bsW_\calO\bs p_\sigma$ is the estimation error for a single realization involving $M$ measures of the state $\sigma$ to estimate $\bs p_\sigma$, and using the estimator $\hatbsW$.
The finite-statistics MSE admits a bias-variance decomposition separating the errors due to an imprecise $\hatbsW$ from the errors due to fluctuations caused by the finite test budget:
\begin{equation}
    \MSE(\hatbsW,M)
    =
    \bias^2_{\rm test}(\hatbsW)
    +
    \Var_{\rm test}(\hatbsW,M),
\end{equation}
where the \emph{test bias} and \emph{test variance} are defined by
\begin{equation}
\label{eq:def_biasvartest}
\begin{aligned}
        \bias^2_{\rm test}(\hatbsW)
    &\equiv
    \left(
    \EE_{\hatbs p_\sigma}
    [\hatbsW\hatbs p_\sigma]
    -\bsW_\calO\bsp_\sigma
    \right)^2,
    \\
    \Var_{\rm test}(\hatbsW,M)
    &\equiv
    \EE_{\hatbs p_\sigma}
    \left[(\hatbsW\hatbs p_\sigma)^2\right]
    -
    (\hatbsW\bsp_\sigma)^2.
\end{aligned}
\end{equation}
Using \cref{eq:distribution_p}, these two quantities can be evaluated
explicitly:
\begin{equation}
\label{eq:test_bias_var}
\begin{aligned}
    \bias_{\rm test}^2(\hatbsW)
    &=
    [ (\hatbsW-\bsW_\calO)\bsp_\sigma ]^2,
    \\
    \Var_{\rm test}(\hatbsW,M)
    &=
    M^{-1}\hatbsW\bs\Sigma_\sigma\hatbsW^T.
\end{aligned}
\end{equation}
The interpretation of the two terms in \cref{eq:test_bias_var} is straightforward.
The test bias is the squared systematic prediction error due to the discrepancy between the ideal readout $\bsW_\calO$ and its estimate $\hatbsW$. The test variance quantifies fluctuations in the prediction caused by shot noise at test time.
Indeed, the test variance vanishes as $M\to\infty$, whereas the test bias is independent of $M$ and survives in this limit.
Thus, increasing the number of test shots cannot reduce the bias caused by the choice of $\hatbsW$.
We therefore refer to the test bias as the \emph{infinite-statistics MSE} and denote it by $\MSE(\hatbsW,\infty)$.
We obtain the following useful decomposition:
\begin{equation}
\label{eq:MSE_splitting_better}
    \MSE(\hatbsW,M)
    =
    \MSE(\hatbsW,\infty)
    +
    \frac{\hatbsW\bs\Sigma_\sigma\hatbsW^T}{M}.
\end{equation}
\Cref{eq:MSE_splitting_better} shows that we can take $M\to\infty$ to remove test sampling contributions and expose the error term $\MSE(\hatbsW,\infty)$, which is entirely caused by the training noise. Note that, when the readout is taken to be ideal, the infinite statistics MSE identically vanishes, and the coefficient
\begin{equation}
    \Var(\bsW_\calO)\equiv \bsW_\calO{\bs \Sigma}_\sigma\bsW_\calO^T
\end{equation}
acquires a special meaning: it is the single-shot variance of the estimator that would be obtained with exact knowledge of the dynamics, which typically coincides with the variance obtained via a shadow-tomography approach~\cite{huang2020PredictingManyProperties,innocenti2023potential,innocenti2023shadow}.
Going beyond the ideal case, we now specialize to the case where $\hatbsW$ is derived via pseudoinverse on the noisy probability matrix holding the measurement data, and then average that remaining error over the training noise realizations.

\begin{figure}[t]
    \centering
    \includegraphics[width=1.\linewidth]{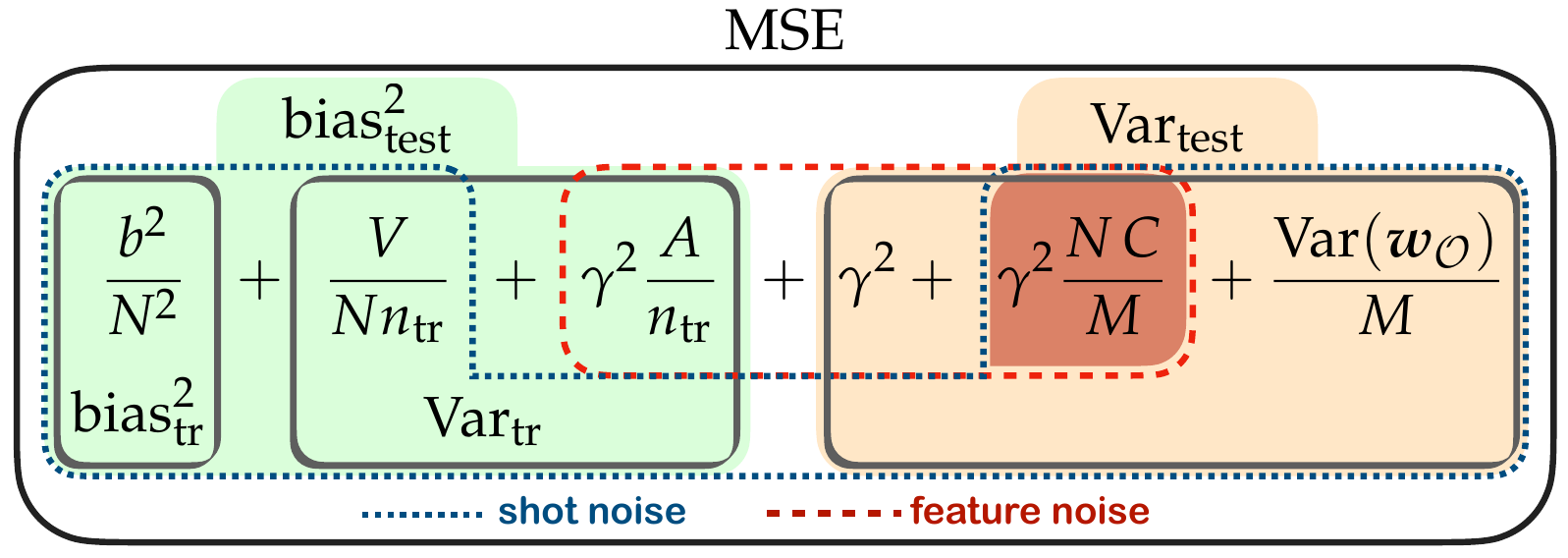}
    \caption{Decomposition of the MSE in terms of the different sources of noise. The effect of finite resource during the training is manifestly represented by the  training bias, while test sampling contributes to the test variance. Target noise induces corrections that also appear in the more familiar setting of classical linear regression. The only exception is the term $\gamma^2 NC/M$, which combines training and test shot noise with target noise and gives rise to a quantum overfitting phenomenon that amplifies the uncertainty introduced during the training stage. The coefficients $b$, $V$, $A$, $C$, and $\Var(\bsW_\calO)$ depend only on the dimension $d$ of the system.}
    \label{fig:scheme}
\end{figure}

Let $\bsP\in\RR^{\nout\times\ntr}$ denote the matrix of exact measurement
probabilities for the $\ntr$ training states, with $i$-th column
$\bsp(\rho_i)$. Its rank satisfies
\begin{equation}
    \operatorname{rank}(\bsP)\leq\min(d^2,\nout,\ntr),
\end{equation}
because the real vector space of Hermitian operators on a $d$-dimensional
Hilbert space has dimension $d^2$. If the training states span this operator
space and the effective POVM is informationally complete, then
$\operatorname{rank}(\bsP)=d^2$. We focus on the overcomplete regime
$\ntr,\nout>d^2$, under these spanning and informational-completeness
assumptions. In this regime, $\bsP$ is necessarily rank deficient. We show in
\cref{sec:qlr_acuteness,sec:perturbation} how this rank constraint controls the
stability of quantum linear regression.

In realistic settings, only finite measurement resources are available and one
must work with a noisy empirical version of $\bsP$. If $N$ measurements are
performed independently on each training state, we define the empirical
frequency matrix $\hatbs P_N$ by taking its $i$-th column as
\begin{equation}
    \hatbs p_i
    =
    \hatbs n^{(i)}/N,
    \qquad
    \hatbs n^{(i)}
    \sim
    \operatorname{Mult}(N,\bsp_i).
\end{equation}

Under the assumptions above, when $\ntr,\nout>d^2$, the regression problem is exactly solvable for any target observable $\calO$~\cite{innocenti2023potential}.
Let $\bs y=(y_1,\ldots,y_{\ntr})^T$ be the vector of training targets.
Among the generally nonunique exact solutions, which correspond to different choices of dual frame for the effective POVM~\cite{zhu2014QuantumStateEstimation,garciaperez2021LearningMeasureAdaptive,innocenti2023shadow}, the canonical minimum-norm solution is $\bsW_\calO=\bs y^T\bs P^+$.
Shot-noise fluctuations in $\hatbs P_N$ generally prevent the exact recovery of these ideal weights from finite data.
The standard unregularized pseudoinverse training rule instead yields the estimator
\begin{equation}%
\label{eq:W}
    \hatbsW_N
    =
    \bs y^T\hatbs P_N^+.
\end{equation}
The estimator $\hatbsW_N$ is a random variable that depends nonlinearly on the particular realization of the training shot noise through the Moore--Penrose pseudoinverse $\hatbsP_N^+$.

To characterize the effect of shot noise during the training stage, we average the infinite-test-statistics MSE over realizations of the training noise, and define
\begin{equation}
    \MSE(N,\infty)
    \equiv
    \EE_{\hatbs P_N}
    \left[\MSE(\hatbsW_N,\infty)\right].
\end{equation}
Here and below, we suppress the dependence on the exact design matrix $\bs P$
and the fixed POVM $\bs\mu$. As before, the averaged MSE can be decomposed into
a \emph{training bias} and a \emph{training variance}:
\begin{equation}
\label{eq:bias_var_decomposition_train}
\begin{aligned}
    &\MSE(N,\infty)
    =
    \bias^2(N)
    +
    \Var_{\rm tr}(N),
    \\
    &\bias^2(N)
    \equiv
    \left[
    \left(\bsW_\calO-\EE_{\hatbs P_N}[\hatbsW_N]\right)
    \bsp_\sigma
    \right]^2,
    \\
    &\Var_{\rm tr}(N)
    \equiv
    \EE_{\hatbs P_N}
    \left[
    \left(
    (\hatbsW_N-\EE_{\hatbs P_N}[\hatbsW_N])
    \bsp_\sigma
    \right)^2
    \right].
\end{aligned}
\end{equation}
The training variance $\Var_{\rm tr}(N)$ quantifies the fluctuations of the prediction around its mean over realizations of the training noise.
The training bias $\bias^2(N)$ is the squared difference between the target prediction and this mean prediction. Its source is explicit in
\cref{eq:bias_var_decomposition_train}: in general,
$\EE_{\hatbs P_N}[\hatbsW_N]\neq\bsW_\calO$ because the pseudoinverse depends nonlinearly on the empirical frequency matrix.
The same nonlinearity makes the dependence of the bias and variance on both $N$ and the number of training states $\ntr$ highly nontrivial.

Our main objective is to show that, despite the nonlinear dependence of the MSE on the training noise, its dependence on the resources used in the protocol can be derived explicitly. These resources are the number of training shots per state $N$, the number of training states $\ntr$, and the number of test shots $M$. Anticipating the results of the following sections, we find that the scaling of the MSE with $N$, $\ntr$, and $M$ can be characterized compactly as the sum of six contributions. This decomposition is shown schematically in \cref{fig:scheme}. Each term originates from one of the three sources of noise discussed in this paper: training shot noise, test shot noise, and label noise. Label noise is the type usually considered in classical linear regression. It can be thought of as the error arising from incomplete knowledge of the states, for instance due to noisy state preparation on a quantum device. All the constants appearing in \cref{fig:scheme} depend only on the dimension of the input system, and are studied explicitly in \cref{sec:asymptotic_bias,sec:asymptotic_var}. Despite its apparent simplicity, this compact expansion captures all the interesting regimes of QELM, as supported by the numerical results shown in \cref{fig:N_scaling,fig:ntrain_saturation,fig:bias_var_scaling,fig:label_noise} and in \cref{app:numerics}.

\section{Quantum linear regression is ill-conditioned}%
\label{sec:qlr_acuteness}

The MSE decomposition in \cref{sec:quantum_setup}, and in particular~\cref{eq:W,eq:bias_var_decomposition_train}, reduced the problem to characterizing the behaviour of the learned readout $\bs y^T\hatbsP_N^+$.
The main obstruction to do this is determining how errors in $\hatbsP_N$ propagate to its pseudoinverse.
Indeed, for an informationally complete POVM and spanning training states with $\nout,\ntr>d^2$, shot noise can increase the rank of $\bsP$,
and the Moore-Penrose pseudoinverse is discontinuous under perturbations that change the rank.
More precisely, if $\rank(\bs P+\delta\bs P)\neq\rank(\bs P)$, then $\lim_{\delta\bs P\to 0}(\bs P+\delta\bs P)^+$ does not generally converge to $\bs P^+$~\cite{wedin1973perturbation,stewart1977perturbation,stewart1990Matrix,gonzalez2015laurent}. 

The notion of \textit{acute perturbations} in matrix perturbation theory is key to understanding the behaviour of the Moore--Penrose pseudoinverse under noise.
When the rank increases, new singular values emerge at the noise scale, and their inversion produces contributions that diverge as the noise vanishes.
\\
More formally, let $\bs A\in\CC^{n\times m}$ and
$\hatbs A_{\tau}=\bs A+\tau\hatbs E$, where $\tau>0$ is a generic
perturbation parameter.
We say that $\tau\hatbs E$ is an \textit{acute perturbation} for $\bs A$ if
\begin{equation*}
    \|\hatbs A_{\tau}\hatbs A_{\tau}^+-\bs A\bs A^+\|_{\rm op} < 1
    \quad\text{and}\quad
    \|\hatbs A_{\tau}^+\hatbs A_{\tau}-\bs A^+\bs A\|_{\rm op} < 1.
\end{equation*}
Here $\bs A\bs A^+$  is the projector onto the image of $\bs A$, while $\bs A^+\bs A$ is the projector onto the row space of $\bs A$.

The acuteness of a perturbation tells us whether the inverse of a matrix is stable under perturbation.
In particular, the instability for non-acute perturbations can be quantified as follows:
\begin{theorem}[\cite{stewart1977perturbation}]
\label{thm:moore_perturbation}
    If $\tau\hatbs E$ is \textit{not} an acute perturbation for $\bs A$,
    \begin{equation}
        \|\hatbs A^+_\tau-\bs A^+\|_{\rm op}\geq \tau^{-1}\|\hatbs E\|^{-1}_{\rm op}.
    \end{equation}
    If, further, ${\rm rank}(\hatbs A_\tau)\geq {\rm rank}(\bs A)$, then
    \begin{equation}
        \|\hatbs A^+_\tau\|_{\rm op}\geq \tau^{-1}\|\hatbs E\|^{-1}_{\rm op}.
    \end{equation}
\end{theorem}
For a fixed rank-increasing direction $\hatbs E$ for which nonacuteness
persists as $\tau\to0$, these bounds exhibit the resulting
$\tau^{-1}$ divergence.  The following rank criterion is also useful:
\begin{theorem}[\cite{wedin1973perturbation}]
    If two matrices $\bs A,\bs B\in\CC^{n\times m}$ are acute with respect
    to one another, then ${\rm rank}(\bs A)={\rm rank}(\bs B)$. If
    ${\rm rank}(\bs A)={\rm rank}(\bs B)$ and
    $\|\bs B-\bs A\|_{\rm op}\leq \|\bs A^+\|_{\rm op}^{-1}$, then they are
    acute with respect to one another.
\end{theorem}

For instance, when the ideal matrix is square and invertible, as are the calibration matrices inverted in readout-error mitigation~\cite{bravyi2021MitigatingMeasurementErrors,maciejewski2020MitigationReadoutNoise}, the second statement guarantees that sufficiently small perturbations are acute, so that finite calibration statistics only induce a regular error in the inverse.
By contrast, in the overcomplete regime considered here, shot noise generically increases the rank:
\begin{equation}
    \rank(\bs P)=d^2< {\rm rank}(\hatbsP_N)=\min(\nout,\ntr).
\end{equation}
On this full-rank event the perturbation is nonacute, and writing the perturbation as $\hatbsP_N=\bsP + N^{-1/2}\hatbsXi$, we have
\begin{equation}%
\label{eq:bound_N}
    \|\hatbs P_N^+\|_{\rm op}\geq
    \sqrt{N}
    \|\hatbsXi\|^{-1}_{\rm op}.
\end{equation}
This growth is intrinsic to the unregularized pseudoinverse estimator in the overcomplete setting.
The same mechanism affects quantum fidelity kernels: the ideal Gram matrix $K_{ij}=\trace(\rho_i\rho_j)$ of $\ntr$ training states has rank at most $d^2$, whereas its finite-shot estimate is generically full rank. Finite-shot effects on estimated kernel matrices have been discussed in~\cite{wang2021UnderstandingPowerQuantum,hubregtsen2022TrainingQuantumEmbedding}.

\Cref{eq:bound_N} concerns the norm of the full pseudoinverse and does not by
itself imply that physical predictions diverge.
We will show in \cref{sec:perturbation} that its singular block is annihilated by left and right multiplication by $\bs y^T$ and $\bs p_\sigma$.
Nevertheless, this does not mean that the singularity has no effect at all: it becomes relevant when the target $\bs y$ is noisy, as we will discuss in~\cref{sec:asymptotic_var}.

\section{Perturbative analysis of noise perturbations}%
\label{sec:perturbation}

As discussed in \cref{sec:qlr_acuteness}, in the regime $\ntr,\nout>d^2$, the pseudoinverse $\hatbsP_N^+$ diverges in the same $N\to\infty$ limit where the measured probabilities become exact, due to the rank increase caused by shot noise.
However, as per \cref{eq:W}, the learned readout is the product $\hatbsW_N=\bs y^T\hatbsP_N^+$, which leaves some hope that the diverges of $\hatbsP_N^+$ do not translate into divergences of $\hatbsW_N$. Furthermore,  the infinite-test-statistics MSE is $\hatbsW_N\bs p_\sigma=\bs y^T\hatbsP_N^+ \bs p_\sigma$, which might further regularize the singular behaviour of $\hatbsP_N^+$.
We will show here that this is precisely what happens by developing, in the $\ntr>\nout>d^2$ regime, for any fixed noise realization, a tailored perturbative approach for the pseudoinverse that isolates its divergences in specific matrix blocks.
This will be used to derive the MSE scaling results in \cref{sec:gaussian_approximation}.

We decompose the noisy probability matrix as
$\hatbs P_N=\bsP+\hatbsXi/\sqrt{N}$, where the columns $\hatbs\xi_i$ represent the centered and rescaled multinomial fluctuations. Each column is independent but not identically distributed, as $\Cov(\hatbs \xi_i)=\bs\Sigma_i$.
The $N^{-1/2}$ rescaling makes $\hatbsXi$ only depend on $N$ via its third and higher moments, and we thus suppress its explicit $N$ dependence for sake of notation here.
We also fix the regime $\ntr>\nout$.
With high probability, $\hatbsP_N$ is surjective, and its pseudoinverse can be written as $\hatbsP_N^+=\hatbsP_N^T(\hatbsP_N\hatbsP_N^T)^{-1}$.
The discontinuity of the pseudoinverse therefore arises from the inverse of the $\nout\times\nout$ matrix
\begin{equation}
    \hatbsP_N\hatbsP_N^T=\bs{P}\bs{P}^T+
    \frac{1}{\sqrt N}\left(
    \bs{P}
    \hsXi^T+
    \hsXi
    \bs{P}^T
    \right)
    +\frac{1}{N}\hsXi\hsXi^T.
\end{equation}
To distinguish directions carrying ideal probabilities data from directions populated only by noise, we use the row and column spaces of $\bsP$. Rank deficiency then causes nontrivial complementary kernel and cokernel subspaces.
Choose isometries $V_1,V_2$ in the training space and $U_1,U_2$ in the outcome space such that the corresponding orthogonal projectors are
\begin{equation}
    \begin{aligned}
        \Pi_1
        &\equiv \Pi_{\supp(\bs{P})}
        =\bs{P}^+\bs{P}=V_1V_1^T,
        \\
        \Pi_{\bar 1}
        &\equiv\Pi_{\im(\bs{P})}
        =\bs{P}\bs{P}^+=U_1U_1^T.
    \end{aligned}
\end{equation}
Write $\Pi_2\equiv\bsI-\Pi_1=V_2V_2^T$ for the projector onto
$\ker(\bsP)$ and
$\Pi_{\bar 2}\equiv\bsI-\Pi_{\bar 1}=U_2U_2^T$ for the projector onto
$\coker(\bsP)$.
By suitable changes of basis in the input and output spaces, we write $\hatbsP_N$ in block form as
\begin{equation}%
\label{eq:PN_expansion}
    \hatbsP_N=
    \begin{pmatrix}
        \bs{P}_{ 11}+N^{-1/2}\hsXi_{11} & N^{-1/2}\hsXi_{12}\\
        N^{-1/2}\hsXi_{ 2 1 } & N^{-1/2}\hsXi_{ 2 2}
    \end{pmatrix}\,,
\end{equation}
where
$\hsXi_{mn}\equiv U_m^T\hsXi V_n$ and
$\bsP_{11}\equiv U_1^T\bsP V_1$.
By definition, $\bsP=\Pi_{\bar1}\bsP\Pi_1$ and
$\bsP_{11}\in\RR^{d^2\times d^2}$ is invertible.

Using~\cref{eq:PN_expansion}, one can find a fixed-realization asymptotic
expansion for $\hatbsP_N^+$. We show in \cref{app:perturbative_P} that
$\Delta\hatbsP_N^+\equiv\hatbsP_N^+-\bsP^+$ has leading terms
\begin{equation}%
\label{eq:Delt_PN_small}
\begin{aligned}
    (\Delta\hatbs{P}_N^+)_{11}
    &=-N^{-1/2}\bs{P}_{11}^+(\hsXi/\hsXi_{22})\bs{P}_{11}^+
      +O(N^{-1}),\\
    (\Delta\hatbs{P}_N^+)_{12}
    &=-\bs{P}_{11}^+\hsXi_{12}\hsXi_{22}^+
      +O(N^{-1/2}),\\
    (\Delta\hatbs{P}_N^+)_{21}
    &=-\hsXi^+_{22}\hsXi_{21}\bs{P}_{11}^+
      +O(N^{-1/2}),\\
    (\Delta\hatbs{P}_N^+)_{22}
    &=\sqrt N\,\hsXi_{22}^++O(1).
\end{aligned}
\end{equation}
where we introduced the \emph{generalized Schur complement}
\begin{equation}
    \hsXi/\hsXi_{22}=\hsXi_{11}-\hsXi_{12}\hsXi_{22}^+\hsXi_{21}\,.
\end{equation}
It is evident from~\cref{eq:Delt_PN_small} that
$\Delta\hatbs{P}_N^+$ does not converge to zero as $N\rightarrow\infty$:
the off-diagonal blocks approach finite stochastic limits, whereas the $(2,2)$
block diverges as $\sqrt N$.
Notably, the discontinuity is primarily controlled by $\hsXi_{22}^+$.

Using the pseudoinverse estimator $\hatbsW_N$ as given in~\cref{eq:W}, the MSE for a given training noise realization and test state $\sigma$, at infinite test statistics, is
\begin{equation}%
\label{eq:MSEWinfty_2}
    (\Delta\hat o)^2\equiv\|(\hatbsW_N-\bsW_\calO)\bsp_\sigma\|^2
    =
    [\bs y^T \Delta\hatbs{P}_N^+ \bsp_\sigma]^2.
\end{equation}
By construction, 
$\bsp_\sigma\equiv\langle\bs\mu,\sigma\rangle$ satisfies $\Pi_{\bar2}\bsp_\sigma=0$, and
$\bs y^T\equiv \langle\calO,\bs\rho\rangle$ satisfies $\bs y^T\Pi_2=0$.
In words, these amount to saying that $\bsp_\sigma$ is in the column space of $\bs P$ while $\bs y^T$ is in its row space.
Thus the divergent $(2,2)$ block in~\cref{eq:Delt_PN_small} does not enter $\MSE(\hatbsW_N, \infty)$.
Expanding~\cref{eq:MSEWinfty_2} via~\cref{eq:Delt_PN_small} gives
\begin{equation}%
\label{eq:delta_P_11}
    \Delta\hat o
    =
    \bsw^T_{\calO} U_1
    \left(
    \frac{\bs{A}_1}{\sqrt N}
    +
    \frac{\bs{A}_2}{N}
    \right)\!
    V_1^T\langle\tilde{\bs\rho},\sigma\rangle
    +
    O\Big( \frac{1}{N^{3/2}} \Big),
\end{equation}
where we kept the first two nontrivial orders in $N$ of $(\Delta\hatbs{P}_N^+)_{11}$, and used $\bs w_{\calO}\equiv\langle\tilde{\bs\mu},\calO\rangle$.
Here
$\tilde{\bs\mu}$ and $\tilde{\bs\rho}$ are, respectively, the canonical dual measurement and training frames defined in \cref{app:frame}.
The operators $\bsA_1,\bsA_2$ are given by
\begin{equation}
\label{eq:A1_A2_definitions}
\begin{gathered}
    \bs{A}_1
    =
    -\hsXi/\hsXi_{22}\,,\\
    \bsA_2
    =
    (\hatbsXi/\hatbsXi_{22})\bsP_{11}^{+}(\hatbsXi/\hatbsXi_{22})
    -
    \hatbsXi_{12} \hatbsQ \hatbsXi_{12}^T
    \bsP_{11}^{+T}\hatbsZ,
\end{gathered}
\end{equation}
where $\hatbsQ\equiv\bsI-\hsXi_{22}^+\hsXi_{22}$ is the projector onto $\ker(\hsXi_{22})$, and $\hatbsZ\equiv\bsI+\hatbsXi_{21}^T(\hatbsXi_{22}\hatbsXi_{22}^T)^{-1}\hatbsXi_{21}$.
When $\ntr=\nout$, the block $\hatbsXi_{22}$ is square; if it is
nonsingular, then $\hatbsQ=0$.  Thus $\hatbs Q$ captures the additional
kernel present in the rectangular case.
Here $O(N^{-3/2})$ denotes the next term in the formal expansion at a
fixed noise realization for which the required block inverses exist: it is
not a uniform probabilistic remainder bound.  In particular, the nonlinear
random coefficients $\hatbsXi_{22}^+$, $\hatbs Q$, and $\hatbs Z$ cannot be
averaged term by term without additional conditioning estimates.

\Cref{eq:delta_P_11} exposes the nontrivial consequences of computing the pseudoinverse in a non-acute perturbation regime: exact targets and test probabilities protect the prediction from the divergent contribution, without eliminating the finite-shot noise.
However, the resulting random coefficients $\bs A_1, \bs A_2$ are nonlinear functions of $\hatbsXi$ and depend in particular on the diverging term $\hatbsXi_{22}^+$, and therefore cannot be averaged term by term without additional control.
We will show how to handle this issue first in \cref{sec:gaussian_approximation}, in a homogeneous centered Gaussian noise model, and then present in \cref{sec:asymptotic_bias} an alternative proof strategy that directly operates on the multinomial noise.

\section{Homogeneous Gaussian approximation}
\label{sec:gaussian_approximation}

The perturbative expression in \cref{eq:delta_P_11} contains, hidden in the $\bs A_1$ and $\bs A_2$ operators, a generalized Schur complement and the projector $\hatbs Q$, both of which are nonlinear
functions of the shot-noise matrix.
We now show that these quantities can be averaged explicitly in a homogeneous reference model, meaning that we assume the shot noise distribution to be Gaussian and state-independent. The purpose of this construction is to expose the origin of the $N^{-2}$ squared bias and the $(N\ntr)^{-1}$ training variance in a setting where the relevant averages can be computed analytically.
Interestingly, although there is, a priori, little reason to expect a noise model that neglects the state dependence of the covariance matrices to reproduce the multinomial case accurately, the resulting MSE agrees remarkably well with both the multinomial prediction and our numerical observations.
In \cref{app:numerics} we further present several simulations clearly showing that this assumption reproduces extremely well the MSE results, even in the regime with low $N$.
The derivation operating directly on the multinomial noise requires different techniques and will be discussed separately in \cref{sec:asymptotic_bias}.

In the homogeneous Gaussian model, we assume that $\hatbsXi$ follows a matrix-normal distribution~\cite{gupta2018matrix}:
\begin{equation}
    \hatbsXi
    \sim\calM\calN(0,\barSigma,\bs I_{\ntr}),
    \qquad
    \barSigma
    \equiv\frac{1}{\ntr}\sum_{i=1}^{\ntr}\bs\Sigma^{(i)}.
    \label{eq:gaussian_reference_model_main}
\end{equation}
Thus the columns of $\hatbsXi$ are independent centered Gaussian vectors, all with the same covariance matrix $\overline{\bs\Sigma}$.
Using the same block coordinates of \cref{eq:PN_expansion}, we write $\barSigma_{ab}\equiv U_a^T\barSigma U_b$, and assume
$\barSigma_{22}$ is positive definite.  Its Schur complement is
\begin{equation}
    \barSigma_{\rm eff}
    \equiv
    \barSigma_{11}
    -\barSigma_{12}\barSigma_{22}^{-1}\barSigma_{21}.
    \label{eq:gaussian_effective_covariance_main}
\end{equation}

Fix $d$ and $\nout$ and let $\ntr\to\infty$. Since $\ntr\ge\nout$, $\hatbsXi_{22}$ has full row rank almost surely, so
$\hatbs Q$ is the orthogonal projector onto $\ker(\hatbsXi_{22})$ and has rank $\ntr-\nout$. 
Conditionally on a realization of $\hatbsXi_{21}$ and $\hatbsXi_{22}$, the generalized Schur complement $\hatbsXi/\hatbsXi_{22}$ becomes linear in the remaining noise matrix terms, and thus also follows a matrix-normal distribution; this conditional distribution will then depend on $\hatbs Z$:
\begin{equation}
    \hatbsXi/\hatbsXi_{22}
    \mathrel{\big|}(\hatbsXi_{21},\hatbsXi_{22})
    \sim
    \calM\calN(0,\barSigma_{\rm eff},\hatbsZ).
    \label{eq:gaussian_schur_conditional_main}
\end{equation}
Under the same conditioning,
$\hatbsXi_{12}\hatbs Q\hatbsXi_{12}^T$ is a random Wishart matrix with scale
$\barSigma_{\rm eff}$ and $\ntr-\nout$ degrees of freedom, and it is
independent of the generalized Schur complement. This means that we can perform all the averages in the following way: first we average any polynomial in the Schur complement and $\hatbsXi_{12}\hatbs Q\hatbsXi_{12}^T$, using standard formulas for Gaussian and Wishart distributions \cite{muirhead2009aspects}; then, we average the resulting moments with respect to $\hatbs Z$.
In fact, let us remember that $\hatbs Z = \bs I+\hatbsXi^T_{21}(\hatbsXi_{22}\hatbsXi_{22}^T)^{-1}\hatbsXi_{21}$ and that $\hatbsXi_{21},\,\hatbsXi_{22}$ are independent. Moreover, $(\hatbsXi_{22}\hatbsXi_{22}^T)^{-1}$ is known to be distributed as an inverse Wishart matrix and moments are explicitly known. For instance, it can be seen that, for $\ntr>\nout+1$,
\begin{equation}
    \EE[\hatbsZ]
    =\frac{\ntr-d^2-1}{\ntr-\nout-1}\,\bs I.
    \label{eq:gaussian_Y_mean_main}
\end{equation}

The conditional Gaussian calculation in \cref{lem:gaussian_A1_A2_moments} allows us to compute the first and second moments of the terms depending on $\bs A_1$ and $\bs A_2$ in \cref{eq:delta_P_11}, at least when $\ntr > \nout + 3$. The terms most relevant to our discussion are:
\begin{align}
    \EE\!\left[
        \bs w_{\calO}^TU_1\bs A_2
        V_1^T\langle\tilde{\bs\rho},\sigma\rangle
    \right]
    &=
    -\left(1-\frac{d^2+1}{\ntr}\right)b_{\rm hom},\nonumber\\
    \Var\!\left[
        \bs w_{\calO}^TU_1\bs A_1
        V_1^T\langle\tilde{\bs\rho},\sigma\rangle
    \right]
    &=
    \frac{V_{\rm hom}}{\ntr}
    \frac{\ntr-d^2-1}{\ntr-\nout-1}.
    \label{eq:gaussian_A1_A2_main}
\end{align}
where we defined the coefficients:
\begin{equation}
\begin{aligned}
    b_{\rm hom}
    &\equiv
    \ntr\,\bs w_{\calO}^TU_1\barSigma_{\rm eff}
    \bsP_{11}^{-T}V_1^T
    \langle\tilde{\bs\rho},\sigma\rangle,\\
    V_{\rm hom}
    &\equiv
    \ntr\,
    \|\langle\tilde{\bs\rho},\sigma\rangle\|_2^2
    \bs w_{\calO}^TU_1\barSigma_{\rm eff}U_1^T
    \bs w_{\calO}.
\end{aligned}
\label{eq:gaussian_b_V_main}
\end{equation}
We recall that $\tilde{\bs\rho}$ is the dual training frame defined in \cref{app:frame}. The subscript \emph{hom} in \cref{eq:gaussian_b_V_main} distinguishes these two coefficients from their counterparts in the full multinomial model discussed in \cref{sec:asymptotic_bias}. In \cref{subsec:gaussian_multinomial_comparison}, we show that the coefficient $b_{\rm hom}$ obtained from the homogeneous Gaussian model in fact coincides with its multinomial counterpart, defined in \cref{eq:multinomial_b_V_main}. The variance contribution, by contrast, differs from the multinomial one in \cref{eq:multinomial_b_V_main}; the two coincide only when the homogeneity condition is imposed on the training covariance matrices $\bs \Sigma_i$.

The moments help distinguish the different roles of the two perturbative orders in \cref{eq:delta_P_11}: $\bs A_1/\sqrt N$ is centered and supplies the leading training variance, while $\bs A_2/N$ contributes the bias and at subleading order the variance.
Defining the second-order truncation
\begin{equation}
    \Delta\hat o^{(2)}
    \equiv
    \bsw^T_{\calO}U_1
    \left(
        \frac{\bs A_1}{\sqrt N}
        +\frac{\bs A_2}{N}
    \right)
    V_1^T\langle\tilde{\bs\rho},\sigma\rangle,
    \label{eq:gaussian_second_order_truncation}
\end{equation}
its expectation value is exactly $-(1-(d^2+1)/\ntr)b/N$, and thus for large $N,\ntr$ at fixed $d$, we have
\begin{align}
    \EE[\Delta\hat o^{(2)}]
    &
    =-\frac{b_{\rm hom}}{N}+O\!\left(\frac{1}{N\ntr}\right),\nonumber\\
    \Var(\Delta\hat o^{(2)})
    &=
    \frac{V_{\rm hom}}{N\ntr}
    +O\!\left(
        \frac{1}{N\ntr^2}
        +\frac{1}{N^2\ntr}
    \right).
    \label{eq:gaussian_mean_variance_main}
\end{align}
Consequently,
\begin{equation}
    \EE[(\Delta\hat o^{(2)})^2]
    =\frac{b_{\rm hom}^2}{N^2}
    +\frac{V_{\rm hom}}{N\ntr}
    +O\!\left(
        \frac{1}{N^2\ntr}
        +\frac{1}{N\ntr^2}
    \right).
\label{eq:gaussian_MSE_main}
\end{equation}
Thus the centered first-order fluctuation produces the $(N\ntr)^{-1}$ training variance, whereas the second-order mean produces
the $N^{-2}$ squared bias. 

The above analysis relies on two main assumptions: first, that the Gaussian homogeneous model provides a good approximation, and second, that the remainder of the Moore–Penrose expansion is well behaved. Although the numerical evidence presented in \cref{app:numerics} suggests that both assumptions hold, proving them rigorously is not straightforward. For instance, attempting to average higher-order terms in the Moore–Penrose expansion reveals that stronger conditions on $\ntr$ are required to ensure the existence of the relevant moments of the Wishart distributions. Moreover, finiteness of every coefficient up to a fixed order does not by
itself justify averaging the entire perturbative series or its truncated
remainder.  The fixed-realization expansion has a random radius of
convergence controlled by the smallest relevant singular values.  Rare
realizations in which $\hatbsXi_{22}$ or the effective Schur complement is
nearly singular can therefore prevent the interchange of expectation with
an infinite series and can invalidate a remainder estimate that is not
uniform in the noise realization. This is why the Gaussian calculation
provides exact moments for fixed coefficients but is not, by itself, a
proof of an all-orders MSE expansion.

For these reasons, we rely on the more rigorous analysis of \cref{sec:asymptotic_bias} to circumvent these issues. Nevertheless, the Gaussian model and the Moore–Penrose expansion remain useful tools for obtaining accurate estimates efficiently.

\begin{figure*}[t]
    \centering
    \begin{minipage}[t]{0.485\textwidth}
        \centering
        \includegraphics[width=\linewidth]{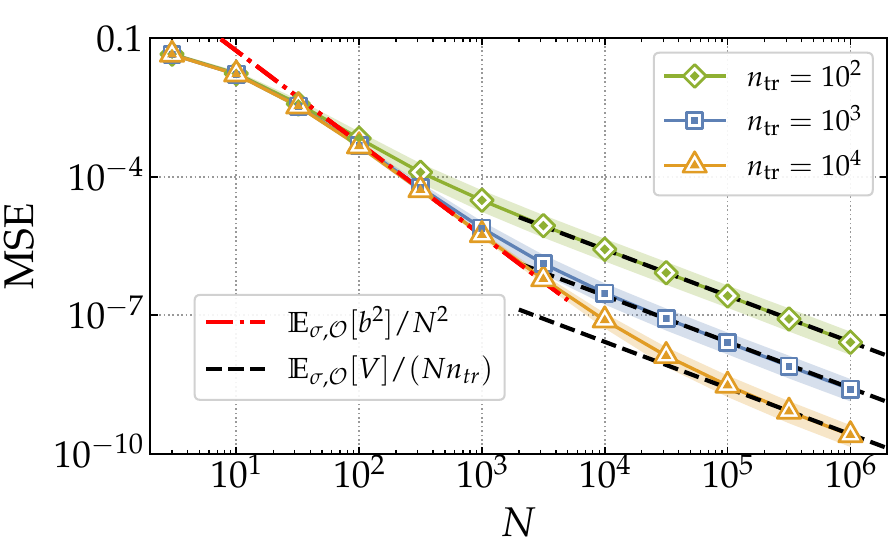}
        \caption{MSE as a function of $N$ in the limit $M\to\infty$ (noiseless test features), for a single-qubit MUB measurement (thus $d=2$, $n_{\rm out}=d(d+1)=6$), for $n_{\rm tr}=10^2$ (green), $10^3$ (blue), and $10^4$ (orange). Markers and solid lines show the median MSE; shaded bands indicate the 10th--90th percentile range. The red dash-dotted line indicates the small-$N$ bias contribution $\mathbb{E}_{\sigma,\mathcal O}[b^2]/N^2$, independent of $n_{\rm tr}$, and black dashed lines indicate the large-$N$ variance contribution $\mathbb{E}_{\sigma,\mathcal O}[V]/(N n_{\rm tr})$, where $\mathbb{E}_{\sigma,\mathcal O}[b^2]=(d-1)(d+2)^2/(d+1)$ and $\mathbb{E}_{\sigma,\mathcal O}[V]=d(d-1)(d+2)/(d+1)$. Other parameters are $n_{\rm test}=100$ and $n_{\rm obs}=50$ target observables, with $n_{\rm reps}=50$ independent draws of the training, test, and target ensembles and $n_{\rm real}=100$ independent shot-noise realizations per draw.}
        \label{fig:N_scaling}
    \end{minipage}
    \hfill
    \begin{minipage}[t]{0.485\textwidth}
        \centering
        \includegraphics[width=\linewidth]{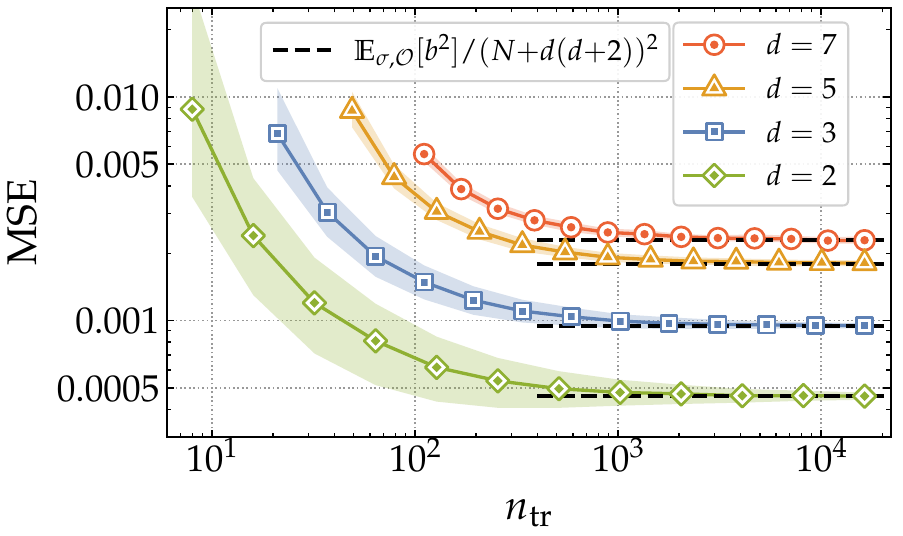}
        \caption{MSE as a function of the training-set size $n_{\rm tr}$ at fixed $N=10^2$ and in the limit $M\to\infty$ (noiseless test features), for MUB measurements in dimension $d=2$ (green), $3$ (blue), $5$ (orange), and $7$ (red), with $n_{\rm out}=d(d+1)$. Markers and solid lines show the median MSE; shaded bands indicate the 10th--90th percentile range. Black dashed lines indicate the large-$n_{\rm tr}$ floor $\mathbb{E}_{\sigma,\mathcal O}[b^2]/(N+d(d+2))^2$, where $\mathbb{E}_{\sigma,\mathcal O}[b^2]=(d-1)(d+2)^2/(d+1)$, set by the finite-$N$ bias, which additional training states cannot remove. Other parameters are $n_{\rm test}=100$ and $n_{\rm obs}=50$ target observables, with $n_{\rm reps}=50$ independent draws of the training, test, and target ensembles and $n_{\rm real}=100$ independent shot-noise realizations per draw.}
        \label{fig:ntrain_saturation}
    \end{minipage}
\end{figure*}

\section{Bias--variance decomposition and resource scaling}
\label{sec:asymptotic_bias}

The Gaussian calculation in \cref{sec:gaussian_approximation} establishes a nontrivial competition between bias and variance contribution to the infinite-test-statistics MSE, under the homogeneous Gaussian noise model, for large $N$.
We now complement this analysis with an approach that works directly with the multinomial noise and that, although asymptotic in $\ntr$, remains valid at small $N$.
This method relies on an exact decomposition of the infinite-test-statistics estimation error $\Delta\hat o=\bs y^T(\hatbsP_N^+-\bsP^+)\bsp_\sigma$ into a deterministic component, a centered fluctuation, and a remainder that collects all nonlinear contributions.
We then show that, at fixed $N$, $\nout$, and $d$, the remainder is subleading with high probability for large $\ntr$.
The analysis holds for any fixed choice of training states, target observable, and test state and, as in the previous sections, assumes exact targets and infinite test statistics.
We further assume that the nonzero eigenvalues of $\ntr^{-1}\bsP\bsP^T$ and of the noise-covariance block $\overline{\bs\Sigma}_{22}$ are bounded below by positive constants $\kappa_{11}$ and $\kappa_{22}$, respectively; see~\cref{eq:direct_uniform_conditioning}.
The full proof is given in~\cref{app:direct_remainder,app:asymptotic}.

\parTitle{Exact decomposition of infinite-test-statistics error}
Specifically, we prove in~\cref{app:direct_remainder} that
\begin{equation}
\begin{gathered}
    \Delta\hat o
    =
    -b_N/N
    +\hat L+\hat R,
    \\
    \EE_{\hatbs\Xi}[\hat L]=0,
    \qquad
    \EE_{\hatbs\Xi}[\hat L^2]
    =
    \frac{V_N}{N\ntr}
    +O\!\left(\frac{1}{N^2\ntr}\right),
\end{gathered}
\label{eq:main_direct_decomposition}
\end{equation}
where
\begin{equation}
\begin{aligned}
    b_N &=
    \bs w_\calO^T\barSigma_{\rm eff}\bs S_{\rm eff}^{-1}\bs p_{\sigma},
    \\
    V_N &= \frac{1}{\ntr}\sum_{i=1}^{\ntr}
    \left(\bs p_{i}^T \bs S_{\rm eff}^{-1} \bs p_{\sigma}\right)^2
    (\bsw^{\rm eff}_\calO)^T \bs \Sigma_i \bsw^{\rm eff}_\calO,
\end{aligned}
\label{eq:multinomial_b_V_N_main}
\end{equation}
with $\bs S_{\rm eff}=\ntr^{-1}\bs P\bsP^T+N^{-1}\overline{\bs\Sigma}_{\rm eff}$.
The effective weight vector $\bs w^{\rm eff}_{\calO}
\equiv (\bs I_{\nout} - U_2 \overline{\bs\Sigma}_{22}^{-1}\overline{\bs\Sigma}_{21} U_1^T)\bs w_\calO$ is a deformation of $\bs w_{\calO}$ induced by the Schur complement of the noise covariance.
Here, $b_N$ is the leading deterministic bias coefficient, $\hat L$ is the centered first-order fluctuation, and $\hat R$ contains all higher-order and nonlinear Schur-complement contributions.
Note that~\cref{eq:main_direct_decomposition} organizes the terms differently from the large-$N$ expansion in~\cref{eq:delta_P_11}: $\hat R$ is not simply an $O(N^{-3/2})$ remainder, as it contains contributions from both the $N^{-1/2}$ and the $N^{-1}$ terms of that expansion.
The main technical difficulty is now controlling $\hat R$ with high probability.

\parTitle{Remainder control}
In~\cref{thm:remainder_high_probability} we show that $\hat R=O_{\PP}(N^{-1/2}\ntr^{-1})$; more explicitly, for sufficiently large $\ntr$, with probability at least $1-\delta$,
\begin{equation}
\label{eq:bound_R_main}
    \lvert\hat R\rvert \lesssim \frac{\log(1/\delta)}{\sqrt{N}\,\ntr}\,,
\end{equation}
where the implicit constant may depend on the training states and the POVM, but not on $N$ or $\ntr$.
As the confidence level $1-\delta$ approaches one, the bound loosens only logarithmically: increasingly large values of $\hat R$ can occur, but only on increasingly improbable events.
This is the key feature of~\cref{thm:remainder_high_probability}, since it ensures that such events make a finite and, at large $\ntr$, subleading contribution to the relevant averages.

\parTitle{MSE control}
Accounting also for the low-probability, ill-conditioned events, we prove in~\cref{thm:direct_joint_mse} that the expected infinite-test-statistics MSE is
\begin{equation}\label{eq:main_MSE_expansion}
    \EE_{\hatbs\Xi}[\Delta\hat o^2]
    =
    \frac{b_N^2}{N^2}
    +
    \frac{V_N}{N\ntr}
    +
    O(\mathrm{Err}_{N,\ntr}).
\end{equation}
The error term $\mathrm{Err}_{N,\ntr}$, defined in~\cref{eq:def_MSE_error}, vanishes in the large-$\ntr$ limit at fixed $d$, $\nout$, and $N$.
The proof, given in~\cref{subsec:direct_mse_bound}, must additionally control rare realizations in which the relevant matrices are nearly singular.
For fixed dimensions, a sufficient condition for these events to have a negligible effect on the second moment, compared with the polynomial error terms in~\cref{eq:main_MSE_expansion}, is $\ntr/\log(2+N)\to\infty$, which includes the limit $\ntr\to\infty$ at fixed $N$.
We do not expect this condition to be necessary: it originates from a rather crude bound on these ill-conditioned events, and a more refined analysis would likely remove it.

This asymptotic MSE has two features that are absent from the Gaussian model.
First, both the training bias and the training variance depend explicitly on $N$ through the deformed population matrix $\bs S_{\rm eff}$, which contains an $N^{-1}$ correction.
Because this dependence is retained exactly rather than expanded in $1/N$, the expressions for the bias and variance remain valid at small $N$, provided $\ntr$ is large.
Second, the remainder $\mathrm{Err}_{N,\ntr}$ is now explicitly controlled and vanishes as $\ntr\to\infty$ for any fixed $N$.

\parTitle{Large-$N$ asymptotics}
For large $N$, the MSE takes a simpler form in terms of $N$-independent coefficients,
\begin{equation}
    \EE_{\hatbs\Xi}[\Delta\hat o^2]
    =
    \frac{b^2}{N^2}
    +
    \frac{V}{N\ntr}
    +
    O\!\left(\mathrm{Err}_{N,\ntr}+N^{-3}\right),
\end{equation}
where
\begin{equation}
\begin{aligned}
    b &=
    \frac{1}{\ntr}
    \sum_{i=1}^{\ntr}
    \Cov_i(\bs w^{\rm eff}_{\calO}, \ntr\langle\bstildemu,\tilde\sigma\rangle),
    \\
    V &= \frac{1}{\ntr}
    \sum_{i=1}^{\ntr}
    (\ntr\langle\tilde\rho_i,\sigma\rangle)^2
    \Var_i(\bs w^{\rm eff}_{\calO})\,,
\end{aligned}
\label{eq:multinomial_b_V_main}
\end{equation}
with \(\operatorname{Cov}_i(\bs a,\bs b)\equiv\bs a^T\bs\Sigma_i\bs b,\) and $\operatorname{Var}_i(\bs a)\equiv \Cov_i(\bs a,\bs a)$. This form, derived in~\cref{sec:costscaling_infteststat}, trades the exact $N$ dependence of $b_N$ and $V_N$ for simpler, $N$-independent coefficients, at the cost of an additional $O(N^{-3})$ error.

Assuming homogeneous covariances, $\bs\Sigma_i=\barSigma$ for all $i$, the coefficients in~\cref{eq:multinomial_b_V_main} coincide exactly with those obtained within the Gaussian approximation, \cref{eq:gaussian_b_V_main}, as shown in~\cref{subsec:gaussian_multinomial_comparison}.
Even in this homogeneous case, however, the Gaussian approximation does not reproduce the exact $N$-dependent coefficients of~\cref{eq:multinomial_b_V_N_main}; the discrepancy is of order $N^{-3}$, independently of $\ntr$.
As expected, the Gaussian approximation thus loses accuracy at small $N$.
This fixed-$N$ control is a further advantage of the exact decomposition developed here over the large-$N$ expansions of~\cref{sec:perturbation,sec:gaussian_approximation}.

The decomposition in~\cref{eq:main_direct_decomposition} thus identifies the two leading contributions to the infinite-test-statistics squared error.
The training bias originates from the deterministic $N^{-1}$ covariance shift in the effective Schur complement, whereas the training variance is generated by the centered linear fluctuation $\hat L$.
In particular, at fixed $N$ and large $\ntr$, the variance contribution becomes negligible and the MSE approaches a bias floor $b_N^2/N^2=O(N^{-2})$.

\parTitle{Quantum vs classical linear regression}
It is instructive to contrast this behavior with classical linear regression. There, if the feature vectors are known exactly and the label noise has zero mean, the least-squares estimator is unbiased in the underparametrized regime: for a design matrix $\bs X\in\mathbb R^{p\times\ntr}$ with $\ntr\ge p$ and $\rank(\bs X)=p$, one has $\bs X\bs X^+=\bs I_p$.
This conclusion does not carry over to the quantum setting considered here. The crucial difference is that the input vectors themselves are subject to stochastic noise: the ideal feature matrix $\bs P$ is replaced by the finite-shot frequency matrix $\hatbs P_N$.
Although the rank of $\bs P$ is at most the dimension $d^2$ of the operator space, $\hatbs P_N$ has rank $\min(\nout,\ntr)$ with high probability.
Finite-shot noise can therefore change the rank of the design matrix and, consequently, induce discontinuous changes in its Moore--Penrose pseudoinverse.

As a result, even in the quantum analogue of the underparametrized regime, $\ntr>\nout$, the pseudoinverse estimator acquires a systematic bias at finite $N$.
Provided the leading coefficients are nonzero and the stated conditioning and growth assumptions hold, our analysis shows that the bias scales as $N^{-1}$, and hence the squared bias as $N^{-2}$, while the training variance scales as $(N\ntr)^{-1}$.
Enlarging the training set therefore suppresses fluctuations across shot-noise realizations, but it neither changes the $N^{-2}$ scaling of the squared bias nor, in general, suppresses the bias floor by an additional factor of $\ntr^{-1}$.
Once the variance becomes subleading, adding training states yields no generic leading-order improvement unless the number of shots per training state is increased as well.
Input noise is, of course, not specific to the quantum setting: it defines errors-in-variables regression, where it is well known to bias least-squares estimators by an amount that does not vanish as the dataset grows~\cite{fuller1987MeasurementErrorModels,carroll2006MeasurementErrorNonlinear}.
The bias in \cref{eq:multinomial_b_V_N_main} has precisely this structure.
Using $N^{-1}\barSigma_{\rm eff}\bs S_{\rm eff}^{-1}=\bsI-\ntr^{-1}\bsP\bsP^T\bs S_{\rm eff}^{-1}$, the leading-order mean prediction can be written as $\bs w_\calO^T\bsp_\sigma-b_N/N=\bs w_\calO^T(\ntr^{-1}\bsP\bsP^T)\bs S_{\rm eff}^{-1}\bsp_\sigma$.
This is the multivariate attenuation formula with measurement-error covariance $\barSigma_{\rm eff}/N$, and the $N^{-1}\barSigma_{\rm eff}$ term in $\bs S_{\rm eff}$ acts as an implicit Tikhonov regularizer, in line with the classic equivalence between training with input noise and regularization~\cite{bishop1995training}.
Two features are nonetheless specific to the quantum setting.
First, the noise covariance is fixed by the Born rule and by the number of shots per state, so that the reliability of the features is a resource under the experimenter's control rather than a property of an external instrument.
Second, the noise also populates $\coker(\bsP)$, which carries no signal: noise in these directions generates the singular block of $\hatbsP_N^+$ discussed in \cref{sec:perturbation} and, through its correlations with the signal-carrying directions, replaces the bare covariance $\barSigma_{11}$ with its Schur complement $\barSigma_{\rm eff}$.

\parTitle{Explicit coefficients in symmetric models}
The leading coefficients can be evaluated exactly in a symmetric reference model, in which the training states form a complex projective two-design, the POVM is given by a fixed balanced rank-one tight frame~\cite{scott2006TightInformationallyComplete}, and the test state $\sigma$ and target projector $\calO$ are drawn independently from the Haar measure.
As shown in \cref{sec:costscaling_averages}, the corresponding averages are then
\begin{equation}
    \begin{aligned}
        \EE_{\sigma,\calO}[b^2]
        &=
        \frac{(d-1)(d+2)^2}{d+1},
        \\
        \EE_{\sigma,\calO}[V]
        &=
        \frac{d(d-1)(d+2)}{d+1}.
    \end{aligned}
    \label{eq:b1b2}
\end{equation}
Both averages scale as $d^2$ for large $d$.
Although a finite i.i.d.\ Haar-random training set is not an exact two-design, it approximately satisfies the same frame identities when its empirical frame is sufficiently well conditioned.

\parTitle{Numerical results}
These resource scalings are clearly visible in our numerical simulations.
\Cref{fig:N_scaling} shows the MSE as a function of the training statistics $N$ for several values of $\ntr$.
At small $N$, the error is dominated by the squared bias and follows the predicted $N^{-2}$ behavior, independently of the training-set size.
At larger $N$, the training variance dominates and the MSE crosses over to the $(N\ntr)^{-1}$ scaling.
The numerical results agree quantitatively with both theoretical contributions throughout the crossover.

A complementary view is obtained by fixing $N$ and increasing $\ntr$.
As shown in~\cref{fig:ntrain_saturation}, the MSE initially decreases with $\ntr$ as the training variance is suppressed, but eventually saturates at the finite-statistics bias floor, beyond which additional training states no longer improve the prediction.
For all dimensions considered, the asymptotic values observed numerically closely follow the theoretical finite-$N$ prediction.

Finally, \cref{fig:bias_var_scaling} resolves the two contributions separately.
The rescaled quantities $N^2\operatorname{bias}_{\rm tr}^2$ and $N\ntr\operatorname{Var}_{\rm tr}$ approach the coefficients $b^2$ and $V$ of~\cref{eq:multinomial_b_V_main}, respectively, and their dependence on the Hilbert-space dimension is accurately captured by the two-design averages in~\cref{eq:b1b2}.
The theory thus reproduces not only the asymptotic resource scalings of the total MSE but also the individual bias and variance coefficients.


\begin{figure}[t]
    \centering
    \includegraphics[width=1.\linewidth]{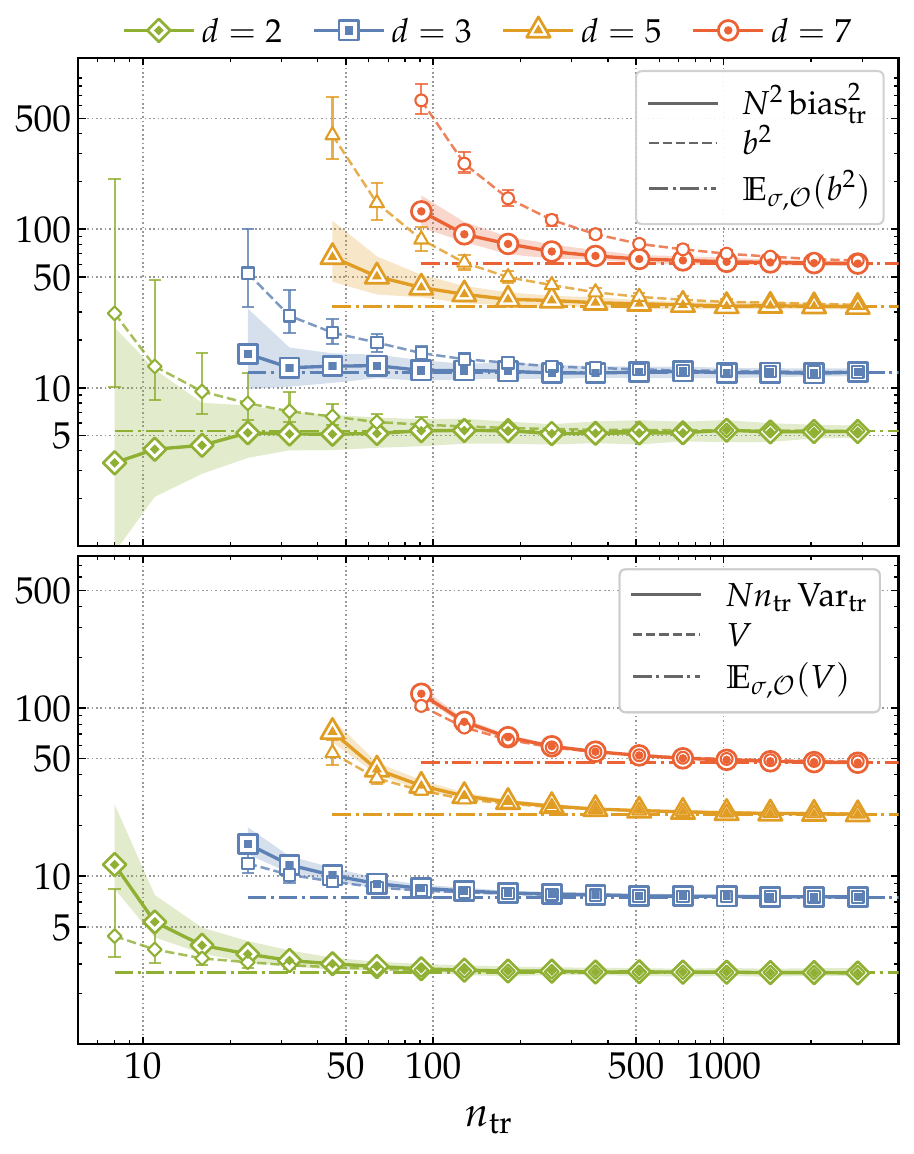}
    \caption{Training bias (top) and training variance (bottom) as functions of the training-set size $n_{\rm tr}$, at fixed shots per training state $N=10^4$, for $d=2,3,5,7$ and the MUB measurement ($n_{\rm out}=d(d+1)$). 
    Each point refers to 100 independent draws of the $n_{\rm tr}$ Haar-random pure training states, of 100 Haar-random test states $\sigma$ and of 50 Haar-random rank-one targets $\mathcal{O}$; test probabilities are exact. 
    For each draw, the mean and variance of the prediction $\hat{\bm W}_N\bm p_\sigma$ are estimated over independent shot-noise realizations of $\hat{\bm P}_N$ and averaged over test states and targets. 
    Top: $N^2\,\mathrm{bias}^2_{\rm tr}$ (solid lines and markers), the coefficient $b^2$ of \cref{eq:multinomial_b_V_main} evaluated on the same ensembles (dashed lines), and the tight-design average \cref{eq:b1b2} (dash-dotted lines). 
    Bottom: the same for $N n_{\rm tr}\,\mathrm{Var}_{\rm tr}$, $V$, and $\mathbb{E}_{\sigma,\mathcal{O}}(V)$. Solid and dashed lines are medians over the draws; shaded bands and error bars span the 10th to 90th percentiles. Both square bias and variance use $100$ shot-noise realizations per draw.}

    \label{fig:bias_var_scaling}
\end{figure}

\section{Target noise amplification}
\label{sec:asymptotic_var}

As discussed in \cref{sec:perturbation}, the divergent terms of $\hatbsP_N^+$ do not affect the MSE when considering infinite-test-statistics and exact features.
This cancellation relies on the identity $\bs y^T V_2=0$, where the columns of $V_2$ span $\ker(\bsP)$.
However, while this holds when $\bs y=\langle\bs\rho,\calO\rangle$, it fails when $\bs y$ comes with its own estimation error. Such an error can be due to imperfect preparation of the states on noisy quantum hardware.

Let the target observed
for a test state $\sigma$ be
$\hat y_\sigma=\bsW\bsp_\sigma+\hat\eta_\sigma$, where
$\hat\eta_\sigma\sim\mathcal N(0,\gamma^2)$ is independent of the empirical
test frequency $\hatbs p_\sigma$.  For any fixed weight vector $\hatbsW$, the
finite-statistics MSE then becomes
\begin{equation}
\label{eq:noisy_target_fixed_W_MSE}
\begin{gathered}
    \MSE(\hatbsW,M,\gamma)
    \equiv
    \EE_{\hatbs p_\sigma,\hat\eta_\sigma}
    \left[
        (\hatbsW\hatbs p_\sigma-\hat y_\sigma)^2
    \right]
    \\
    =
    \gamma^2
    +[(\hatbsW-\bsW)\bsp_\sigma]^2
    +\frac{1}{M}\hatbsW\bs\Sigma_\sigma\hatbsW^T.
\end{gathered}
\end{equation}
The first term is precisely the irreducible noise of classical regression in
\cref{eq:risk_MSE}.  It survives even when the ideal readout and the exact test probabilities are known.  The last term has a related, but operationally
distinct, status.  At fixed test budget $M$ it is a noise floor that no amount
of training data can remove, although it can be suppressed by increasing
$M$.  In particular, even the ideal readout $\bsW_\calO$ has the residual test error
\begin{equation}
    \MSE(\bsW_\calO,M,\gamma)
    =\gamma^2+M^{-1}\bsW_\calO\bs\Sigma_\sigma\bsW_\calO^T.
    \label{eq:fixed_W_noise_floor}
\end{equation}
As a consequence, we may interpret the test variance as an additional irreducible source of noise specific to the quantum setting. Although it can be reduced by increasing the number of measurements during testing, it can never be completely eliminated in realistic scenarios.

Target noise has an additional effect when it is present during training.
We model the training targets by
\begin{equation}
\label{eq:target_noise_model}
    \hatbs y=\bs y+\hatbs\eta,
    \qquad
    \hatbs\eta\mid(\hatbs\Xi,\bsP)
    \sim\mathcal N(0,\gamma^2\bsI_{\ntr}),
\end{equation}
and take $\hatbs\eta$, $\hat\eta_\sigma$, and the test sampling to be mutually
independent.  This homoscedastic Gaussian model isolates the rank effect; it
also describes unbiased target-estimation errors when correlations with the
feature noise can be neglected.  The corresponding estimator is
\begin{equation}
\label{eq:noisy_target_estimator}
    \hatbsW_{N,\gamma}
    =\hatbs y^T\hatbsP_N^+
    =\hatbsW_N+\hatbs\eta^T\hatbsP_N^+.
\end{equation}
Because the second term is conditionally centered, target noise does not
change the training bias.  It does, however, change both variances in the
finite-statistics error decomposition.

Consider first the test variance.  Averaging the fixed-readout expression in
\cref{eq:test_bias_var} over the training randomness gives the exact identity
\begin{equation}
\begin{aligned}
    &\EE_{\hatbsXi,\hatbs\eta}
    [\Var_{\rm test}(\hatbsW_{N,\gamma},M)]=
    \\
    &\quad=
    \EE_{\hatbsXi}
    [\Var_{\rm test}(\hatbsW_N,M)]
    +\Delta\Var_{\rm test}^{(\eta)}(N,M),
\end{aligned}
\end{equation}
where we defined extra noise-dependent reads:
\begin{equation}
\label{eq:target_noise_test_variance_exact}
    \Delta\Var_{\rm test}^{(\eta)}(N,M)
    \equiv
    \frac{\gamma^2}{M}\,
    \EE_{\hatbsXi}
    \Trace\!\left[
        \hatbsP_N^+\bs\Sigma_\sigma(\hatbsP_N^+)^T
    \right].
\end{equation}
Unlike the exact targets $\bs y$, the isotropic noise $\hatbs\eta$ and the categorical covariance matrix $\Sigma_\sigma$ have a
component in $\ker(\bsP)$.  They therefore probes the divergent $(2,2)$ block
of $\hatbsP_N^+$ in \cref{eq:Delt_PN_small}.

The same target fluctuations also contribute to the training variance.  If
$\Var_{\rm tr}(N,\gamma)$ denotes the variance of
$\hatbsW_{N,\gamma}\bsp_\sigma$ over both $\hatbsXi$ and $\hatbs\eta$, the
law of total variance gives
\begin{equation}
\label{eq:target_noise_training_variance_exact}
\begin{aligned}
    \Var_{\rm tr}(N,\gamma)
    &=\Var_{\rm tr}(N)
    +\Delta\Var_{\rm tr}^{(\eta)}(N),
    \\
    \Delta\Var_{\rm tr}^{(\eta)}(N)
    &\equiv
    \gamma^2\,
    \EE_{\hatbsXi}
    \|\hatbsP_N^+\bsp_\sigma\|_2^2.
\end{aligned}
\end{equation}
This expression contains the right projection by the physical probability
vector $\bsp_\sigma\in\im(\bsP)$ and therefore does not probe the divergent
$(2,2)$ block directly. However, it probes the $(2,1)$ block that approaches a constant as $N\to \infty$. This difference between \cref{eq:target_noise_test_variance_exact,eq:target_noise_training_variance_exact}
is responsible for their markedly different scaling with $N$.
Combining the two stages, the MSE averaged over training-target noise,
training-feature noise, test sampling, and test-target noise is therefore
\begin{equation}
\label{eq:full_noisy_target_MSE_decomposition}
\begin{aligned}
    \MSE(N,M,\gamma)
    &=\gamma^2+\bias^2(N)
      +\Var_{\rm tr}(N,\gamma)
    \\
    &\quad+
      \EE_{\hatbsXi,\hatbs\eta}
      [\Var_{\rm test}(\hatbsW_{N,\gamma},M)].
\end{aligned}
\end{equation}

The leading coefficients can be evaluated explicitly in the homogeneous
Gaussian reference model of \cref{sec:gaussian_approximation}.  The derivation
is given in \cref{subsec:gaussian_target_noise_variances}.  Assuming that the
training and test ensembles have the same average categorical covariance,
$\EE_\sigma[\bs\Sigma_\sigma]=\barSigma$, one obtains for large $N$
\begin{equation}
\label{eq:test_var_ampl}
    \EE_\sigma[\Delta\Var_{\rm test}^{(\eta)}(N,M)]
    =
    \frac{N\gamma^2}{M}C
    +O\!\left(\frac{\gamma^2\sqrt N}{M}\right)\,,
\end{equation}
where we introduced the coefficient $C$ that depends on $\ntr, \nout$ and $d$ only as:
\begin{equation}
    C=\frac{\nout-d^2}{\ntr-\nout-1}
\end{equation}
The inverse-Wishart mean in this expression exists for $\ntr>\nout+1$.
This coefficient has the same form as the classical finite-sample variance factor $p/(n-p-1)$ of ridgeless least squares with $p$ Gaussian features and $n$ samples~\cite{hastie2022surprises}, with $p=\nout-d^2$ noise-only directions spanning $\coker(\bsP)$ and $n=\ntr-d^2$ effective samples.
Thus, for fixed nonzero $\gamma$, target noise produces a test-variance contribution proportional to $N/M$.  Increasing the number of training shots drives $\hatbsP_N$ closer to the rank-deficient matrix $\bsP$ and makes the unregularized inverse more sensitive to the component of $\hatbs\eta$ in $\ker(\bsP)$.
This is the same mechanism behind the variance peak of double descent~\cite{belkin2019ReconcilingModernMachine,hastie2022surprises}, except that here the poorly conditioned directions are created by shot noise and become more poorly conditioned as $N$ grows.

By contrast, the training-variance correction is
\begin{equation}
\label{eq:training_var_target_noise_gaussian}
    \Delta\Var_{\rm tr}^{(\eta)}(N)
    =
    \gamma^2
    (C+1)
    \|\bsP^{-1}\bsp_{\sigma}\|_2^2
    +O\!\left(\frac{\gamma^2}{\sqrt N}\right).
\end{equation}
This norm is explicit in the symmetric reference model already used in
\cref{sec:costscaling_averages}.  Indeed,
$\bsP_{11}^{-1}\bsp_{\sigma,1}
=\langle\tilde{\bs\rho},\sigma\rangle$, and the tight dual frame in
\cref{eq:tight_dual_frames} gives, for any density operator $\sigma$,
\begin{equation}
\label{eq:tight_target_noise_dual_norm}
\begin{aligned}
    \|\bsP^{+}\bsp_{\sigma}\|_2^2
    =\sum_{i=1}^{\ntr}
      \langle\tilde\rho_i,\sigma\rangle^2
      =
      \frac{d}{\ntr}
      \left[(d+1)\trace(\sigma^2)-1\right].
\end{aligned}
\end{equation}
In particular, for the pure test states considered there this reduces
exactly to $d^2/\ntr$.  Thus the leading target-noise contribution in
\cref{eq:training_var_target_noise_gaussian} is
$\gamma^2d^2(C+1)/\ntr$ in that model.
Hence, at leading order,
finite-shot training features leave the multiplicative excess
\begin{equation}
    C+1=\frac{\ntr-d^2-1}{\ntr-\nout-1}
    =1+\frac{\nout-d^2}{\ntr-\nout-1}.
    \label{eq:training_variance_inflation_factor}
\end{equation}
Consequently, target noise adds an $O(\gamma^2/\ntr)$ training-variance floor,
but no factor growing with $N$.  For $\ntr\gg\nout$ at fixed $d$, the relative
inflation in \cref{eq:training_variance_inflation_factor} is only
$O(\nout/\ntr)$, while the test-variance correction remains
$O(N\gamma^2\nout/(M\ntr))$.  Increasing $M$ or $\ntr$ therefore suppresses
the amplified test contribution; increasing $N$ alone does not.

The resulting amplification is illustrated in~\cref{fig:label_noise}. For nonzero $\gamma$, the MSE first decreases as the finite-training-statistics error is suppressed, reaches a minimum, and eventually grows linearly with $N$. The simulations are in excellent agreement with the $N/M$ contribution predicted in \cref{eq:test_var_ampl}, while in the absence of target noise the error approaches the usual test-shot-noise floor.

\begin{figure}[t]
    \centering
    \includegraphics[width=1.\linewidth]{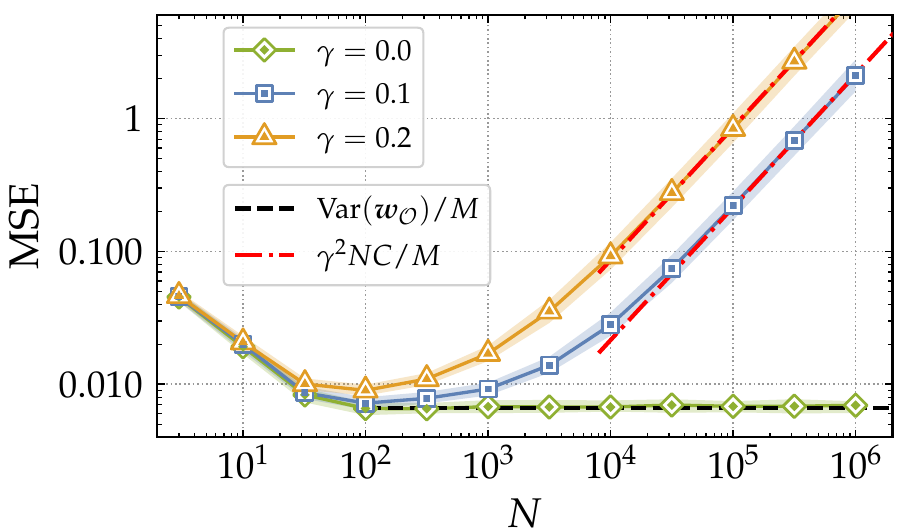}
    \caption{
    MSE as a function of $N$ in the presence of training label noise, for a single-qubit MUB measurement (thus $d=2$, $n_{\rm out}=d(d+1)=6$), with Gaussian noise of standard deviation $\gamma=0$ (green), $0.1$ (blue), and $0.2$ (orange) added to the training targets; test targets are noiseless. Markers and solid lines show the median MSE; shaded bands indicate the 10th--90th percentile range. The black dashed line indicates the large-$N$ floor $\mathrm{Var}(w_{\mathcal O})/M$, where $\mathrm{Var}(w_{\mathcal O})=(d-1)(d+2)/[d(d+1)]$, and red dash-dotted lines indicate the label-noise amplification $\gamma^2 N C/M$, where $C=(n_{\rm out}-d^2)/(n_{\rm tr}-n_{\rm out}-1)$. For $\gamma>0$ the MSE reaches a minimum around $N\sim10^2$ and then grows linearly in $N$, since more precise training features increasingly amplify the label noise. Other parameters are $M=10^2$, $n_{\rm tr}=100$, $n_{\rm test}=100$, and $n_{\rm obs}=50$ target observables, with $n_{\rm reps}=50$ independent draws of the training, test, and target ensembles and $n_{\rm real}=50$ independent shot-noise realizations per draw.
    }
    \label{fig:label_noise}
\end{figure}

\section{Conclusions and perspectives}
\label{sec:conclusions}
We analyzed how finite measurement statistics affect quantum linear
regression with probability features, focusing on the unregularized
pseudoinverse training method.
Unlike many classical settings, the feature (probability) matrix generated by an informationally complete measurement and a sufficiently rich training set is intrinsically constrained by the geometry of quantum state space: in the ideal, infinite-shot limit, it is rank-deficient, while in realistic experiments its empirical estimate is perturbed by shot noise and becomes effectively full rank.
This seemingly innocuous mismatch turns finite sampling into a qualitatively
nontrivial source of error, since learning depends on solving a linear inversion problem whose sensitivity is strongly amplified by rank-increasing fluctuations.

To address this, we developed a perturbative approach tailored to the rank-changing regime induced by shot noise.
This framework revealed an explicit bias–variance decomposition of the mean-squared prediction error, first at fixed noise realization and then after averaging over noise using standard concentration arguments.
A key outcome is that the dominant contributions exhibit distinct and operationally meaningful scalings: the squared bias term can persist even in the quantum analogue of the underparametrized regime ($\ntr > \nout$) and scales as $N^{-2}$, whereas the variance scales as $(N\ntr)^{-1}$.
These behaviors highlight a tradeoff that departs from the intuition built on regression with exactly known features: increasing the number of training states is not necessarily beneficial if the per-state measurement budget is too small, because the irreducible bias induced by rank-increasing noise cannot be mitigated by more data alone.

Beyond the specific QELM setting, our results suggest a broader lesson for quantum machine learning: resource estimates and performance guarantees must explicitly account for finite-shot effects.
In particular, the analysis motivates training protocols and benchmarking criteria that treat the total experimental budget as a joint allocation problem between the number of training states and the number of shots per state, rather than prioritizing dataset size in isolation.

Several directions follow naturally from this work. First, in a companion paper \cite{tripledescend} we will explore additional operating regimes in which finite sampling alone can induce double-descent–like behavior, and we will characterize how the location and sharpness of these transitions depend on the measurement budget and on the structure of the feature matrix.
Second, it will be important to extend our analysis to other quantum linear regression paradigms—most notably fidelity-kernel methods and quantum reservoir computing—where the features have a different statistical origin but are still estimated from finite-shot data.

In the same spirit, a systematic comparison of learning strategies based on regularized estimators (e.g., ridge/Tikhonov-type schemes, truncated inverses, early stopping) and alternative linear solvers is left for future work, with the goal of identifying when regularization mitigates shot noise–induced bias without overly inflating variance.
Finally, a pressing open question is how these effects propagate beyond linear models: it is natural to ask to what extent rank-mismatch and finite-statistics corrections persist, or are amplified, in more general QML architectures—such as variational circuits, hybrid models, and deep quantum kernels—and whether they can be leveraged to derive resource-aware training principles across quantum machine learning.

\section*{Acknowledgements}
We thank Alessandro Romancino for helpful discussions. GLM acknowledges funding from the European Union--NextGenerationEU through the Italian Ministry of University and Research under PNRR-M4C2-I1.3 Project PE-00000019 ``HEAL ITALIA'' (CUP B73C22001250006).
SL, AF, and GMP acknowledge support by MUR under PRIN Project No.~2022FEXLYB, Quantum Reservoir Computing (QuReCo).
LI, SL, MP, and GMP acknowledge funding from the ``National Centre for HPC, Big Data and Quantum Computing (HPC)'' Project CN00000013, HyQELM--Spoke 10.
MP is grateful for support from the Royal Society Wolfson Fellowship (RSWF/R3/183013), the Department for the Economy of Northern Ireland under the US-Ireland R\&D Partnership Programme, the PNRR PE Italian National Quantum Science and Technology Institute (PE0000023), and the EU Horizon Europe EIC Pathfinder project QuCoM (GA no.~10032223).

\section*{Data availability}
The data and code supporting the findings of this study are publicly available in the GitHub repository at \url{https://github.com/salvatore-lorenzo/MTR}.

\clearpage
\onecolumngrid

\setcounter{page}{1}
\renewcommand{\thepage}{S\arabic{page}}

\section*{Supplemental Materials}

\setcounter{section}{0}
\setcounter{subsection}{0}
\setcounter{equation}{0}
\setcounter{figure}{0}
\setcounter{table}{0}

\renewcommand{\thesection}{SM\arabic{section}}
\renewcommand{\thesubsection}{\thesection.\arabic{subsection}}
\renewcommand{\thetable}{SM\arabic{table}}
\renewcommand{\thefigure}{SM\arabic{figure}}
\renewcommand{\theequation}{SM\arabic{equation}}

\makeatletter
\renewcommand{\p@subsection}{}
\renewcommand{\p@subsubsection}{}
\makeatother

\renewcommand{\theHsection}{SM.\arabic{section}}
\renewcommand{\theHsubsection}{SM.\arabic{section}.\arabic{subsection}}
\renewcommand{\theHtable}{SM.\arabic{table}}
\renewcommand{\theHfigure}{SM.\arabic{figure}}
\renewcommand{\theHequation}{SM.\arabic{equation}}

\setcounter{secnumdepth}{3}
\setcounter{tocdepth}{2}

\makeatletter
\let\oldl@subsubsection\l@subsubsection
\renewcommand*\l@subsubsection[2]{}
\printsmtoc
\let\l@subsubsection\oldl@subsubsection
\makeatother

\enablesmtoc

\begin{widetext}
\section{Preliminaries}%
\label{app:notation}

We review in this section the notation used throughout the main text and
the Supplemental Material.
In particular, \cref{subsec:vectorized_notation}
translates operators and operator collections into real vectors and
matrices, \cref{subsec:landau_notation} reviews the Landau notation for deterministic sequences and sequences of random variables used in the main text and the appendix, \cref{subsec:block_decompositions} introduces the signal, kernel, and cokernel blocks of the ideal probability matrix, and \cref{app:prop_psuedo} records the pseudoinverse identities used throughout
\cref{app:perturbative_P,app:direct_remainder}.
Lowercase bold letters, such as $\bs v$, denote vectors, with components $v_k$, whereas uppercase bold letters, such as $\bs M$, denote matrices.

\subsection{Vectorized notation}
\label{subsec:vectorized_notation}

This subsection connects the operator language of quantum states, POVMs, and frames with the real matrix coordinates used in all subsequent calculations.
In particular, it makes the frame operators of \cref{app:frame} and the probability-matrix blocks of \cref{subsec:block_decompositions} part of one consistent notation.
Given a Hilbert space $\calH$ of dimension $d$, let
$\mathscr{L}(\calH)$ denote its space of linear operators, equipped with the Hilbert--Schmidt inner product
$\langle A,B\rangle=\Tr(A^\dagger B)$.  Our coordinate matrices represent the real Hilbert space $\Herm(\calH)\subset\mathscr{L}(\calH)$ using a Hermitian orthonormal basis; hence its dimension is $d^2$, its coordinate vectors are real, and matrix adjoints reduce to transposes.  Ordinary
uppercase Roman and Greek letters denote operators. 
For a vector $\psi\in\calH$, we write
$\PP(\psi)=\PP_\psi\equiv\ket{\psi}\bra{\psi}$.  We also use the same
symbol for the rank-one superoperator generated by
$A\in\mathscr{L}(\calH)$:
\begin{equation}
    \PP(A):\mathscr{L}(\calH)\to\mathscr{L}(\calH),
    \qquad
    \PP(A)(B)\equiv\langle A,B\rangle A
    =\Tr(A^\dagger B)A.
\end{equation}
This superoperator is an orthogonal projector when
$\langle A,A\rangle=1$; otherwise it is the corresponding unnormalized
rank-one outer product.

Given a Hermitian orthonormal operator basis
$(\calP_i)_{i=1}^{d^2}$ (for example, normalized Pauli operators), any
$\calO\in\Herm(\calH)$ has the real vectorized representation
$\bsv_{\calO}=\vect{\calO}$
and
$(\bsv_\calO)_i=\Tr(\calP_i^\dagger\calO)$.
Collections of Hermitian operators are denoted by bold Greek letters, as in the POVM ${\bs\mu}=(\mu_a)_{a=1}^m$.  Such a
collection is represented by the matrix
\begin{equation}
    \bsM_{\bs\mu}
    \equiv\sum_{a=1}^m\bsv_{\mu_a}{\bs e}_a^T
    =\sum_{a=1}^m\vect{\mu_a}{\bs e}_a^T,
\end{equation}
where ${\bs e}_a$ is the $a$th canonical basis vector of $\RR^m$.

\subsection{Block decompositions}
\label{subsec:block_decompositions}

We decompose noisy matrices with respect to the support and image of the
exact probability matrix.  These coordinates isolate the rank-preserving and
rank-increasing directions and are used both in the formal expansion of
\cref{app:perturbative_P} and in the direct remainder analysis of
\cref{app:direct_remainder}.  In particular, in these coordinates,
\begin{equation}
    \hatbsXi=
    \begin{pmatrix}
        \hatbsXi_{11} & \hatbsXi_{12}\\
        \hatbsXi_{21} & \hatbsXi_{22}
    \end{pmatrix},
\end{equation}
where the block decomposition is induced by the support and image of the
noiseless probability matrix \(\bsP\).  More explicitly, let
\begin{equation}
    \Pi_1\equiv \Pi_{\supp(\bsP)} =V_1V_1^T,
    \qquad
    \Pi_2\equiv \bsI-\Pi_1 =V_2V_2^T,
    \qquad
    \Pi_{\bar 1}\equiv \Pi_{\im(\bsP)} =U_1U_1^T,
    \qquad
    \Pi_{\bar 2}\equiv \bsI-\Pi_{\bar 1} =U_2U_2^T.
\end{equation}
We write
$r\equiv \on{rank}(\bsP)=d^2$.
Thus \(V_1\in\RR^{\ntr\times r}\),
\(V_2\in\RR^{\ntr\times(\ntr-r)}\),
\(U_1\in\RR^{\nout\times r}\), and
\(U_2\in\RR^{\nout\times(\nout-r)}\).
The blocks of \(\hatbsXi\) are then
$\hatbsXi_{ab}\equiv U_a^T\hatbsXi V_b$,
where $a,b\in\{1,2\}$.
Here and below the index \(a\) labels the image/cokernel block, while \(b\) labels the support/kernel block.
The normalized kernel-block covariances introduced below require
$\ntr-r>0$.  When $\ntr=r$ there is no kernel block, and the corresponding
statements are read after deleting every block with second index \(2\).

Throughout this section, the training states
$\bs\rho=(\rho_1,\ldots,\rho_{\ntr})$ are fixed and all averages are
conditional on them; the only randomness is the training shot noise.  We
write the columns of \(\hatbsXi\) as
\begin{equation}
    \hatbsXi=[\hat{\bs\xi}_1,\ldots,\hat{\bs\xi}_{\ntr}],
    \qquad
    \hat{\bs\xi}_i=\sqrt N(\hat\bsp_i-\bsp_i),
    \qquad
    \bs{\Sigma}_i\equiv
    \bs\Sigma(\rho_i) \equiv
    \EE[\hat{\bs\xi}_i\hat{\bs\xi}_i^T\mid\bs\rho]
    =
    \diag(\bsp_i)-\bsp_i\bsp_i^T ,
\end{equation}
where
$\bsp_i\equiv \langle\bs\mu,\rho_i\rangle$,
$\hat\bsp_i=\frac1N\hat{\bs n}^{(i)}$,
$\hat{\bs n}^{(i)}\sim \on{Mult}(N,\bsp_i)$.
Conditionally on \(\bs\rho\), the columns \(\hat{\bs\xi}_i\) are
independent and centered, with covariance matrices $\bs\Sigma_i$.  They
become asymptotically Gaussian as $N\to\infty$, but no Gaussian
approximation is assumed here.

It will be useful to introduce the following normalized covariance blocks.
Given a nonzero projection \(\Pi\in\RR^{\ntr\times\ntr}\) and an isometry
\(U\in\RR^{\nout\times m}\), define
$
    \hatbsC_{U,\Pi}
    \equiv
    \frac1{\Trace(\Pi)}
    U^T\hatbsXi\Pi\hatbsXi^TU,
$
and
$
    \bsC_{U,\Pi}
    \equiv
    \frac1{\Trace(\Pi)}
    \sum_{i=1}^{\ntr}\Pi_{ii}\,U^T\boldsymbol{\Sigma}_iU,
$
and note that
$\EE[\hatbsC_{U,\Pi}]=\bsC_{U,\Pi}$.
The off-diagonal entries of \(\Pi\) do not contribute to the expectation because distinct columns of \(\hatbsXi\) are independent and centered.
For the blocks relevant to the perturbative expansion we use the shorthand
\begin{equation}\label{eq:def_covblocks}
\begin{gathered}
    \hatbsC_{\alpha\beta}\equiv
    \frac{U_\alpha^T \hatbsXi \Pi_2 \hatbsXi^T U_\beta}{\ntr-r}
    =
    \frac{\hatbsXi_{\alpha 2}\hatbsXi_{\beta 2}^T}{\ntr-r},
    \qquad
    \bsC_{\alpha\beta}\equiv \EE[\hatbsC_{\alpha\beta}]
    =
    \frac1{\ntr-r}\sum_i(\Pi_2)_{ii}U_\alpha^T\boldsymbol{\Sigma}_iU_\beta,
    \qquad
    \alpha,\beta\in\{1,2\}.
\end{gathered}
\end{equation}
Thus $\hatbsC_{11}\in\mathbb{R}^{r\times r}$, $\hatbsC_{12}\in\mathbb{R}^{r\times(\nout-r)}$, $\hatbsC_{21}\in\mathbb{R}^{(\nout-r)\times r}$, and $\hatbsC_{22}\in\mathbb{R}^{(\nout-r)\times(\nout-r)}$.
Moreover,
$\bsC_{11},\bsC_{22}\ge0$ and
$\trace(\bsC_{11})+\trace(\bsC_{22})\le1$.
Because $\Pi_2$ projects onto $\supp(\bsP)^\perp$,
$\trace(\Pi_2)=\sum_i(\Pi_2)_{ii}=\ntr-r$; hence the normalization prevents
$\bsC_{ab}$ from acquiring a generic factor of $\ntr$.
The diagonal weights can be made explicit.  Since
\[
    \Pi_1=\bsP^+\bsP
    =\bsP^T(\bsP\bsP^T)^+\bsP
    =\bsM_{\bs\rho}^T
      \frameOp_{\bs\rho}^{-1}
      \bsM_{\bs\rho},
\]
we have
$(\Pi_1)_{ii}
=\langle\rho_i,\frameOp_{\bs\rho}^{-1}(\rho_i)\rangle
=\langle\rho_i,\tilde\rho_i\rangle$ and
$(\Pi_2)_{ii}=1-\langle\rho_i,\tilde\rho_i\rangle$.
In the special case where the training states are pure and form a complex-projective 2-design, we get $(\Pi_2)_{ii} = 1-d^2/\ntr=(\ntr-r)/\ntr$, and thus $\bsC_{ab}=U_a^T\overline{\bs\Sigma} U_b$ where $\overline{\bs\Sigma}\equiv\frac{1}{\ntr} \sum_{i=1}^{\ntr} \bs\Sigma_i$ is the average covariance matrix.
The 2-design property is also sufficient to get an explicit deterministic formula for $\overline{\bs\Sigma}$, which in this case equals
\begin{equation}
    \overline{\bs\Sigma}_{ab} =
    \frac{\trace(\mu_a)}{d} \delta_{ab}
    -
    \frac{\trace(\mu_a)\trace(\mu_b)+\trace(\mu_a\mu_b)}{d(d+1)}.
\end{equation}
If we further assume that the POVM is balanced and rank-one, i.e. $\mu_a=\frac{d}{\nout}\ketbra{\phi_a}$ for some pure states $\ket{\phi_a}$ satisfying $\frac{1}{\nout}\sum_{a=1}^{\nout}\ketbra{\phi_a}=\frac{\bsI}{d}$, then $\bsC_{12}=0$ and $\bsC_{22}=\frac{\bsI}{\nout}$.
These formulas will be used in \cref{sec:costscaling_averages} and in the Gaussian consistency check of \cref{subsec:gaussian_scaling_recovery}.




\subsection{Properties of pseudoinverses}%
\label{app:prop_psuedo}

We review here some basic facts about pseudoinverses that will be useful in the proofs.

\parTitle{Definition via SVD}
The pseudoinverse of a rectangular matrix is most conveniently expressed
through its compact singular value decomposition (SVD).
Given a matrix $\bs A$, write
$\bsA=\bs U\bsD\bs V^\dagger$, where $\bs U,\bs V$ are isometries and
$\bs D>0$ is diagonal.
Its pseudoinverse then reads $\bsA^+={\bs V} \bsD^{-1} {\bs U}^\dagger$.
If $\bsA$ is $n\times m$, then $\bs{U}$ is $n\times r$, $\bsD$ is $r\times r$, and $\bs V$ is $m\times r$, where $r\equiv\on{rank}(\bsA)$.
In dyadic notation, this amounts to $\bsA=\sum_k d_k \boldsymbol u_k \boldsymbol v_k^\dagger$ having pseudoinverse $\bsA^+=\sum_k d_k^{-1} \boldsymbol v_k \boldsymbol u_k^\dagger$.

\parTitle{Properties of the pseudoinverse}
The pseudoinverse $\bsA^+$ is a true (left and right) inverse iff $\bsA$ is invertible.
More generally, $\bsA^+$ is a left inverse,
$\bsA^+\bsA=\bsI$, if and only if $\bsA$ is injective, and it is a right
inverse, $\bsA\bsA^+=\bsI$, if and only if $\bsA$ is surjective.
If $\bsA$ is not full-rank, then $\bsA^+$ is neither a left nor a right inverse, but it can still be seen as an inverse on the range of support of $\bsA$.
In general, $\bsA^+\bsA=\Pi_{\supp(\bsA)}$ is the orthogonal projector
onto the row space of $\bsA$, whereas
$\bsA\bsA^+=\Pi_{\im(\bsA)}$ projects onto its image.

\parTitle{Explicit formulas for special cases}
The pseudoinverse can be given an explicit formula when $\bsA$ is injective or surjective.
For injective $\bsA$, we have $\bsA^+=(\bsA^\dagger \bsA)^{-1} \bsA^\dagger$. For surjective $\bsA$, we have $\bsA^+=\bsA^\dagger (\bsA\bsA^\dagger)^{-1}$.

\parTitle{Product formulas}
It is not generally the case that $(\bsA\bsB)^+=\bsB^+ \bsA^+$.
However, it does hold when $\bsA$ is injective and $\bsB$ is
surjective.  Indeed, the compact SVD of $\bsB$ is
$\bsB={\bs U}_B\bsD_B{\bs V}_B^\dagger$, with ${\bs U}_B$ unitary,
whereas that of $\bsA$ is
$\bsA={\bs U}_A\bsD_A{\bs V}_A^\dagger$, with ${\bs V}_A$ unitary.
Therefore
\begin{equation}
    \bsA\bsB = {\bs U}_A \bsD_A {\bs V}_A^\dagger {\bs U}_B \bsD_B {\bs V}_B^\dagger,
\end{equation}
where ${\bs V}_A^\dagger{\bs U}_B$ is unitary and
$\bsD_A,\bsD_B$ are square and invertible.  Hence
$\bsC\equiv\bsD_A{\bs V}_A^\dagger{\bs U}_B\bsD_B$ is also square and
invertible.  Let $\bsC=\bs U_C\bsD_C\bs V_C^\dagger$ be its SVD.  Then
\begin{equation}
    \bsA\bsB = ({\bs U}_A {\bs U}_C) \bsD_C  ({\bs V}_B {\bs V}_C)^\dagger 
\end{equation}
is an SVD for $\bsA\bsB$, with ${\bs U}_C, {\bs V}_C$ unitaries, hence
\begin{equation}
    (\bsA\bsB)^+ = ({\bs V}_B {\bs V}_C) \bsD_C^{-1} ({\bs U}_A {\bs U}_C)^\dagger
    = {\bs V}_B \bsC^{-1} {\bs U}_A^\dagger,
\end{equation}
which equals $\bsB^+ \bsA^+$.

\parTitle{Block inverses and Schur complements}
Given an invertible square matrix
$\bsM=\left(\begin{smallmatrix}\bsA&\bsB\\\bsC&\bsD\end{smallmatrix}\right)$,
its inverse can be written in terms of Schur complements
\cite{zhang2005SchurComplementIts} as
\begin{equation}
\label{eq:M-1}
\newcommand{\redSchurOneInv}{(\bsM/\bsA)^{-1}}
\newcommand{\blueSchurTwoInv}{(\bsM/\bsD)^{-1}}
\begin{aligned}
    \bsM^{-1}
    &= \begin{pmatrix}
        \bsA^{-1} + \bsA^{-1} \bsB \redSchurOneInv\bsC \bsA^{-1} &
        -\bsA^{-1}\bsB \redSchurOneInv
        \\
        -\redSchurOneInv \bsC \bsA^{-1}
        & \redSchurOneInv
    \end{pmatrix}
    \\
    &= \begin{pmatrix}
        \blueSchurTwoInv &
        - \blueSchurTwoInv \bsB\bsD^{-1}
        \\
        -\bsD^{-1}\bsC \blueSchurTwoInv
        &
        \bsD^{-1} + \bsD^{-1}\bsC\blueSchurTwoInv\bsB\bsD^{-1}
    \end{pmatrix},
\end{aligned}
\end{equation}
where
\begin{equation}
    {(\bsM/\bsA)} \equiv \bsD-\bsC\bsA^{-1}\bsB,
    \qquad
    {(\bsM/\bsD)} \equiv \bsA-\bsB\bsD^{-1}\bsC.
\end{equation}
The first representation assumes that $\bsA$ is invertible, whereas the
second assumes that $\bsD$ is invertible.  This identity is used by both
the complementary structural analysis and the direct proof below.

\subsection{Frame operators}%
\label{app:frame}

The block formulas describe the inverse in matrix coordinates, but the conditioning constants and prediction coefficients must ultimately be expressed in terms of the states and the POVM. The formalism of frame operators is useful to provide this link. In particular, their eigenvalues quantify how well training states and POVM elements cover the operator space, and their canonical duals give the ideal readout weights.
We thus review here the basic language of operator frames which will be useful for some of the later results in \cref{app:asymptotic}.

We make frequent use of frames and frame operators
\cite{casazza2013introduction,casazza2016brief,christensen2003introduction}.
For a finite-dimensional inner-product space $V$, a frame is a collection
$(v_k)_k\subset V$ for which
\begin{equation}
    \alpha\|w\|^2\leq\sum_k|\langle v_k,w\rangle|^2\leq \beta\|w\|^2
\end{equation}
for some $0<\alpha\leq\beta<\infty$.  Frames generalize orthonormal bases
by allowing overcompleteness.  Given any dual frame
$(\tilde v_k)_k$, every $w\in V$ admits the reconstructions
\begin{equation}
    w=\sum_k\langle v_k,w\rangle\tilde v_k
    =\sum_k\langle\tilde v_k,w\rangle v_k.
\end{equation}
The canonical dual is
$\tilde v_k=\frameOp^{-1}(v_k)$, where the positive, invertible frame
operator is $\frameOp=\sum_k v_kv_k^\dagger$.  An overcomplete frame has nonunique duals.

If we take $V$ to be the Hilbert space of linear operators on $\calH$, \textit{i.e.}, $V=\calL(\calH)$, then both the set of training states $\bs{\rho}$ and the POVM $\bs{\mu}$ can be thought of as frames in $\calL(\calH)$, provided that $\ntr,\nout\geq d^2$ and that each ensemble contains at least $d^2$ linearly independent elements. This means that any operator $\calO\in\calL(\calH)$ can be decomposed as
\begin{equation}
    \calO = \sum_{i=1}^{\ntr}\langle\calO,\tilde{\rho}_i\rangle\,\rho_i = \sum_{a=1}^{\nout}\langle\calO,\tilde{\mu}_a\rangle\,\mu_a\,,
\end{equation}
where $\tilde{\bs{\rho}}=\{\tilde{\rho}_i\}_{i=1}^{\ntr}$ and $\tilde{\bs{\mu}}=\{\tilde{\mu}_a\}_{a=1}^{\nout}$ are the dual training-state ensemble and the dual POVM, respectively. As mentioned above, whenever $\ntr,\nout>d^2$, the dual frame is not unique. However, there is a canonical choice, which we will always adopt:
\begin{equation}
    \tilde{\bs{\rho}}=\frameOp_{\bs{\rho}}^{-1}(\bs{\rho})\,,\qquad \tilde{\bs{\mu}}=\frameOp_{\bs{\mu}}^{-1}(\bs{\mu})\,,
\end{equation}
where $\frameOp_{\bs{\rho}}$ is the superoperator defined by $\frameOp_{\bs{\rho}}(\calO)=\sum_{i=1}^{\ntr}\Tr(\calO\rho_i)\,\rho_i$, and $\frameOp_{\bs{\mu}}$ is defined analogously.

In the vectorized formalism, the frame operator associated with the training set $\bs\rho$ can be written as $\frameOp_{\bs \rho} \equiv \bsMrho\bsMrho^T$ and the frame operator associated with the POVM as $\bs\mu$ $\frameOp_{\bsmu}\equiv \bsMmu\bsMmu^T$. They are both intimately related to the probability matrix $\bsP\equiv\langle\bs\mu,\bs\rho\rangle\equiv \bsMmu^T\bsMrho$.
In particular, this relation naturally emerges by looking at their singular values.
The singular values of the probability matrix are by definition the square roots of the eigenvalues of $\bsP\bsP^T=\bsMmu^T\bsMrho\bsMrho^T\bsMmu$ that in turn are the same as the nonzero eigenvalues of $\bsMrho\bsMrho^T\bsMmu\bsMmu^T= \frameOp_{\bs\rho}\frameOp_{\bs\mu}$.
We therefore conclude that the nonzero singular values of $\bsP$ are equal to the square roots of the nonzero eigenvalues of $\frameOp_{\bs\mu}\frameOp_{\bs\rho}$.

\parTitle{Case of symmetric states and measurement}
In particular, if the training states are pure,
$\rho_k=\PP(\psi_k)$, and form a complex-projective 2-design, then
in superoperator notation
\begin{equation}\label{eq:frameop_symmetric_states}
    \frameOp_{\bs\rho}
    \equiv
    \sum_{k=1}^{\ntr}\PP(\rho_k)
    = \ntr\frac{\Id+\PP(\bsI)}{d(d+1)},
\end{equation}
where $\Id$ is the identity superoperator and
$\PP(\bsI)(X)=\Trace(X)\bsI$, with $\bsI$ the identity operator
on $\calH$.
Note that, as a linear operator, this $\frameOp_{\bs\rho}$ has eigenvalues $\frac{\ntr}{d}$ and $\frac{\ntr}{d(d+1)}$, with degeneracies $1$ and $d^2-1$, respectively.
If the training states are independent Haar-random pure states, then their normalized empirical frame operator approaches this form with high probability as $\ntr$ grows.
Similarly, symmetric rank-one POVMs have the form $\mu_b=\frac{d}{\nout} \PP(\phi_b)$ with the states $\ket{\phi_b}$ such that
$\sum_{b=1}^{\nout} \PP(\phi_b)=\frac{\nout}{d}\bsI$. If these states further form a complex-projective 2-design, then the frame superoperator reads
\begin{equation}\label{eq:frameop_symmetric_POVMs}
    \frameOp_{\bs\mu} \equiv
    \sum_{b=1}^{\nout} \PP(\mu_b) =
    \frac{d^2}{\nout} \frac{\Id+\PP(\bsI)}{d(d+1)}.
\end{equation}
Furthermore, if both states and POVM have these symmetric properties, then we find
\begin{equation}
    \frameOp_{\bs\rho}\circ\frameOp_{\bs\mu}
    = \frac{\ntr d^2}{\nout[d(d+1)]^2}
    (\Id+\PP(\bsI))^2
    = \frac{\ntr}{\nout(d+1)^2}(\Id+(2+d)\PP(\bsI)),
\end{equation}
whose eigenvalues are
\begin{equation}
    \sigma(\frameOp_{\bs\rho}\circ\frameOp_{\bs\mu})
    \in \left\{
    \frac{\ntr}{\nout},
    \frac{\ntr}{\nout(d+1)^2}
    \right\}.
\end{equation}
The eigenvector corresponding to $\ntr/\nout$ is the rescaled identity, $\bsI/\sqrt d$, while the other eigenvalue has degeneracy $d^2-1$, and eigenspace spanned by any choice of orthonormal basis for the space of Hermitian traceless $d\times d$ operators.
The singular values of $\bsP$ are then the square roots of the eigenvalues of $\frameOp_{\bs\rho}\circ\frameOp_{\bs\mu}$.
Thus all $d^2$ nonzero singular values of the
$\nout\times\ntr$ probability matrix scale as
$\sqrt{\ntr/\nout}$ (with the two dimension-dependent prefactors shown
above).

\section{Perturbative pseudoinverse expansion}%
\label{app:perturbative_P}

We discussed in the main text how the full empirical pseudoinverse diverges even though physical predictions can remain finite.
In this section we present a perturbative analysis of the pseudoinverse of the empirical probability matrix, identifying where exactly divergences arise.
The main results are the expansions in~\cref{eq:expansionPp11,eq:expansionPp_other}.
The signal block in \cref{eq:expansionPp11} is used later in the homogeneous Gaussian calculation of \cref{app:concentration}, while the remaining blocks in \cref{eq:expansionPp_other} are needed for the target-noise correction.

\parTitle{Discontinuity of pseudoinverses under rank-increasing perturbations}
The main technical difficulty we seek to overcome is that matrix inversion is not continuous with respect to perturbations that change the rank, and the same is true for the Moore-Penrose pseudoinverse~\cite{stewart1977perturbation,stewart1990Matrix,gonzalez2015laurent}.
Let $\bsP$ be the unperturbed matrix and let
$\hatbsP_\tau=\bsP+\tau\hatbsXi$.
Here $\tau$ is a generic parameter; we will use $\tau=N^{-1/2}$ when specializing the discussion to the empirical frequency matrix.
If $\bsP$ is rank deficient and the perturbation $\hatbsXi$ has components outside the row or column spaces of $\bsP$, then $\tau\mapsto\hatbsP_\tau^+$ is discontinuous at $\tau=0$, and in particular $\hatbsP_\tau^+$ can diverge as $\tau\to0$.
More precisely, when the perturbation changes the rank, one has the lower bound~\cite{stewart1990Matrix}
\begin{equation}
\label{eq:discontinuity_of_pinv}
    \|\hatbsP_\tau^+ - \bsP^+\|_{\rm op}
    \ge
    \frac{1}{\tau\|\hatbsXi\|_{\rm op}}.
\end{equation}
For finite-dimensional matrices, continuity of the pseudoinverse is equivalent to the rank being stable under the perturbation.
However, the stochastic setting considered here generally does cause the rank to increase, being $\bsP$ has rank at most $d^2$, while shot noise makes $\hatbsP$ have maximal rank $\min(\ntr,\nout)$ with high probability.
Here and throughout, $\|\cdot\|_{\rm op}$ denotes the spectral norm, i.e.\ the largest singular value.

\parTitle{Rank-increasing perturbations for square matrices}
We can show more explicitly the effects of the discontinuity leveraging the block formula in \cref{eq:M-1}. Before discussing the rectangular case, consider first a symmetric singular matrix $\bsM$ that undergoes a rank-increasing perturbation $\bsM_\tau\equiv \bsM+\tau \bs X$, and let $\bs x\in\supp(\bsM)$ and $\bs y\in\im(\bsM)$.
In a basis adapted to $\supp(\bsM)\oplus\ker(\bsM)$, we can write these as
\begin{equation}
    \bsM = \begin{pmatrix}
        \bsM_{11} & 0 \\ 0 & 0
    \end{pmatrix},
    \qquad
    \bsM_\tau = \begin{pmatrix}
        \bsM_{11} + \tau \bsX_{11}
        & \tau\bsX_{12}
        \\
        \tau\bsX_{21} & \tau\bsX_{22}
    \end{pmatrix},
    \qquad
    \bs x = \binom{\bs x_1}{0},
    \qquad
    \bs y = \binom{\bs y_1}{0}.
\end{equation}
The Schur-complement identities in \cref{eq:M-1} then give
\begin{equation}
    \bsx^T \bsM_\tau^{-1} \bs y
    =
    \bsx_1^T
    (\bsM_\tau/\tau\bsX_{22})^{-1}
    \bs y_1
    =
    \bsx_1^T
    [\bsM_{11}+\tau(\bsX/\bsX_{22})]^{-1}
    \bs y_1,
\end{equation}
where $\bsX_{ij}$ are the blocks of $\bsX$ and
$\bsX/\bsX_{22}
\equiv\bsX_{11}-\bsX_{12}\bsX_{22}^{-1}\bsX_{21}$.
Thus, to leading order in $\tau$, we have
\begin{equation}
    \bsx^T \bsM_\tau^{-1}\bs y
    =
    \bsx_1^T \bsM_{11}^+ \bs y_1
    - \tau
      \bsx_1^T\bsM_{11}^{-1}
      (\bsX/\bsX_{22})
      \bsM_{11}^{-1}\bs y_1
    + O(\tau^2).
\end{equation}
This can be compared to the naive perturbative expansion, which would have been
$\bsx^T\bsM_\tau^{-1}\bs y
= \bsx_1^T \bsM_{11}^+ \bs y_1
-\tau\bsx_1^T
\bsM_{11}^{-1}\bsX_{11}\bsM_{11}^{-1}\bs y_1
+O(\tau^2)$.
Thus, although the full pseudoinverse is discontinuous under a rank-increasing perturbation, the bilinear form
$\bsx^T\bsM_\tau^{-1}\bs y$ remains continuous when $\bs x,\bs y$ are only nonzero in the support of $\bsM$. Nonetheless, the divergent block of $\bsM_\tau^{-1}$ changes the order-$\tau$ correction by replacing $\bsX_{11}$ with $\bsX_{11}-\bsX_{12}\bsX_{22}^{-1}\bsX_{21}$.

\parTitle{Pseudoinverse of noisy probability matrices}
Let us now specialize this observation to the probability matrix, with $\hatbsP_\tau\equiv \bsP+\tau \hatbsXi$, $\tau\equiv N^{-1/2}$.
Note that $\hatbsXi$ does not depend on $\tau$ in its first two moments, but might depend on it at higher order; for example, its third and higher cumulants are $O(N^{-1/2})$. The dependence on $N$ disappears from $\hatbsXi$ only in the Gaussian limit $N\to\infty$.
In the block coordinates introduced in \cref{subsec:block_decompositions}, we write $\bs A_{ij}\equiv U_i^T \bs A V_j$, $i,j=1,2$, and thus
\begin{equation}
    \hatbsP_\tau
    \equiv
    \begin{pmatrix}
        U_1^T\hatbsP_\tau V_1
        & U_1^T\hatbsP_\tau V_2 \\
        U_2^T\hatbsP_\tau V_1
        & U_2^T\hatbsP_\tau V_2
    \end{pmatrix}
    =
    \begin{pmatrix}
        \bsP_{11} + \tau\,\hatbsXi_{11}
        & \tau\,\hatbsXi_{12} \\
        \tau\,\hatbsXi_{21}
        & \tau\,\hatbsXi_{22}
    \end{pmatrix}.
\end{equation}
We assume throughout that the empirical matrix is wide, $\ntr\ge\nout$, and has full row rank after perturbation.
The formulas in this formal expansion are conditional on $\hatbsXi_{22}$ having full row rank and on the displayed inverses existing.
This holds almost surely in the homogeneous Gaussian reference model when $\ntr > \nout$.
For the true multinomial model, the probability of $\hatbsP_\tau$ and $\hatbsXi_{22}$ not having full row rank is nonzero, but decreases exponentially fast with $\ntr$.
When $\hatbsP_\tau$ is full rank, the pseudoinverse is
$\hatbsP_\tau^+
=
\hatbsP_\tau^T
(\hatbsP_\tau\hatbsP_\tau^T)^{-1}$,
where
\begin{equation}
\label{eq:Mtau_definition}
    \hatbsM_\tau
    \equiv
    \hatbsP_\tau\hatbsP_\tau^T
    =
    \bsP\bsP^T
    +
    \tau\left(\bsP\hatbsXi^T+\hatbsXi\bsP^T\right)
    +
    \tau^2\hatbsXi\hatbsXi^T .
\end{equation}
In the block decomposition the components of $\hatbsM_\tau$ are, explicitly,
\begin{equation}\label{eq:componentsMtau}
\begin{alignedat}{2}
    \hatbsM_{11} &= \bsP\bsP^T
    + \tau(\bsP\hatbsXi^T+\hatbsXi\bsP^T)_{11}
    + \tau^2 \sum_{c=1}^2 \hatbsXi_{1c}\hatbsXi^T_{1c},
    \qquad&
    \hatbsM_{12} &= \tau \bsP\hatbsXi^T_{21}
    + \tau^2 \sum_{c=1}^2 \hatbsXi_{1c}\hatbsXi^T_{2c}, \\
    \hatbsM_{21} &= \tau \hatbsXi_{21}\bsP^T
    + \tau^2 \sum_{c=1}^2 \hatbsXi_{2c}\hatbsXi^T_{1c},
    \qquad&
    \hatbsM_{22} &= \tau^2 \sum_{c=1}^2 \hatbsXi_{2c}\hatbsXi^T_{2c},
\end{alignedat}
\end{equation}
where we suppressed for notational simplicity the subscript $\tau$ on $\hatbsM_\tau$ to write its blocks as
$\hatbsM_{ab}$.
From \cref{eq:M-1} we know that $(\hatbsM^{-1})_{11} = \blueMSchurComplementTwoInv$ and $(\hatbsM^{-1})_{22} = (\hatbsM/\hatbsM_{11})^{-1}$, where
\begin{equation}\label{eq:definitions_Schurs_M}
\begin{gathered}
    {\hatbsM/\hatbsM_{11}} \equiv
    \hatbsM_{22} -
    \hatbsM_{21}\hatbsM_{11}^{-1}\hatbsM_{12},
    \qquad
    \blueMSchurComplementTwo \equiv
    \hatbsM_{11} -
    \hatbsM_{12}\hatbsM_{22}^{-1}\hatbsM_{21}.
\end{gathered}
\end{equation}
More generally, we can write $\hatbsM^{-1}$ as
\begin{equation}
\newcommand{\colsc}{(\hatbsM/\hatbsM_{11})^{-1}}
\begin{aligned}
    \hatbsM_\tau^{-1}
    &=
    \begin{pmatrix}
        \hatbsM_{11}^{-1}
        + \hatbsM_{11}^{-1}\hatbsM_{12}\colsc\hatbsM_{21}\hatbsM_{11}^{-1}
        &
        - \hatbsM_{11}^{-1}\hatbsM_{12}\colsc \\
        -\colsc\hatbsM_{21} \hatbsM_{11}^{-1}
        & \colsc
    \end{pmatrix}
    \\
    &=
    \begin{pmatrix}
        \blueMSchurComplementTwoInv
        &
        -\blueMSchurComplementTwoInv\hatbsM_{12}\hatbsM_{22}^{-1} \\
        -\hatbsM_{22}^{-1} \hatbsM_{21} \blueMSchurComplementTwoInv
        &
        \hatbsM_{22}^{-1} + \hatbsM_{22}^{-1}\hatbsM_{21}\blueMSchurComplementTwoInv
        \hatbsM_{12}\hatbsM_{22}^{-1}
    \end{pmatrix}.
\end{aligned}
\end{equation}

\parTitle{Simplification of Schur complements}
To calculate perturbative orders for small $\tau$ it is useful to simplify these expressions.
From \cref{eq:definitions_Schurs_M,eq:componentsMtau} we get
\begin{equation}\label{eq:cool_expr_MM22}
\begin{gathered}
    {\hatbsM/\hatbsM_{22}} =
    [\bsP + \tau(\hatbsXi/\hatbsXi_{22})]
    \hatbsZ^{-1}
    [\bsP + \tau(\hatbsXi/\hatbsXi_{22})]^T
    + \tau^2 \hatbsXi_{12}\hatbsQ\hatbsXi_{12}^T,
    \\ 
    \hatbsZ\equiv \bsI + \hatbsXi_{21}^T (\hatbsXi_{22}\hatbsXi_{22}^T)^{-1} \hatbsXi_{21},
    \qquad
    \hatbsXi/\hatbsXi_{22}
    \equiv
    \hatbsXi_{11}
    -
    \hatbsXi_{12}\hatbsXi_{22}^+\hatbsXi_{21},
    \qquad
    \hatbsQ
    \equiv
    \bsI-\hatbsXi_{22}^+\hatbsXi_{22}.
\end{gathered}
\end{equation}
Here we introduced the generalized Schur complement of the perturbation $\hatbsXi/\hatbsXi_{22}$, and the operator
$\hatbsQ$, which projects onto the kernel of $\hatbsXi_{22}$ and therefore captures the correction that is present in the rectangular case and vanishes when $\hatbsXi_{22}$ is square and full rank.
At order $\tau^0$ this becomes
\begin{equation}
    \left.\hatbsM/\hatbsM_{22}\right|_{\tau=0}=\bsP\bsP^T
    -
    \bsP\hatbsXi_{21}^T
    (
    \hatbsXi_{21}\hatbsXi_{21}^T + \hatbsXi_{22}\hatbsXi_{22}^T
    )^{-1}
    \hatbsXi_{21}\bsP^T
    = \bsP(
    \bsI + \hatbsXi_{21}^T
    (
    \hatbsXi_{22}\hatbsXi_{22}^T
    )^{-1}
    \hatbsXi_{21}
    )^{-1} \bsP^T,
\end{equation}
where we used the \textit{Woodbury identity}
$\bsI - \bsA^T(\bsA\bsA^T + \bsQ)^{-1}\bsA = (\bsI+\bsA^T\bsQ^{-1}\bsA)^{-1}$.
Similar manipulations are used to derive the rest of~\cref{eq:cool_expr_MM22}. 
We now use these expressions to perform a perturbative analysis of the pseudoinverse $\hatbsP_\tau^+=\hatbsP_\tau^T \hatbsM^{-1}$.
We have
\begin{equation}%
\label{eq:expansionPp11}
\begin{aligned}
    (\hatbsP_\tau^+)_{11}
    &=
    (\bsP^T+\tau \hatbsXi^T_{11})(\hatbsM^{-1})_{11}
    + \tau \hatbsXi^T_{21} (\hatbsM^{-1})_{21}
    = (\bsP^T+\tau \hatbsXi_{11}^T
    - \tau \hatbsXi^T_{21} \hatbsM_{22}^{-1}\hatbsM_{21})
    \blueMSchurComplementTwoInv
    \\
    &= \bsP^+
    -
    \tau\bsP^+(\hatbsXi/\hatbsXi_{22})\bsP^+
    +
    \tau^2\bsP^+(\hatbsXi/\hatbsXi_{22})\bsP^+(\hatbsXi/\hatbsXi_{22})\bsP^+
    -
    \tau^2
    \bsP^+\hatbsXi_{12}\hatbsQ\hatbsXi_{12}^T
    \bsP^{-T}\hatbsZ\bsP^+
    +
    O(\tau^3).
\end{aligned}
\end{equation}
This $(1,1)$ block of $\hatbsP_\tau^+$ is particularly important as it is the only part of $\hatbsP_\tau$ that enters predictions in the limit of infinite test statistics and exact targets.
In the square case, when $\hatbsQ=0$, this expression reduces to $[\bsP+\tau(\hatbsXi/\hatbsXi_{22})]^{-1}$, which can be expanded as
\begin{equation}
    [\bsP+\tau(\hatbsXi/\hatbsXi_{22})]^{-1}
    =
    \sum_{i=0}^{\infty}
    [-\tau\bsP^{-1}(\hatbsXi/\hatbsXi_{22})]^i\bsP^{-1}.
\end{equation}
This expansion is convergent whenever
$\tau\|\bsP^{-1}(\hatbsXi/\hatbsXi_{22})\|_{\rm op}<1$.
We can similarly compute the expressions for the other components of the pseudoinverse:
\begin{equation}
\begin{split}
\label{eq:expansionPp_other}
    (\hatbsP_\tau^+)_{12} &=
    -\left[\bsP^{-1}
    - \tau\bsP^{-1}(\hatbsXi/\hatbsXi_{22}) \bsP^{-1}\right]
    \hatbsXi_{12}\hatbsXi_{22}^+
    + \tau \bsP^{-1} \hatbsXi_{12}\hatbsQ\hatbsXi_{12}^T \bsP^{-T}\hatbsXi_{21}^T(\hatbsXi_{22}\hatbsXi_{22}^T)^{-1}
    + O(\tau^2)\\
    (\hatbsP_\tau^+)_{21}
    &= - \hatbsXi_{22}^+\hatbsXi_{21}
    \left[\bsP^{-1}
    - \tau\bsP^{-1}(\hatbsXi/\hatbsXi_{22}) \bsP^{-1}\right]
    + \tau \hatbsQ\hatbsXi_{12}^T\bsP^{-T}\hatbsZ\bsP^{-1}
    + O(\tau^2),\\
    (\hatbsP_\tau^+)_{22}
    &= \frac{\hatbsXi_{22}^+}{\tau}
    + \hatbsXi_{22}^+\hatbsXi_{21} \bsP^{-1} \hatbsXi_{12} \hatbsXi_{22}^+
    - \hatbsQ
    \hatbsXi_{12}^T \bsP^{-T} \hatbsXi_{21}^T(\hatbsXi_{22}\hatbsXi_{22}^T)^{-1}
    + O(\tau).
\end{split}
\end{equation}
Note the asymmetry in the corrections for $(\hatbsP_\tau^+)_{12}$ and $(\hatbsP_\tau^+)_{21}$. This is caused by the asymmetric nature of the problem, due to the difference in the number of rows and columns of $\hatbsP_N$.
The off-diagonal blocks $(\hatbsP_\tau^+)_{12}$ and $(\hatbsP_\tau^+)_{21}$ have finite limits as $\tau\to0$.
By contrast, the singular contribution is localized in the $(2,2)$ block, which contains the divergent term $\tau^{-1}\hatbsXi_{22}^+$.
This is the expected discontinuity of the pseudoinverse when the perturbation populates directions outside the row and column spaces of the unperturbed matrix.

In conclusion, the perturbative expansion of $\hatbsP_\tau^+$ coincides with
that of the square case upon replacing $\hatbsXi_{22}^{-1}$ with
$\hatbsXi_{22}^+$ in the Schur complements, up to rectangular corrections
whose leading orders are
\begin{equation}\label{eq:leading_corrections_rectangular_case}
    \Delta_{\hatbsQ}\hatbsP^+_\tau=
    \begin{pmatrix}
        - \tau^2\bsP^{-1}\hatbsXi_{12}\hatbsQ \hatbsXi_{12}^T\bsP^{-T}\hatbsZ\bsP^{-1}
        &
        \tau\bsP^{-1}\hatbsXi_{12} \hatbsQ \hatbsXi_{12}^T \bsP^{-T}\hatbsXi_{21}^T(\hatbsXi_{22}\hatbsXi_{22}^T)^{-1}
        \\
        \tau \hatbsQ\hatbsXi_{12}^T\bsP^{-T}\hatbsZ \bsP^{-1}
        &
        - \hatbsQ
    \hatbsXi_{12}^T \bsP^{-T} \hatbsXi_{21}^T(\hatbsXi_{22}\hatbsXi_{22}^T)^{-1}
    \end{pmatrix}
    +\text{subleading terms depending on }\hatbsQ.
\end{equation}
When $\ntr=\nout$ and $\hatbsXi_{22}$ is nonsingular,
$\hatbsQ=0$ and these corrections vanish.  Nonsingularity holds almost
surely in the continuous Gaussian reference model; for finite multinomial
noise it is an event whose probability must be controlled rather than an
identity.

\begin{table}[h!]
    \small
    \renewcommand{\arraystretch}{1.18}
    \begin{tabular}{@{}ll@{}}
        \toprule
        \proofgroup{Basic definitions (\cref{app:notation})}
        \proofsymcell{$r$}
        &
        \proofdefcell{%
            $r=\rank(\bsP)=d^2$
        }
        \\
        \proofsymcell{$\bsP$}
        &
        \proofdefcell{%
            $\bsP = \langle\bs\mu,\bs\rho\rangle$,
            $\qquad\bsP^+ = \langle\tilde{\bs\rho},\bstildemu\rangle$,
            $\qquad\bsP_i\equiv U_i^T\bsP$
        }
        \\
        \proofsymcell{$\bs y$, $\bsp_\sigma$}
        &
        \proofdefcell{%
            $\bs y
            =\langle\bs\rho,\calO\rangle
            =\bsP^T\bs w_\calO
            =\bsP_1^T U_1^T\bs w_\calO
            \in\im(\bsP^T)$,
            \quad
            $\bsp_\sigma
            =\langle\bs\mu,\sigma\rangle
            \in\im(\bsP)$}
        \\
        \proofsymcell{$\bs w_\calO$}
        &
        \proofdefcell{%
            $\bs w_\calO
            =\langle\bstildemu,\calO\rangle$}
        \\
        \proofsymcell{$U_1,U_2$}
        &
        \proofdefcell{%
            $\im(U_1)=\im(\bsP)$,
            \quad
            $\im(U_2)=\coker(\bsP)$,
            \quad
            $U_1\in\RR^{\nout\times r}$,
            \quad
            $U_2\in\RR^{\nout\times (\nout-r)}$
        }
        \\
        \proofgroup{Basic stochastic quantities}
        \proofsymcell{$\hatbsP, \hatbsP_N$, $\Delta\hatbs P_N$}
        &
        \proofdefcell{%
            $\hatbsP_N=\hatbsP
            =\bsP+N^{-1/2}\hatbsXi$,
            \qquad
            $\Delta\hatbs P_N = \hatbs P_N-\bs P$
            }
        \\
        \proofsymcell{$\hatbsS,\hatbs\ell$}
        &
        \proofdefcell{%
            $\hatbsS
            =\ntr^{-1}\hatbsP_N\hatbsP_N^T$,
            \quad
            $\hatbs\ell
            =\ntr^{-1}\hatbsP_N\bs y$
            \quad
            (\cref{eq:error_with_pseudoinverse})
        }
        \\
        \proofsymcell{$\Delta\hat o$}
        &
        \proofdefcell{%
            $\Delta\hat o=
            \bs y^T(\Delta\hatbs P_N^+)\bsp_\sigma
            =
            (\hatbs\ell^T
            \hatbsS^+
            - \bs\ell^T \bsS^+)\bsp_\sigma
            =
            \hatbs q^T\hatbsS_{\rm eff}^{-1}\bsp_{\sigma,1}
            $
            \quad
            (\cref{eq:error_with_pseudoinverse,eq:Deltao_with_q})
        }
        \\
        \proofgroup{Population quantities}
        \proofsymcell{$\bsS,\bs\ell$}
        &
        \proofdefcell{%
            $\bsS
            =\ntr^{-1}\bsP\bsP^T$,
            \quad
            $\bs\ell
            =\ntr^{-1}\bsP\bs y$
            \quad
            (\cref{eq:error_with_pseudoinverse})
        }
        \\

        \proofsymcell{$\overline{\bs\Sigma}$, $\overline{\bs\Sigma}_{ab}$, $\bs\Sigma_i$}
        &
        \proofdefcell{%
            $\overline{\bs\Sigma}
            =
            \ntr^{-1}\sum_{i=1}^{\ntr}\bs\Sigma_i$,
            \quad
            $\overline{\bs\Sigma}_{ab}
            =
            U_a^T \overline{\bs\Sigma} U_b$
            \quad
            $\bs\Sigma_i = \EE[\hatbs\xi_i\hatbs\xi_i^T]$
        }
        \\
        \proofgroup{Effective Schur-reduced quantities}
        \proofsymcell{$\bsS_{\rm eff}$}
        &
        \proofdefcell{%
            $\bsS_{\rm eff}
            =
            \bsS_{11}
            +N^{-1}\overline{\bs\Sigma}_{\rm eff}$
            \quad(\cref{eq:block_of_Delta_definitions})
        }
        \\
        \proofsymcell{$\hatbsS_{\rm eff}$}
        &
        \proofdefcell{%
            $\hatbsS_{\rm eff}
            =
            \hatbsS_{11}
            -\hatbsS_{12}
            \hatbsS_{22}^{-1}
            \hatbsS_{21}
            = \bsS_{\rm eff} + (\Delta_c\hatbsS_{\rm eff})_{\poly}
            + (\Delta_c\hatbsS_{\rm eff})_{\rm nl}$
            \quad(\cref{eq:decomposition_DeltaSeff})
        }
        \\
        \proofsymcell{$\overline{\bs\Sigma}_{\rm eff}$}
        &
        \proofdefcell{%
            $\overline{\bs\Sigma}_{\rm eff}
            =
            \overline{\bs\Sigma}_{11}
            -\overline{\bs\Sigma}_{12}
             \overline{\bs\Sigma}_{22}^{-1}
             \overline{\bs\Sigma}_{21}$}
        \\
        \proofsymcell{$\bs w^{\rm eff}_{\calO}$}
        &
        \proofdefcell{%
            $\bs w^{\rm eff}_{\calO}
            =
            (\bsI_{\nout} - U_2 \overline{\bs\Sigma}_{22}^{-1}\overline{\bs\Sigma}_{21} U_1^T)\bs w_\calO.$
            \quad(\cref{eq:def_weff})
        }
        \\
        \proofsymcell{$\hatbs\ell_{\rm eff}$}
        &
        \proofdefcell{%
            $\hatbs\ell_{\rm eff}
            =
            \hatbs\ell_1 - \hatbsS_{12}\hatbsS_{22}^{-1}\hatbs\ell_2
            = \bs\ell_1 + (\Delta_c 
            \hatbs\ell_{\rm eff})_{\poly} + (\Delta_c 
            \hatbs\ell_{\rm eff})_{\rm nl}$
            (\cref{eq:block_of_Delta_definitions,eq:def_ellterms})
        }
        \\
        \proofsymcell{$\hatbs q$}
        &
        \proofdefcell{%
            $\hatbs q \equiv
            \hatbs\ell_{\rm eff} - \hatbsS_{\rm eff}
            U_1^T \bs w_\calO
            =
            \bs q_0 + \hatbs q_{\poly} + \hatbs q_{\rm nl}$
            \quad(\cref{eq:def_q})
        }
        \\
        \bottomrule
    \end{tabular}
    \caption{Summary of main definitions.
    All quantities are conditional on the fixed training states.
    Stochastic quantities are denoted with a hat, and stochasticity is always understood with respect to the training noise realizations.
    The symbols $d$, $\ntr$, $\nout$, and $N$ denote the Hilbert-space dimension, number of training states, number of POVM outcomes, and training statistics, respectively.
    Expressions involving stochastic inverses are understood on events where the displayed inverses exist; the pseudoinverse formula for $\Delta\hat o$ is defined for every realization.
    }
    \label{tab:direct_proof_geometry}
\end{table}

\section{Homogeneous Gaussian noise model}
\label{app:concentration}

This section evaluates the random coefficients in the perturbative
prediction error under the homogeneous Gaussian reference model used in
\cref{sec:gaussian_approximation}.  We first derive the conditional laws of
the generalized Schur complement and of the term containing $\hatbs Q$.
We then use these laws to compute the mean and variance of the two
coefficients in \cref{eq:A1_A2_definitions}.  Conditional on the training
states, we assume
\begin{equation}
    \hatbsXi
    \sim
    \calM\calN(0,\overline{\bs\Sigma},\bs I_{\ntr}),
    \qquad
    \overline{\bs\Sigma}
    \equiv
    \frac{1}{\ntr}\sum_{i=1}^{\ntr}\bs\Sigma_i.
    \label{eq:homogeneous_gaussian_SM}
\end{equation}
Thus the columns are independent Gaussian vectors with a common covariance.
Note that Gaussian sampling and covariance homogeneity are separate assumptions: a large $N$ naturally motivates Gaussian column fluctuations but does not make the state-dependent multinomial covariances equal.
The fact that the homogeneous Gaussian model reproduces performances so well is thus due to specific features of the MSE, not to the empirical probability matrices for multinomial and homogeneous Gaussian models actually being close.

We keep $d$ and $\nout$ fixed and let $\ntr\to\infty$, and assume the
deterministic conditioning bounds
\begin{equation}
    \lambda_{\min}(\overline{\bs\Sigma}_{22})
    \ge\kappa_{22}>0,
    \qquad
    \lambda_{\min}(\bsS_{11})
    \ge\kappa_{11}>0,
    \qquad
    \bsS_{11}=\frac{\bsP_{11}\bsP_{11}^T}{\ntr},
    \label{eq:gaussian_uniform_conditioning}
\end{equation}
uniformly in $\ntr$.  The first bound concerns the covariance of the noise
in the kernel block.  The second concerns the noiseless feature matrix and
is the frame-conditioning assumption used below.  Since the eigenvalues of
$\bsP_{11}\bsP_{11}^T$ are the nonzero squared singular values of $\bsP$,
the frame-operator identity in \cref{app:frame} gives the sufficient bound
\begin{equation}
    \lambda_{\min}(\bsS_{11})
    \ge
    \lambda_{\min}(\frameOp_{\bsmu})
    \lambda_{\min}\!\left(\frac{\frameOp_{\bs\rho}}{\ntr}\right).
    \label{eq:gaussian_frame_sufficient_condition}
\end{equation}
Thus a fixed informationally complete POVM together with a family of
training frames satisfying
$\lambda_{\min}(\frameOp_{\bs\rho}/\ntr)>0$ uniformly in $\ntr$ implies the
second condition in \cref{eq:gaussian_uniform_conditioning}.  Completeness
of every finite training set alone would not be sufficient, since its lower
frame bound could still tend to zero.  For the balanced tight designs of
\cref{eq:frameop_symmetric_states,eq:frameop_symmetric_POVMs}, one obtains
explicitly
\begin{equation}
    \lambda_{\min}(\bsS_{11})
    =\frac{1}{\nout(d+1)^2},
    \label{eq:tight_gaussian_frame_bound}
\end{equation}
which is independent of $\ntr$.  The bound also holds with probability
tending to one for independent Haar-random training states.  Indeed, the
rank-one superoperators
$\PP(\rho_i)$ are uniformly bounded and the operator space has fixed
dimension $d^2$.  The entrywise law of large numbers and norm equivalence in
fixed dimension therefore give, in operator norm,
\begin{equation}
    \frac{\frameOp_{\bs\rho}}{\ntr}
    =
    \frac{\Id+\PP(\bs I)}{d(d+1)}+o_{\PP}(1).
    \label{eq:Haar_normalized_frame_limit}
\end{equation}
The smallest eigenvalue of the limiting operator is $1/[d(d+1)]$.
Consequently, for example,
$\lambda_{\min}(\frameOp_{\bs\rho}/\ntr)
\ge1/[2d(d+1)]$ with probability tending to one.  Together with
\cref{eq:gaussian_frame_sufficient_condition}, this proves the required
uniform lower-frame bound for Haar-random training designs with probability
tending to one. The block dimensions
and the identity $r=d^2$ are those of \cref{subsec:block_decompositions}.

\subsection{Conditional Gaussian laws}
\label{app:QY_control}

The Gaussian model is useful only if it resolves the nonlinear inverse appearing in the perturbative coefficients.
We do this by holding the lower blocks fixed: the upper-block residuals then have Gaussian laws, and orthogonal projections separate the Schur term from the kernel contribution.
Assuming $\barSigma_{22}\succ0$ and $\ntr\ge\nout$ ensures the lower block has full row rank almost surely; averaging its inverse will require the stronger threshold stated below.

Recall the two quantities entering these conditional laws:
\[
    \overline{\bs\Sigma}_{\rm eff}
    =\overline{\bs\Sigma}_{11}
    -\overline{\bs\Sigma}_{12}
     \overline{\bs\Sigma}_{22}^{-1}
     \overline{\bs\Sigma}_{21},
    \qquad
    \hatbsZ
    =\bs I_r
    +\hatbsXi_{21}^T
     (\hatbsXi_{22}\hatbsXi_{22}^T)^{-1}
     \hatbsXi_{21}.
\]
The first is the covariance Schur complement used above, and the second is
the matrix introduced in \cref{eq:cool_expr_MM22}.  Moreover,
$\hatbsXi_{22}=\overline{\bs\Sigma}_{22}^{1/2}\hatbs G$, where
$\hatbs G\in\RR^{(\nout-r)\times(\ntr-r)}$ has independent standard normal
entries.  Hence $\hatbsXi_{22}$ has full row rank almost surely when
$\ntr\ge\nout$.  It follows that
$\hatbsXi_{22}^+\hatbsXi_{22}$ is the orthogonal projector onto the row
space of $\hatbsXi_{22}$, whereas
\begin{equation}
    \hatbs Q
    =\bs I_{\ntr-r}-\hatbsXi_{22}^+\hatbsXi_{22}
    \label{eq:gaussian_Q_projector}
\end{equation}
is the orthogonal projector onto $\ker(\hatbsXi_{22})$.  Since
$\hatbsXi_{22}$ has $\nout-r$ rows and $\ntr-r$ columns, the rank-nullity
theorem gives $\rank(\hatbs Q)=\ntr-\nout$.  Conditioning on the lower
blocks $(\hatbsXi_{21},\hatbsXi_{22})$ fixes both $\hatbs Q$ and $\hatbsZ$,
while leaving the Gaussian fluctuations in the upper blocks to be
averaged.

\begin{lemma}[Conditional laws of the Schur and $Q$ terms]
\label{lem:gaussian_conditional_block_laws}
Assume \cref{eq:homogeneous_gaussian_SM} and $\ntr\ge\nout$.  Conditional
on $(\hatbsXi_{21},\hatbsXi_{22})$,
\begin{equation}
    \hatbsXi/\hatbsXi_{22}
    \sim
    \calM\calN(0,\overline{\bs\Sigma}_{\rm eff},\hatbsZ),
    \label{eq:conditional_schur_law}
\end{equation}
whereas
$\hatbsXi_{12}\hatbs Q\hatbsXi_{12}^T$ has a Wishart distribution with
scale $\overline{\bs\Sigma}_{\rm eff}$ and $\ntr-\nout$ degrees of
freedom.  The two random matrices in this statement are conditionally
independent.  Consequently,
\begin{gather}
    \EE\!\left[
        \hatbsXi_{12}\hatbs Q\hatbsXi_{12}^T
        \mathrel{\big|}\hatbsXi_{21},\hatbsXi_{22}
    \right]
    =(\ntr-\nout)\overline{\bs\Sigma}_{\rm eff},
    \label{eq:gaussian_Q_conditional_mean}\\
    \EE[\hatbsZ]
    =
    \frac{\ntr-r-1}{\ntr-\nout-1}\,\bs I_r,
    \qquad \ntr>\nout+1.
    \label{eq:gaussian_Y_mean}
\end{gather}
\end{lemma}

\begin{proof}
The conditional-normal formula, applied separately to the first $r$ and
last $\ntr-r$ columns, gives \cite{eaton2007wishart}:
\begin{align*}
    \hatbsXi_{11}\mid\hatbsXi_{21}
    &\sim
    \calM\calN\!\left(
        \overline{\bs\Sigma}_{12}
        \overline{\bs\Sigma}_{22}^{-1}\hatbsXi_{21},
        \overline{\bs\Sigma}_{\rm eff},\bs I_r
    \right),\\
    \hatbsXi_{12}\mid\hatbsXi_{22}
    &\sim
    \calM\calN\!\left(
        \overline{\bs\Sigma}_{12}
        \overline{\bs\Sigma}_{22}^{-1}\hatbsXi_{22},
        \overline{\bs\Sigma}_{\rm eff},\bs I_{\ntr-r}
    \right).
\end{align*}
For a jointly Gaussian pair, the residual after subtracting the conditional
mean is independent of the variable on which one conditions.  Hence
\[
    \hatbsXi_{12}
    -\overline{\bs\Sigma}_{12}
     \overline{\bs\Sigma}_{22}^{-1}\hatbsXi_{22}
\]
is independent of $\hatbsXi_{22}$ and has conditional mean zero.  This can
be seen directly, column by column, from
\[
    \Cov\!\left(
        \hatbsXi_{12}
        -\overline{\bs\Sigma}_{12}
         \overline{\bs\Sigma}_{22}^{-1}\hatbsXi_{22},
        \hatbsXi_{22}
    \right)
    =\overline{\bs\Sigma}_{12}
     -\overline{\bs\Sigma}_{12}
      \overline{\bs\Sigma}_{22}^{-1}
      \overline{\bs\Sigma}_{22}
    =0:
\]
zero cross-covariance implies independence because the pair is jointly
Gaussian.  The
analogous statement holds for the residual formed from
$\hatbsXi_{11}$ and $\hatbsXi_{21}$.  Moreover, the two residuals are
independent of each other because the first $r$ columns of $\hatbsXi$ are
independent of the last $\ntr-r$ columns.  Consequently, even after
conditioning on both lower blocks, the second term below has conditional
mean
\begin{align*}
    \EE\!\left[
        \left(
            \hatbsXi_{12}
            -\overline{\bs\Sigma}_{12}
             \overline{\bs\Sigma}_{22}^{-1}\hatbsXi_{22}
        \right)
        \hatbsXi_{22}^+\hatbsXi_{21}
        \mathrel{\big|}\hatbsXi_{21},\hatbsXi_{22}
    \right]=
    \EE\!\left[
        \hatbsXi_{12}
        -\overline{\bs\Sigma}_{12}
         \overline{\bs\Sigma}_{22}^{-1}\hatbsXi_{22}
        \mathrel{\big|}\hatbsXi_{22}
    \right]
    \hatbsXi_{22}^+\hatbsXi_{21}
    =0.
\end{align*}
Here $\hatbsXi_{22}^+\hatbsXi_{21}$ is fixed under the conditioning.

The conditional mean of the generalized Schur complement can also be
checked directly:
\begin{align*}
    \EE\!\left[
        \hatbsXi/\hatbsXi_{22}
        \mathrel{\big|}\hatbsXi_{21},\hatbsXi_{22}
    \right]=
    \overline{\bs\Sigma}_{12}
    \overline{\bs\Sigma}_{22}^{-1}\hatbsXi_{21}
    -\overline{\bs\Sigma}_{12}
     \overline{\bs\Sigma}_{22}^{-1}\hatbsXi_{22}
     \hatbsXi_{22}^+\hatbsXi_{21}
    =0,
\end{align*}
where the last equality uses
$\hatbsXi_{22}\hatbsXi_{22}^+=\bs I_{\nout-r}$.  Equivalently, subtracting
the two conditional means gives the centered representation
\begin{align*}
    \hatbsXi/\hatbsXi_{22}
    ={}
    \left(
        \hatbsXi_{11}
        -\overline{\bs\Sigma}_{12}
         \overline{\bs\Sigma}_{22}^{-1}\hatbsXi_{21}
    \right)-\left(
        \hatbsXi_{12}
        -\overline{\bs\Sigma}_{12}
         \overline{\bs\Sigma}_{22}^{-1}\hatbsXi_{22}
    \right)
    \hatbsXi_{22}^+\hatbsXi_{21}.
\end{align*}
The first centered matrix in the last display has column covariance
$\bs I_r$; the second is independent of it and, after right multiplication by
$\hatbsXi_{22}^+\hatbsXi_{21}$, has column covariance
$(\hatbsXi_{22}^+\hatbsXi_{21})^T
 (\hatbsXi_{22}^+\hatbsXi_{21})$.  Therefore the column covariance of their
difference is
\[
    \bs I_r
    +(\hatbsXi_{22}^+\hatbsXi_{21})^T
      (\hatbsXi_{22}^+\hatbsXi_{21})
    =\bs I_r
      +
      \hatbsXi_{21}^T
      (\hatbsXi_{22}\hatbsXi_{22}^T)^{-1}
      \hatbsXi_{21}
    =\hatbsZ.
\]
The row covariance remains $\overline{\bs\Sigma}_{\rm eff}$, which proves
\cref{eq:conditional_schur_law}.

The same centered representation also exposes the $\hatbs Q$ term.  Since
$\hatbsXi_{22}\hatbs Q=0$ and $\hatbs Q^2=\hatbs Q$, we have
\begin{align*}
    \hatbsXi_{12}\hatbs Q\hatbsXi_{12}^T
    ={}
    \left(
        \hatbsXi_{12}
        -\overline{\bs\Sigma}_{12}
         \overline{\bs\Sigma}_{22}^{-1}\hatbsXi_{22}
    \right)
    \hatbs Q
    \left(
        \hatbsXi_{12}
        -\overline{\bs\Sigma}_{12}
         \overline{\bs\Sigma}_{22}^{-1}\hatbsXi_{22}
    \right)^T.
\end{align*}
In an orthonormal basis that diagonalizes $\hatbs Q$, the nonzero part of
the centered matrix in this display consists of $\rank(\hatbs Q)
=\ntr-\nout$ independent Gaussian columns, each with covariance
$\overline{\bs\Sigma}_{\rm eff}$.  Their sum of outer products is therefore
Wishart with the scale and degrees of freedom stated in the lemma.

It remains to justify the claimed conditional independence.  In the Schur
complement, the contribution from the centered $\hatbsXi_{12}$ block is
right-multiplied by $\hatbsXi_{22}^+\hatbsXi_{21}$ and hence depends only on
its projection onto the row space of $\hatbsXi_{22}$.  By contrast, the
Wishart matrix in the preceding display depends only on its projection onto
$\ker(\hatbsXi_{22})$.  These two projections are orthogonal, since
\[
    (\hatbsXi_{22}^+\hatbsXi_{22})\hatbs Q=0.
\]
The centered $\hatbsXi_{12}$ block is matrix normal with identity column
covariance, so its two orthogonal projections have zero cross-covariance
and, being jointly Gaussian, are independent.  The centered
$\hatbsXi_{11}$ block appearing in the Schur complement is independent of
both because it consists of different columns.  This proves the conditional
independence of the two matrices in the lemma.  Finally, the mean of a
Wishart matrix is its number of degrees of freedom times its scale, which
gives \cref{eq:gaussian_Q_conditional_mean}.

The upper-block calculation is complete, but the coefficient averages still contain $\hatbsZ$.
Its mean is the point at which integrability of the random lower inverse enters.
For \cref{eq:gaussian_Y_mean}, the two lower blocks are independent and
$\hatbsXi_{22}\hatbsXi_{22}^T$ is Wishart with scale
$\overline{\bs\Sigma}_{22}$ and $\ntr-r$ degrees of freedom.  Conditional
on $\hatbsXi_{22}$, the $r$ columns of $\hatbsXi_{21}$ are independent with
covariance $\overline{\bs\Sigma}_{22}$.  Consequently,
\begin{align*}
    \EE\!\left[
        \hatbsZ-\bs I_r
        \mathrel{\big|}\hatbsXi_{22}
    \right]=
    \Trace\!\left[
        (\hatbsXi_{22}\hatbsXi_{22}^T)^{-1}
        \overline{\bs\Sigma}_{22}
    \right]\bs I_r.
\end{align*}
The inverse-Wishart identity
\[
    \EE[(\hatbsXi_{22}\hatbsXi_{22}^T)^{-1}]
    =\frac{\overline{\bs\Sigma}_{22}^{-1}}
    {\ntr-\nout-1}\quad \Rightarrow\quad
    \EE[\hatbsZ-\bs I_r]
    =\frac{\nout-r}{\ntr-\nout-1}\bs I_r.
\]
Here the denominator is
$(\ntr-r)-(\nout-r)-1=\ntr-\nout-1$, and the numerator comes from
$\Trace(\overline{\bs\Sigma}_{22}^{-1}
\overline{\bs\Sigma}_{22})=\nout-r$.
Adding $\bs I_r$ and combining the two terms gives
$(\ntr-r-1)/(\ntr-\nout-1)$, as in
\cref{eq:gaussian_Y_mean}; see also
\cite[Chapter~3]{muirhead2009aspects}.
\end{proof}

\subsection{Leading resource scalings}
\label{subsec:gaussian_scaling_recovery}

The truncated infinite statistics MSE in \cref{eq:delta_P_11} is expressed in terms of the scalar contractions $\bs w_{\calO}^TU_1\bs A_kV_1^T\langle\tilde{\bs\rho},\sigma\rangle$, $k=1,2$, with prefactors $N^{-1/2}$ and $N^{-1}$, where
$\bs A_1=-\hatbsXi/\hatbsXi_{22}$ and
$\bs A_2=(\hatbsXi/\hatbsXi_{22})\bsP_{11}^{-1}(\hatbsXi/\hatbsXi_{22})-\hatbsXi_{12}\hatbs Q\hatbsXi_{12}^T\bsP_{11}^{-T}\hatbsZ$.
Their means, variances, and covariance determine the squared bias and the training variance.
Conditional on the lower blocks, \cref{lem:gaussian_conditional_block_laws} makes the Schur complement centered matrix-normal and $\hatbsXi_{12}\hatbs Q\hatbsXi_{12}^T$ an independent Wishart matrix, so each of these moments is a Gaussian or Wishart moment at fixed $\hatbsZ$, subsequently averaged over the lower blocks.
\Cref{lem:gaussian_A1_A2_moments} below provides such means and variances, depending on the training frame, the test state, and the observable only through the two deterministic coefficients $b_{\rm hom}$ and $V_{\rm hom}$ given in \cref{eq:gaussian_b_V_identification}.

Although both coefficients carry an explicit factor $\ntr$ in their definition, they are actually $O(1)$ with respect to the number of training states. In fact, since $\bsS_{11}=\bsP_{11}\bsP_{11}^T/\ntr$, that condition implies $\sigma_{\min}(\bsP_{11})^2\ge\kappa_{11}\ntr$, and hence
\begin{equation}
    \|\bsP_{11}^{-1}\|_{\rm op}
    =\frac{1}{\sigma_{\min}(\bsP_{11})}
    \le\frac{1}{\sqrt{\kappa_{11}\ntr}}
    =O(\ntr^{-1/2}),
    \qquad
    \|V_1^T\langle\tilde{\bs\rho},\sigma\rangle\|_2
    \le
    \frac{\|\bsp_{\sigma,1}\|_2}{\sqrt{\kappa_{11}\ntr}}
    \le\frac{1}{\sqrt{\kappa_{11}\ntr}}
    =O(\ntr^{-1/2}),
    \label{eq:gaussian_frame_scaling}
\end{equation}
because
$V_1^T\langle\tilde{\bs\rho},\sigma\rangle
=\bsP_{11}^{-1}\bsp_{\sigma,1}$ and
$\|\bsp_{\sigma,1}\|_2\le\|\bsp_\sigma\|_2\le1$.  The last inequality
uses that $\bsp_\sigma$ is a probability vector. In this notation, the coefficients introduced in \cref{sec:gaussian_approximation} read:
\begin{equation}
\begin{aligned}
    b_{\rm hom}
    &\equiv
    \ntr\,\bs w_{\calO}^TU_1\overline{\bs\Sigma}_{\rm eff}
    \bsP_{11}^{-T}V_1^T
    \langle\tilde{\bs\rho},\sigma\rangle,\\
    V_{\rm hom}
    &\equiv
    \ntr\,
    \|V_1^T\langle\tilde{\bs\rho},\sigma\rangle\|_2^2
    \bs w_{\calO}^TU_1\overline{\bs\Sigma}_{\rm eff}U_1^T
    \bs w_{\calO}.
\end{aligned}
\label{eq:gaussian_b_V_identification}
\end{equation}
The explicit factor $\ntr$ in $b_{\rm hom}$ is canceled by the two
$O(\ntr^{-1/2})$ factors just identified.  In $V_{\rm hom}$, it is canceled by the
squared norm of the dual-frame vector.  Since $d$ and $\nout$ are fixed,
the remaining contractions are bounded under the stated assumptions, and
hence $b_{\rm hom},V_{\rm hom}=O(1)$.

\begin{lemma}[Moments of the two perturbative coefficients]
\label{lem:gaussian_A1_A2_moments}
For $\ntr>\nout+3$, the contractions appearing in
\cref{eq:delta_P_11} satisfy
\begin{align}
    \EE\!\left[
        \bs w_{\calO}^TU_1\bs A_1
        V_1^T\langle\tilde{\bs\rho},\sigma\rangle
    \right]
    &=0,
    \label{eq:gaussian_A1_mean}\\
    \Var\!\left[
        \bs w_{\calO}^TU_1\bs A_1
        V_1^T\langle\tilde{\bs\rho},\sigma\rangle
    \right]
    &=
    \frac{V_{\rm hom}}{\ntr}
    \frac{\ntr-r-1}{\ntr-\nout-1}
    =\frac{V_{\rm hom}}{\ntr}+O(\ntr^{-2}),
    \label{eq:gaussian_A1_variance}\\
    \EE\!\left[
        \bs w_{\calO}^TU_1\bs A_2
        V_1^T\langle\tilde{\bs\rho},\sigma\rangle
    \right]
    &=-\left(1-\frac{r+1}{\ntr}\right)b_{\rm hom}
    =-b_{\rm hom}+O(\ntr^{-1}),
    \label{eq:gaussian_A2_mean}\\
    \Var\!\left[
        \bs w_{\calO}^TU_1\bs A_2
        V_1^T\langle\tilde{\bs\rho},\sigma\rangle
    \right]
    &=O(\ntr^{-1}).
    \label{eq:gaussian_A2_variance}
\end{align}
The covariance between the two displayed contractions is zero.
The condition $\ntr>\nout+3$ ensures that the second inverse-Wishart
moment entering \cref{eq:gaussian_A2_variance} is finite; the mean formulas
only require $\ntr>\nout+1$.
In particular, the two contractions are respectively
$O_{\PP}(\ntr^{-1/2})$ and
$-b+O_{\PP}(\ntr^{-1/2})$.
\end{lemma}

\begin{proof}
By \cref{eq:A1_A2_definitions}, $\bs A_1$ is minus the generalized Schur
complement.  Its conditional law in \cref{eq:conditional_schur_law} is
centered, which proves \cref{eq:gaussian_A1_mean} after averaging over the
lower blocks.  The matrix-normal covariance rule gives, under the same
conditioning,
\begin{align*}
    \Var\!\left[
        \bs w_{\calO}^TU_1\bs A_1
        V_1^T\langle\tilde{\bs\rho},\sigma\rangle
        \mathrel{\big|}\hatbsXi_{21},\hatbsXi_{22}
    \right]=
    \bs w_{\calO}^TU_1
    \overline{\bs\Sigma}_{\rm eff}
    U_1^T\bs w_{\calO}\,
    \bigl(V_1^T\langle\tilde{\bs\rho},\sigma\rangle\bigr)^T
    \hatbsZ
    \bigl(V_1^T\langle\tilde{\bs\rho},\sigma\rangle\bigr).
\end{align*}
Averaging this identity and inserting \cref{eq:gaussian_Y_mean} replaces
$\hatbsZ$ by a scalar multiple of $\bs I_r$.  The definition of $V$ in
\cref{eq:gaussian_b_V_identification} then gives
\cref{eq:gaussian_A1_variance}.

For the second coefficient, condition again on
$(\hatbsXi_{21},\hatbsXi_{22})$.  If a matrix is matrix normal with row
covariance $\overline{\bs\Sigma}_{\rm eff}$ and column covariance
$\hatbsZ$, its entries obey
\[
    \EE\!\left[
        (\hatbsXi/\hatbsXi_{22})_{ia}
        (\hatbsXi/\hatbsXi_{22})_{bj}
        \mathrel{\big|}\hatbsXi_{21},\hatbsXi_{22}
    \right]
    =
    (\overline{\bs\Sigma}_{\rm eff})_{ib}(\hatbsZ)_{aj}.
\]
Contracting the two inner indices therefore gives
\[
    \EE\!\left[
        (\hatbsXi/\hatbsXi_{22})\bsP_{11}^{-1}
        (\hatbsXi/\hatbsXi_{22})
        \mathrel{\big|}\hatbsXi_{21},\hatbsXi_{22}
    \right]
    =\overline{\bs\Sigma}_{\rm eff}\bsP_{11}^{-T}\hatbsZ.
\]
On the other hand, the conditional Wishart law gives
\[
    \EE\!\left[
        \hatbsXi_{12}\hatbs Q\hatbsXi_{12}^T
        \bsP_{11}^{-T}\hatbsZ
        \mathrel{\big|}\hatbsXi_{21},\hatbsXi_{22}
    \right]
    =(\ntr-\nout)\overline{\bs\Sigma}_{\rm eff}
      \bsP_{11}^{-T}\hatbsZ.
\]
The matrix $\hatbsZ$ is fixed under the conditioning, so subtracting the
two conditional means according to \cref{eq:A1_A2_definitions} gives
\begin{align*}
    \EE[\bs A_2\mid\hatbsXi_{21},\hatbsXi_{22}]
    =[1-(\ntr-\nout)]
      \overline{\bs\Sigma}_{\rm eff}
      \bsP_{11}^{-T}\hatbsZ=-(\ntr-\nout-1)
      \overline{\bs\Sigma}_{\rm eff}
      \bsP_{11}^{-T}\hatbsZ.
\end{align*}
Now \cref{eq:gaussian_Y_mean} cancels the denominator
$\ntr-\nout-1$ and yields
\[
    \EE[\bs A_2]
    =-(\ntr-r-1)\overline{\bs\Sigma}_{\rm eff}\bsP_{11}^{-T},
\]
which is \cref{eq:gaussian_A2_mean} after contraction and use of
\cref{eq:gaussian_b_V_identification}: dividing $\ntr-r-1$ by the factor
$\ntr$ contained in $b$ produces $1-(r+1)/\ntr$.  This calculation also
shows explicitly how the dependence on $\nout$ cancels between the
Wishart degrees of freedom and the inverse-Wishart mean of $\hatbsZ$.

We next spell out the variance estimate.  The second inverse-Wishart moment,
which is finite for $\ntr>\nout+3$, implies
\[
    \EE[\|\hatbsZ\|_{\rm op}^2]=O(1),
    \qquad
    \EE[\|\hatbsZ-\bs I_r\|_{\rm op}^2]=O(\ntr^{-2}).
\]
Together with \cref{eq:gaussian_frame_scaling}, the vector multiplying the
Wishart matrix on its right satisfies the second-moment bound
\[
    \EE\!\left[
    \left\|
        \bsP_{11}^{-T}\hatbsZ
        V_1^T\langle\tilde{\bs\rho},\sigma\rangle
    \right\|_2^2
    \right]
    =O(\ntr^{-2}).
\]
In fixed dimension, a Wishart matrix with $\ntr-\nout$ degrees of freedom
has variance proportional to $\ntr-\nout$ in every fixed scalar
contraction.  After the preceding right contraction, its centered part
therefore contributes
$O(\ntr\cdot\ntr^{-2})=O(\ntr^{-1})$ to the variance.  The fluctuation of
its conditional mean caused by $\hatbsZ-\bs I_r$ contributes only
$O(\ntr^{-2})$.  For the first, Schur-complement term of $\bs A_2$, the
conditional fourth Gaussian moment is bounded because
$\EE[\|\hatbsZ\|_{\rm op}^2]=O(1)$.  Its two deterministic contractions in
\cref{eq:gaussian_frame_scaling} supply a total factor $O(\ntr^{-1})$ in
amplitude, so its second moment is $O(\ntr^{-2})$.  The Wishart fluctuation
is therefore the leading random contribution, and
\cref{eq:gaussian_A2_variance} follows.

Finally, condition on the lower blocks once more.  The $\bs A_1$
contraction is linear in the centered Gaussian Schur complement, whereas
the first term of the $\bs A_2$ contraction is quadratic in that same
centered matrix.  Their product has odd total degree and hence zero
conditional expectation.  The second term of $\bs A_2$ depends on the
orthogonal Gaussian projection that produces the Wishart matrix; by
\cref{lem:gaussian_conditional_block_laws} it is independent of the Schur
complement.  Since the conditional mean of the $\bs A_1$ contraction is
zero, this product also has zero conditional expectation.  Averaging over
the lower blocks proves that the two contractions have zero covariance.
\end{proof}

The two terms retained in \cref{eq:delta_P_11} therefore obey
\begin{align}
\EE\!\left[
\left\{
\bs w_{\calO}^TU_1
\left(\frac{\bs A_1}{\sqrt N}+\frac{\bs A_2}{N}\right)
V_1^T\langle\tilde{\bs\rho},\sigma\rangle
\right\}^{\!2}
\right]\nonumber&=
\left(1-\frac{r+1}{\ntr}\right)^2\frac{b_{\rm hom}^2}{N^2}
+\frac{V_{\rm hom}}{N\ntr}
 \frac{\ntr-r-1}{\ntr-\nout-1}
+O\!\left(\frac{1}{N^2\ntr}\right)\nonumber\\
&=
\frac{b_{\rm hom}^2}{N^2}+\frac{V_{\rm hom}}{N\ntr}
+O\!\left(\frac{1}{N^2\ntr}+\frac{1}{N\ntr^2}\right).
\label{eq:gaussian_recovered_resource_scaling}
\end{align}
In the first equality, the term proportional to $b^2$ is the squared mean
of the $N^{-1}\bs A_2$ contribution, while the term proportional to $V$ is
the variance of the centered $N^{-1/2}\bs A_1$ contribution.  The
$O(N^{-2}\ntr^{-1})$ term contains the variance of the $\bs A_2$
contraction, and no mixed term remains because the two contractions have
zero covariance.  The second equality uses fixed $d$ and $\nout$ and
bounded $b,V$: both rational prefactors differ from one by
$O(\ntr^{-1})$.  This moment expansion concerns the first two coefficients
of the fixed-realization expansion in $N^{-1/2}$.  Standard inverse-Wishart
moment bounds ensure integrability of every coefficient at any prescribed
fixed order once $\ntr-\nout$ is sufficiently large.  Uniform control of the
expansion remainder, as well as the result for the true multinomial model,
is established in \cref{app:direct_remainder}.

\subsection{Target-noise corrections to the test and training variances}
\label{subsec:gaussian_target_noise_variances}

We now compute the two target-noise corrections introduced in
\cref{eq:target_noise_test_variance_exact,eq:target_noise_training_variance_exact}.
The calculation uses the homogeneous Gaussian model
\cref{eq:homogeneous_gaussian_SM} and the block coordinates of
\cref{subsec:block_decompositions}.  We write
\begin{equation}
    \bs\Sigma_{\sigma,22}=U_2^T\bs\Sigma_\sigma U_2,
    \qquad
    \bsp_{\sigma,1}=U_1^T\bsp_\sigma.
    \label{eq:target_noise_test_blocks}
\end{equation}

\begin{lemma}[Gaussian target-noise corrections]
\label{lem:gaussian_target_noise_corrections}
Assume \cref{eq:homogeneous_gaussian_SM,eq:gaussian_uniform_conditioning},
$\ntr>\nout+3$, and the independent Gaussian target-noise model in
\cref{eq:target_noise_model}.  For every fixed physical test vector
$\bsp_\sigma\in\im(\bsP)$, the large-$N$ corrections are
\begin{align}
    \Delta\Var_{\rm test}^{(\eta)}(N,M,\sigma)
    &=
    \frac{N\gamma^2}{M}
    \frac{
        \Trace(\overline{\bs\Sigma}_{22}^{-1}
        \bs\Sigma_{\sigma,22})
    }{\ntr-\nout-1}
    +O\!\left(\frac{\gamma^2\sqrt N}{M}\right),
    \label{eq:gaussian_target_test_variance_fixed_sigma}
    \\
    \Delta\Var_{\rm tr}^{(\eta)}(N,\sigma)
    &=
    \gamma^2
    \frac{\ntr-r-1}{\ntr-\nout-1}
    \|\bsP_{11}^{-1}\bsp_{\sigma,1}\|_2^2
    +O\!\left(\frac{\gamma^2}{\sqrt N}\right).
    \label{eq:gaussian_target_training_variance}
\end{align}
The leading coefficients themselves require only $\ntr>\nout+1$; the
stronger condition stated above ensures that the moments entering the
remainders are finite.  The $O(\cdot)$ terms are taken at fixed
$\bsP,\ntr$, and $\nout$.  If the test ensemble satisfies
$\EE_\sigma[\bs\Sigma_\sigma]=\overline{\bs\Sigma}$, then
\begin{equation}
    \EE_\sigma[
        \Delta\Var_{\rm test}^{(\eta)}(N,M,\sigma)
    ]
    =\frac{N\gamma^2}{M}
     \frac{\nout-r}{\ntr-\nout-1}
     +O\!\left(\frac{\gamma^2\sqrt N}{M}\right).
    \label{eq:gaussian_target_test_variance_averaged}
\end{equation}
\end{lemma}

\begin{proof}
We first separate the target fluctuations from the feature fluctuations.
For any deterministic matrix $\bs A$ of compatible size, the conditional
Gaussian covariance in \cref{eq:target_noise_model} gives
\begin{equation}
    \EE_{\hatbs\eta}
    [\hatbs\eta^T\bs A\hatbs\eta\mid\hatbsXi]
    =\gamma^2\Trace(\bs A).
    \label{eq:gaussian_target_quadratic_identity}
\end{equation}
Taking
$\bs A=\hatbsP_N^+\bs\Sigma_\sigma(\hatbsP_N^+)^T/M$
gives the last line of
\cref{eq:target_noise_test_variance_exact}.  Taking instead
$\bs A=\hatbsP_N^+\bsp_\sigma\bsp_\sigma^T(\hatbsP_N^+)^T$
and applying the law of total variance gives
\cref{eq:target_noise_training_variance_exact}.  The cross terms with the
exact-target estimator vanish because
$\EE[\hatbs\eta\mid\hatbsXi]=0$.

We now insert the large-$N$ block expansion of the pseudoinverse already used
in \cref{sec:perturbation}.  In particular,
\cref{eq:Delt_PN_small,eq:expansionPp_other} give
\begin{equation}
    \hatbsP_N^+
    =\sqrt N\,V_2\hatbsXi_{22}^+U_2^T+O(1).
    \label{eq:target_noise_pinv_singular_block}
\end{equation}
Substitution in \cref{eq:target_noise_test_variance_exact} yields
\begin{equation}
    \Delta\Var_{\rm test}^{(\eta)}(N,M,\sigma)
    =\frac{N\gamma^2}{M}
      \EE_{\hatbsXi}\Trace[
        \hatbsXi_{22}^+\bs\Sigma_{\sigma,22}
        (\hatbsXi_{22}^+)^T
        +O\!\left(\frac{\gamma^2\sqrt N}{M}\right).
        \label{eq:target_noise_test_leading_block}
\end{equation}
The remainder comes from the cross products between the singular term in
\cref{eq:target_noise_pinv_singular_block} and the $O(1)$ blocks.  Its
expectation is finite under the moment condition in the lemma.

To compute the leading coefficient, use the Gaussian factorization from
\cref{app:QY_control},
\begin{equation}
    \hatbsXi_{22}
    =\overline{\bs\Sigma}_{22}^{1/2}\hatbs G,
    \qquad
    \hatbs G\in\RR^{(\nout-r)\times(\ntr-r)},
    \label{eq:target_noise_G_factorization}
\end{equation}
where the entries of $\hatbs G$ are independent standard normal variables.
The matrix $\hatbs G$ has full row rank almost surely and
\begin{equation}
\Trace[
        \hatbsXi_{22}^+\bs\Sigma_{\sigma,22}
        (\hatbsXi_{22}^+)^T
    ]=
    \Trace\!\left[
        \overline{\bs\Sigma}_{22}^{-1/2}
        \bs\Sigma_{\sigma,22}
        \overline{\bs\Sigma}_{22}^{-1/2}
        (\hatbs G\hatbs G^T)^{-1}
    \right].
\label{eq:target_noise_test_wishart_reduction}
\end{equation}
The matrix $\hatbs G\hatbs G^T$ is Wishart with $\ntr-r$ degrees of freedom
and identity scale.  Its inverse mean is~\cite{muirhead2009aspects}
\begin{equation}
    \EE[(\hatbs G\hatbs G^T)^{-1}]
    =\frac{\bsI_{\nout-r}}{\ntr-\nout-1}.
    \label{eq:target_noise_inverse_wishart_mean}
\end{equation}
Substitution in \cref{eq:target_noise_test_leading_block} proves
\cref{eq:gaussian_target_test_variance_fixed_sigma}.  Averaging over
$\sigma$ replaces $\bs\Sigma_{\sigma,22}$ by
$\overline{\bs\Sigma}_{22}$, whose contraction with its inverse has trace
$\nout-r$, and proves
\cref{eq:gaussian_target_test_variance_averaged}.

For the training variance, the right multiplication by the physical vector
$\bsp_\sigma=U_1\bsp_{\sigma,1}$ projects out the singular $(2,2)$ block.
Defining $\bs a=\bsP_{11}^{-1}\bsp_{\sigma,1}$, the $(1,1)$ and $(2,1)$
blocks of \cref{eq:expansionPp11,eq:expansionPp_other} give
\begin{equation}
    \hatbsP_N^+\bsp_\sigma
    =V_1\bs a
     -V_2\hatbsXi_{22}^+\hatbsXi_{21}\bs a
     +O(N^{-1/2}).
    \label{eq:target_noise_physical_vector_expansion}
\end{equation}
Since the two displayed terms belong to the orthogonal subspaces spanned by
$V_1$ and $V_2$, respectively, taking the squared norm and averaging gives
\begin{equation}
\begin{aligned}
    \EE_{\hatbsXi}\|\hatbsP_N^+\bsp_\sigma\|_2^2
    &=\|\bs a\|_2^2
    +\EE_{\hatbsXi}
      \|\hatbsXi_{22}^+\hatbsXi_{21}\bs a\|_2^2
      +O(N^{-1/2}).
    \label{eq:target_noise_training_two_terms}
\end{aligned}
\end{equation}
The remaining average is already contained in the matrix $\hatbsZ$ introduced
in \cref{eq:cool_expr_MM22}.  Indeed,
\begin{equation}
    \|\hatbsXi_{22}^+\hatbsXi_{21}\bs a\|_2^2
    =\bs a^T(\hatbsZ-\bsI_r)\bs a.
    \label{eq:target_noise_training_Y_identity}
\end{equation}
Using the exact Gaussian mean \cref{eq:gaussian_Y_mean} in
\cref{eq:target_noise_training_Y_identity} gives
\begin{equation}
    \EE_{\hatbsXi}
    \|\hatbsXi_{22}^+\hatbsXi_{21}\bs a\|_2^2
    =\frac{\nout-r}{\ntr-\nout-1}\|\bs a\|_2^2.
    \label{eq:target_noise_training_Y_mean}
\end{equation}
Substituting this identity in \cref{eq:target_noise_training_two_terms} yields
\begin{equation}
\begin{aligned}
    \EE_{\hatbsXi}\|\hatbsP_N^+\bsp_\sigma\|_2^2
    &=\left(
        1+\frac{\nout-r}{\ntr-\nout-1}
      \right)\|\bs a\|_2^2
      +O(N^{-1/2})
    \\
    &=\frac{\ntr-r-1}{\ntr-\nout-1}
      \|\bsP_{11}^{-1}\bsp_{\sigma,1}\|_2^2
      +O(N^{-1/2}),
\end{aligned}
\label{eq:target_noise_training_wishart_result}
\end{equation}
which proves \cref{eq:gaussian_target_training_variance} after multiplication
by $\gamma^2$.
\end{proof}

For comparison, if the feature matrix is known exactly, then
$\bsP^+\bsp_\sigma=V_1\bsP_{11}^{-1}\bsp_{\sigma,1}$ and target noise adds
\begin{equation}
    \Delta\Var_{\rm tr,exact}^{(\eta)}(\sigma)
    =\gamma^2\|\bsP_{11}^{-1}\bsp_{\sigma,1}\|_2^2.
    \label{eq:target_noise_exact_feature_training_variance}
\end{equation}
At leading order, the rank-increasing feature noise therefore multiplies the
exact-feature answer by
$(\ntr-r-1)/(\ntr-\nout-1)$.  The frame bound in
\cref{eq:gaussian_frame_scaling} gives
$\|\bsP_{11}^{-1}\bsp_{\sigma,1}\|_2^2=O(\ntr^{-1})$, so this contribution
is $O(\gamma^2/\ntr)$ for fixed $d$ and $\nout$.  This is to be contrasted
with \cref{eq:gaussian_target_test_variance_averaged}, whose unnormalized
test-variance contribution is
$O(N\gamma^2\nout/(M\ntr))$ when $\ntr\gg\nout$.

\section{Multinomial noise model: bias, linear fluctuation, and remainder}
\label{app:direct_remainder}

This section contains the proof of the MSE and resource-scaling result stated in \cref{sec:asymptotic_bias}, and is logically independent of the formal small-$\tau$ expansion in \cref{app:perturbative_P}.
Both analyses use the exact wide-matrix identity $\hatbsP^+=\hatbsP^T(\hatbsP\hatbsP^T)^{-1}$ when $\hatbsP$ has full row rank, and the block inverse formula \cref{eq:M-1}, but none of the truncated expansions \cref{eq:expansionPp11,eq:expansionPp_other} is used below.
Instead, we first eliminate the rank-increasing directions by an exact Schur complement and only then expand the resulting well-conditioned $d^2\times d^2$ matrix.
This order of operations enables a uniform high-probability bound on the remainder in the $\ntr\to\infty$ regime, even for small $N$.

\parTitle{Basic assumptions}
Throughout this proof, $d$ and $\nout$ are fixed, the targets and test probabilities are exact, and all expectations are conditional on the training states.
Each training state is measured with $N\ge 1$ independent categorical draws, independently across states.
No Gaussian or homogeneous-covariance approximation is used.
We assume $r\equiv\rank(\bsP)=d^2<\nout\le\ntr$ and consider families of training sets for which there exist constants $\kappa_{11},\kappa_{22}>0$, independent of $N$ and $\ntr$, such that
\begin{equation}\label{eq:direct_uniform_conditioning}
    \lambda_{\min}\!\left[
        \frameOp_{\bsmu}^{1/2}
        \frac{\frameOp_{\bs\rho}}{\ntr}
        \frameOp_{\bsmu}^{1/2}
    \right]
    \ge \kappa_{11},
    \qquad
    \bs v^T\overline{\bs\Sigma}\bs v
    \ge
    \kappa_{22}\|\bs v\|_2^2
    \quad
    \text{for all }\bs v\in\ker(\bsMmu),
\end{equation}
where
$\overline{\bs\Sigma}\equiv\ntr^{-1}\sum_{i=1}^{\ntr}\bs\Sigma_i$
and $\bs\Sigma_i=\EE[\hatbs\xi_i\hatbs\xi_i^T]$.
The first condition ensures that the training states are well-conditioned and that even in the large $\ntr$ limit there are no operator directions that become increasingly less covered by the training states.
The second condition requires instead the averaged shot-noise covariance to remain nondegenerate on the redundant outcome directions.
A sufficient criterion to guarantee this is a uniform lower bound on the average training state: because $\bs v^T\bs p_i=0$ for $v\in\ker(\bsMmu)$,
\begin{equation}
    \bs v^T\overline{\bs\Sigma}\bs v
    =\frac{1}{\ntr}\sum_{i=1}^{\ntr}
    \bs v^T\on{diag}(\bs p_i)\bs v
    = \sum_{a=1}^{\nout}
    v_a^2 \langle\mu_a,\bar\rho\rangle
    \ge c (\min_a \tr(\mu_a))\|\bs v\|_2^2,
\end{equation}
where $\bar\rho\equiv\ntr^{-1}\sum_{i=1}^{\ntr} \rho_i$.
Hence one may take $c=\lambda_{\rm min}(\bar\rho)$.
A uniform $\kappa_{22}>0$ follows if $\lambda_{\rm min}(\bar\rho)$ stays bounded away from zero and the fixed POVM has no zero effects. On the other hand, $\bar\rho>0$ is sufficient to guarantee $\kappa_{22}>0$ but not its independence from $\nout,\ntr$.
We also use the elementary uniform bound $|y_i|=|\langle\calO,\rho_i\rangle|\le\|\calO\|_\infty$ in the concentration estimates, with $B_y\equiv\max_i|y_i|$ introduced explicitly in \cref{eq:def_good_event}.

The two main results that will be proved in this section are the following:
\begin{theorem}
Under the standing assumptions above, there exist constants
$u_0\in(0,1]$ and $K_{\rm good},K_{\rm bad}>0$, independent of $N$ and $\ntr$, such that the following holds for every $N\ge1$ and $\ntr\ge\nout/u_0$.
We have
\begin{equation}
\label{eq:opening_error_bias_definitions}
\begin{gathered}
    \Delta\hat o\equiv \bs y^T(\hatbsP_N^+-\bsP^+)\bsp_\sigma
    = -\frac{b_N}{N} + \hat L + \hat R,
    \qquad
    b_N\equiv
    \bs w_\calO^TU_1\overline{\bs\Sigma}_{\rm eff}
    \bsS_{\rm eff}^{-1}\bsp_{\sigma,1},
    \\
    \hat L \equiv 
    (\Delta_c\hatbs\ell_{\rm eff})_{\poly}^T
    \bsS_{\rm eff}^{-1} \bsp_{\sigma,1}
    -
    \bs\ell_1^T
    \bsS_{\rm eff}^{-1}
    (\Delta_c\hatbsS_{\rm eff})_{\poly}
    \bsS_{\rm eff}^{-1}
    \bsp_{\sigma,1},
\end{gathered}
\end{equation}
and $\hat R\equiv\Delta\hat o+b_N/N-\hat L$.
These satisfy $\EE[\hat L]=0$, and
\begin{equation}
\begin{gathered}
    \EE[\hat L^2]
    = \frac{V_N}{N \ntr} + O\Big(\frac{1}{N^2\ntr}\Big),
    \qquad
    V_N = \frac{1}{\ntr} \sum_{i=1}^{\ntr}
    (\bs p_{i,1}^T
    \bsS_{\rm eff}^{-1} \bsp_{\sigma,1})^2
    \Var_i(\bs w^{\rm eff}_{\calO}),
    \qquad
    \Var_i(\bs w^{\rm eff}_{\calO})
    \equiv
    (\bs w^{\rm eff}_{\calO})^T
    \bs\Sigma_i
    \bs w^{\rm eff}_{\calO}.
\end{gathered}
\end{equation}
The remainder satisfies
\begin{equation}
\label{eq:opening_total_remainder_bound}
    \EE[\hat R^2]\le \mathcal R^{(2)}_{N,\ntr}
    \equiv
    \frac{K_{\rm good}\nout^2}{N\ntr^2}
    +16K_{\rm bad}\ntr(N\nout)^{2\nout}e^{-u_0\ntr+\nout}.
\end{equation}
Consequently,
\begin{equation}
\label{eq:opening_direct_moment_bounds}
\begin{gathered}
    \left|\EE[\Delta\hat o]+\frac{b_N}{N}\right|
    \le \sqrt{\mathcal R^{(2)}_{N,\ntr}},
    \qquad
    \left|\Var(\Delta\hat o)-\EE[\hat L^2]\right|
    \le 2\sqrt{\EE[\hat L^2]\mathcal R^{(2)}_{N,\ntr}}
        +\mathcal R^{(2)}_{N,\ntr},
    \\
    \left|\EE[\Delta\hat o^2]-\frac{b_N^2}{N^2}-\EE[\hat L^2]\right|
    \le 2\left(\frac{|b_N|}{N}+\sqrt{\EE[\hat L^2]}\right)
        \sqrt{\mathcal R^{(2)}_{N,\ntr}}
        +\mathcal R^{(2)}_{N,\ntr}.
\end{gathered}
\end{equation}
Thus the deterministic floor and the exact first-order fluctuation are
controlled at arbitrary shot count, without a Gaussian approximation.
\end{theorem}

\begin{theorem}
    Assume in addition that $N,\ntr\to\infty$ and $\ntr/\log(2+N)\to\infty$, with $d$ and $\nout$ fixed.
    Use $\Delta\hat o$ and $b_N$ as defined in \cref{eq:opening_error_bias_definitions}.
    Then
    \begin{equation}
    \label{eq:opening_direct_joint_mse}
        \EE[\Delta\hat o^2]
        =\frac{b_N^2}{N^2}+\frac{V_N}{N\ntr}
        +O\!\left(
            \frac{1}{N^2\ntr}
            +\frac{\nout}{N^{3/2}\ntr}
            +\frac{\nout}{N\ntr^{3/2}}
            +\frac{\nout^2}{N\ntr^2}
        \right)
        =\frac{b_N^2}{N^2}+\frac{V_N}{N\ntr}
        +o\!\left(N^{-2}+(N\ntr)^{-1}\right).
    \end{equation}
    The first term is the squared finite-shot population shift and the second
    is the training fluctuation.  If $|b_N|$ and $V_N$ stay bounded above and
    away from zero, their crossover occurs at
    $\ntr\asymp NV_N/b_N^2$.
\end{theorem}

\parTitle{Section outline}
In~\cref{subsec:deterministicdecomposition}, we prove a decomposition $\Delta\hat o=B+\hat L + \hat R$ of the infinite-test-statistics error $\Delta\hat o$ into a deterministic component $B$, a centered fluctuation $\hat L$, and a remainder $\hat R$.
We identify the deterministic bias component $B$ in~\cref{subsec:direct_bias}.
In \cref{subsec:direct_linear_fluctuation} we derive the form of the fluctuation term $\hat L$, show that it has zero average, and compute its variance.
We then prove in \cref{subsec:direct_remainder_bound} that the remainder $\hat R$ is subleading, conditionally on a stochastic event where the pseudoinverses are sufficiently well-conditioned.
In \cref{sec:highprobbounds} we then prove that this event has high probability for large $\ntr$.
To conclude the bound on the MSE we only need to make sure that the stochastic events in which pseudoinverses are ill-conditioned do not make the MSE explode.
We do this in \cref{subsec:direct_mse_bound}, where we show that the divergence of $\hat R$ in the exceptional realizations does not overcome the exponentially small probability of these events. This concludes the proof and provides us with an unconditional MSE bound for large $\ntr$.
A summary of the the notation used in the proof is given in~\cref{tab:direct_proof_geometry}.

\subsection{Deterministic decomposition of infinite-test-statistics error}
\label{subsec:deterministicdecomposition}

We begin by providing an exact decomposition of the estimation error with infinite-test-statistics and exact targets, $\Delta\hat o=\bs y^T(\hatbsP_N^+-\bsP^+)\bsp_\sigma$, as $\Delta\hat o=B+\hat L+\hat R$, with $B$ a deterministic bias, $\hat L$ a centered fluctuation, and $\hat R$ a remainder term that is subleading with high proability for large $\ntr$.

\parTitle{Block decomposition notation}
Let $U=(U_1,U_2)$ be orthogonal, with
$\im(U_1)=\im(\bsP)$ and
$\im(U_2)=\ker(\bsP^T)=\im(\bsP)^\perp$.
We use this basis to separate the $r=d^2$ directions already present in
the noiseless data from the $\nout-r$ directions generated by shot noise.
Write
$\bsP_1\equiv U_1^T\bsP$, $\bsP_2=0$,
$U_a^T\hatbsP=\hatbsP_a$, and
$U_1^T\bsp_\sigma=\bsp_{\sigma,1}$.
Thus $\bsP_1,\hatbsP_1\in\RR^{r\times\ntr}$ and
$\hatbsP_2\in\RR^{(\nout-r)\times\ntr}$.
We denote the $i$-th column of $\bsP_1$ by $\bs p_{i,1}$ and set
$\hatbsXi_a\equiv U_a^T\hatbsXi$, with columns
$\hatbs\xi_{i,a}\equiv U_a^T\hatbs\xi_i$.

\parTitle{Basic rewriting of $\Delta\hat o$}
We study the infinite-test-statistics estimation error
$\Delta\hat o=\bs y^T(\hatbsP_N^+-\bsP^+)\bsp_\sigma$, where
$\bs y^T=\langle\calO,\bs\rho\rangle$ is the row of training targets and
$\bsp_\sigma=\langle\bs\mu,\sigma\rangle$ is the test-state probability
vector.  We assume
$r\equiv\rank(\bsP)=d^2<\nout\le\ntr$ and, as in the main text,
$\bs y\in\im(\bsP^T)$ and $\bsp_\sigma\in\im(\bsP)$.
The minimal informationally complete case $\nout=d^2$ has no cokernel
block and is handled by the ordinary inverse expansion on
$\im(\bsP)$; the Schur-complement analysis below concerns the
rank-increasing case \(\nout>d^2\).
The identity $\bs A^+=\bs A^T(\bs A\bs A^T)^+$ gives
\begin{equation}
\label{eq:error_with_pseudoinverse}
    \Delta\hat o =
    \hatbs\ell^T
    \hatbsS^+ \bsp_\sigma
    - \bs\ell^T \bsS^+ \bsp_\sigma,
    \qquad
    \hatbs\ell\equiv \frac{\hatbsP\bs y}{\ntr},
    \quad
    \bs\ell\equiv \frac{\bsP\bs y}{\ntr},
    \quad
    \hatbsS \equiv \frac{\hatbsP\hatbsP^T}{\ntr},
    \quad
    \bsS \equiv \frac{\bsP\bsP^T}{\ntr}.
\end{equation}
Here $\hatbsP\equiv\hatbsP_N=\bsP+N^{-1/2}\hatbsXi$.
In this section we omit the $N$ subscript on stochastic quantities for notational brevity.
Conditional on the training
states, the columns of $\hatbsXi$ are independent and centered, with
$\EE[\hatbs\xi_i\hatbs\xi_i^T]=\bs\Sigma_i$.  Consequently,
\begin{equation}
    \EE[\hatbs\ell]=\bs\ell,
    \qquad
    \EE[\hatbsS]=\bsS+\frac{1}{N}\overline{\bs\Sigma},
    \qquad
    \overline{\bs\Sigma}\equiv
    \frac{1}{\ntr}\sum_{i=1}^{\ntr}\bs\Sigma_i.
\end{equation}
The deterministic $N^{-1}\overline{\bs\Sigma}$ shift is the source of
the leading bias.  Directly expanding $\hatbsS^{-1}$ is nevertheless
unhelpful: the eigenvalues created outside $\im(\bsP)$ are of order
$N^{-1}$, so that $\|\EE[\hatbsS]^{-1}\|_{\rm op}$ is of order $N$.
The Schur complement below removes these apparently singular directions before the fluctuation expansion is made.
In this notation, the assumptions in \cref{eq:direct_uniform_conditioning} can be restated as
\begin{equation}\label{eq:direct_uniform_conditioning_2}
\lambda_{\min}(\bsS_{11})\ge\kappa_{11}>0,
\qquad
\lambda_{\min}(\overline{\bs\Sigma}_{22})\ge\kappa_{22}>0,
\end{equation}
with $\kappa_{11},\kappa_{22}$ independent of $N$ and $\ntr$.

\parTitle{Additional notation to get a finer-grained decomposition $\Delta\hat o$}
We use
\[
    \hatbs\ell_a=\ntr^{-1}\hatbsP_a\bs y,\qquad
    \bs\ell_a=\ntr^{-1}\bsP_a\bs y,\qquad
    \hatbsS_{ab}=U_a^T\hatbsS U_b,\qquad
    \bsS_{ab}=U_a^T\bsS U_b,
\]
where $\bsP_2=0$, and hence
$\bs\ell_2=\bsS_{12}=\bsS_{21}=\bsS_{22}=0$.  When $\hatbsS_{22}$ and its Schur complement are invertible, we have
\begin{equation}
\begin{gathered}
    \hatbs\ell^T\hatbsS^{-1}\bsp_\sigma
    =
    \frac{1}{\ntr}(\hatbsP\bs y)^T \hatbsS^{-1}\bsp_\sigma
    =
    \hatbs\ell_{\rm eff}^T
    \hatbsS_{\rm eff}^{-1}
    \bsp_{\sigma,1},
    \qquad
    \hatbs\ell_{\rm eff} \equiv
    \hatbs\ell_1 - \hatbsS_{12}\hatbsS_{22}^{-1}\hatbs\ell_2,
    \qquad
    \hatbsS_{\rm eff} \equiv
    \hatbsS_{11}
    -
    \hatbsS_{12} \hatbsS_{22}^{-1} \hatbsS_{21},
\end{gathered}
\end{equation}
Observe that 
$\bs y=\bsP^T\bs w_{\calO}$ and $\bs\ell_1=\bsS_{11}U_1^T\bs w_{\calO}$.
If $\overline{\bs\Sigma}_{22}$ is invertible, define
$\overline{\bs\Sigma}_{\rm eff}
=\overline{\bs\Sigma}_{11}
-\overline{\bs\Sigma}_{12}
\overline{\bs\Sigma}_{22}^{-1}
\overline{\bs\Sigma}_{21}$.
Introduce
\begin{equation}\label{eq:block_of_Delta_definitions}
\begin{alignedat}{3}
    \bsS_{\rm eff}
    &\equiv
    \EE[\hatbsS_{11}]
    -\EE[\hatbsS_{12}]
     \EE[\hatbsS_{22}]^{-1}
     \EE[\hatbsS_{21}]
    =\bsS_{11}
    +\frac{1}{N}\overline{\bs\Sigma}_{\rm eff},
    &\qquad
    \Delta\hatbsS_{ab}
    &\equiv
    \hatbsS_{ab}-\EE[\hatbsS_{ab}],
    &\qquad
    \Delta_c\hatbsS_{\rm eff}
    &\equiv
    \hatbsS_{\rm eff}-\bsS_{\rm eff},
    \\
    \bs\ell_{\rm eff}
    &\equiv
    \EE[\hatbs\ell_1]
    -\EE[\hatbsS_{12}]
     \EE[\hatbsS_{22}]^{-1}
     \EE[\hatbs\ell_2]
    =\bs\ell_1,
    &\qquad
    \Delta\hatbs\ell_a
    &\equiv
    \hatbs\ell_a-\bs\ell_a,
    &\qquad
    \Delta_c\hatbs\ell_{\rm eff}
    &\equiv
    \hatbs\ell_{\rm eff}-\bs\ell_{\rm eff}.
\end{alignedat}
\end{equation}
where $\Delta$ denotes a centered block fluctuation, whereas $\Delta_c$
denotes deviation from the deterministic concentration target.  Thus
$\EE[\Delta\hatbsS_{ab}]=0$, but in general
$\EE[\Delta_c\hatbsS_{\rm eff}]\ne0$ because the Schur complement is
nonlinear.

\parTitle{Rewriting the terms of $\Delta\hat o$ via effective quantiites}
Observing that $\hatbs\ell^T\hatbsS^{-1}\bs p_\sigma=\hatbs\ell_{\rm eff}^T\hatbsS_{\rm eff}^{-1}\bs p_{\sigma,1}$ and $\bs\ell^T\bsS^+\bs p_\sigma=\bs\ell_1^T\bsS_{11}^{-1}\bs p_{\sigma,1}$, and adding and subtracting $\bs\ell_1^T\bsS_{\rm eff}^{-1}\bsp_{\sigma,1}$ and $\hatbs\ell_{\rm eff}^T\bsS_{\rm eff}^{-1}\bsp_{\sigma,1}$ in \cref{eq:error_with_pseudoinverse}, we get
\begin{equation}\label{eq:exact_Deltao_decomp}
    \Delta\hat o =
    \bs\ell_1^T[
    \bsS_{\rm eff}^{-1} - \bsS_{11}^{-1}
    ]\bsp_{\sigma,1}
    +
    \hatbs\ell_{\rm eff}^T
    (\hatbsS_{\rm eff}^{-1} - \bsS_{\rm eff}^{-1})\bsp_{\sigma,1}
    + \Delta_c\hatbs\ell_{\rm eff}^T
    \bsS_{\rm eff}^{-1} \bsp_{\sigma,1}.
\end{equation}
Applying the resolvent identity twice,
\begin{equation}
    \hatbs\ell_{\rm eff}^T
    (\hatbsS_{\rm eff}^{-1} - \bsS_{\rm eff}^{-1})\bsp_{\sigma,1}
    =
    -
    \hatbs\ell_{\rm eff}^T
    \bsS_{\rm eff}^{-1}
    \Delta_c\hatbsS_{\rm eff}
    \bsS_{\rm eff}^{-1}
    \bsp_{\sigma,1}
    + 
    \hatbs\ell_{\rm eff}^T
    \bsS_{\rm eff}^{-1}
    \Delta_c\hatbsS_{\rm eff}
    \bsS_{\rm eff}^{-1}
    \Delta_c\hatbsS_{\rm eff}
    \hatbsS_{\rm eff}^{-1}
    \bsp_{\sigma,1},
\end{equation}
We next decompose $\hatbs\ell_{\rm eff}$ in its deterministic, centered fluctuations, and nonlinear components, as:
\begin{equation}\label{eq:def_ellterms}
\begin{gathered}
    \hatbs\ell_{\rm eff}
    =
    \bs\ell_1
    +
    (\Delta_c\hatbs\ell_{\rm eff})_{\poly}
    +
    (\Delta_c\hatbs\ell_{\rm eff})_{\rm nl},
    \\
    (\Delta_c\hatbs\ell_{\rm eff})_{\poly}
    \equiv
    \Delta\hatbs\ell_1
    -
    \overline{\bs\Sigma}_{12}
    \overline{\bs\Sigma}_{22}^{-1}
    \Delta\hatbs\ell_2
    =
    \frac{1}{\ntr \sqrt N} (\hatbsXi_1-\overline{\bs\Sigma}_{12}
    \overline{\bs\Sigma}_{22}^{-1}
    \hatbsXi_2)
    \bs y
    =
    \frac{1}{\ntr\sqrt N}
    \sum_{i=1}^{\ntr}
    \hatbs\xi_{i,\rm eff} y_i,
\end{gathered}
\end{equation}
where $\hatbs\xi_{i,\rm eff}\equiv
\hatbs\xi_{i,1}-\overline{\bs\Sigma}_{12}
\overline{\bs\Sigma}_{22}^{-1}\hatbs\xi_{i,2}$ denotes the $i$-th
column of $\hatbsXi_{\rm eff} \equiv
\hatbsXi_1-\overline{\bs\Sigma}_{12}
\overline{\bs\Sigma}_{22}^{-1}\hatbsXi_2$.
Similarly,    
\begin{equation}%
\label{eq:decomposition_DeltaSeff}
\begin{gathered}
    \Delta_c\hatbsS_{\rm eff}
    =
    (\Delta_c\hatbsS_{\rm eff})_{\poly}
    +
    (\Delta_c\hatbsS_{\rm eff})_{\rm nl},
    \\
    (\Delta_c\hatbsS_{\rm eff})_{\poly}
    \equiv
    \Delta\hatbsS_{11}
    -
    \Delta\hatbsS_{12}
    \overline{\bs\Sigma}_{22}^{-1}
    \overline{\bs\Sigma}_{21}
    -
    \overline{\bs\Sigma}_{12}
    \overline{\bs\Sigma}_{22}^{-1}
    \Delta\hatbsS_{21}
    +
    \overline{\bs\Sigma}_{12}
    \overline{\bs\Sigma}_{22}^{-1}
    \Delta\hatbsS_{22}
    \overline{\bs\Sigma}_{22}^{-1}
    \overline{\bs\Sigma}_{21}
    \\
    =
    \frac{1}{\ntr\sqrt N}
    \sum_{i=1}^{\ntr}
    \left(
        \bs p_{i,1}\hatbs\xi_{i,\rm eff}^T
        +
        \hatbs\xi_{i,\rm eff}\bs p_{i,1}^T
    \right)
    +
    \frac{1}{N\ntr}
    \sum_{i=1}^{\ntr}
    \left(
        \hatbs\xi_{i,\rm eff}\hatbs\xi_{i,\rm eff}^T
        -
        \EE[\hatbs\xi_{i,\rm eff}\hatbs\xi_{i,\rm eff}^T]
    \right).
\end{gathered}
\end{equation}
Both first-order terms are centered:
$\EE[(\Delta_c\hatbs\ell_{\rm eff})_{\poly}]
=\EE[(\Delta_c\hatbsS_{\rm eff})_{\poly}]=0$.
Here the subscript ``$\poly$'' denotes terms that are linear or quadratic in $\hatbsXi$;
the last term of $(\Delta_c\hatbsS_{\rm eff})_{\poly}$ is an example of the latter.

\parTitle{Exact error decomposition in bias, fluctuation, and remainder terms}
Expanding~\cref{eq:exact_Deltao_decomp} with \cref{eq:def_ellterms,eq:decomposition_DeltaSeff}, we finally get
\begin{equation}%
\label{eq:def_Deltao_andL}
\begin{gathered}
    \Delta\hat o =
    B + \hat L + \hat R,
    \\
    B \equiv
    \bs\ell_1^T
    \left(
        \bsS_{\rm eff}^{-1}
        -
        \bsS_{11}^{-1}
    \right)
    \bsp_{\sigma,1},
    \qquad
    \hat L \equiv 
    (\Delta_c\hatbs\ell_{\rm eff})_{\poly}^T
    \bsS_{\rm eff}^{-1} \bsp_{\sigma,1}
    -
    \bs\ell_1^T
    \bsS_{\rm eff}^{-1}
    (\Delta_c\hatbsS_{\rm eff})_{\poly}
    \bsS_{\rm eff}^{-1}
    \bsp_{\sigma,1}.
\end{gathered}
\end{equation}
Here $B$ is deterministic, $\hat L$ is centered and first order in the centered block fluctuations, and $\hat R\equiv\Delta\hat o-B-\hat L$ will be shown to be subleading with high probability.
The exact statistical bias is thus $\EE[\Delta\hat o]=B+\EE[\hat R]$.

\subsection{Deterministic bias component}
\label{subsec:direct_bias}

From~\cref{eq:def_Deltao_andL}, remembering $\bsS_{\rm eff}=\bsS_{11}
+N^{-1}\overline{\bs\Sigma}_{\rm eff}$ and
$\bs\ell_1=\bsS_{11}U_1^T\bs w_{\calO}$, the resolvent identity gives
\begin{equation}
    B
    =
    \bs\ell_1^T
    \left(
        \bsS_{\rm eff}^{-1}
        -
        \bsS_{11}^{-1}
    \right)
    \bsp_{\sigma,1}
    =
    -\frac{1}{N}
    \bs w_{\calO}^T
    U_1
    \overline{\bs\Sigma}_{\rm eff}
    \bsS_{\rm eff}^{-1}
    \bsp_{\sigma,1}
    =
    -\frac{1}{N}
    \bs w_{\calO}^T
    U_1
    \overline{\bs\Sigma}_{\rm eff}
    \bsS_{11}^{-1}
    \bsp_{\sigma,1}
    +
    O\left(\frac{1}{N^2}\right).
\label{eq:bias_exact}
\end{equation}
We can therefore write the deterministic bias component as $B=-b_N/N$ with $b_N=\bs w_{\calO}^T U_1
\overline{\bs\Sigma}_{\rm eff} \bsS_{\rm eff}^{-1} \bsp_{\sigma,1}$.
We can further replace $\bsS_{\rm eff}^{-1}$ by $\bsS_{11}^{-1}$ when $\|\bsS_{11}^{-1/2}\overline{\bs\Sigma}_{\rm eff}\bsS_{11}^{-1/2}\|_{\rm op}\ll N$.

\subsection{Centered first-order fluctuation}
\label{subsec:direct_linear_fluctuation}

Having identified the deterministic shift in \cref{subsec:direct_bias}, we now compute the centered fluctuation that contributes the leading training variance.
The goal is to express it as a sum of independent column contributions and compute its second moment before controlling the nonlinear remainder in
\cref{subsec:direct_remainder_bound}.
From \cref{eq:def_Deltao_andL} we rewrite $\hat L$ as
\begin{equation}
    \hat L = \left[
    \frac{1}{\ntr\sqrt N}
    \bs w_\calO^T U_1 \bsP_1\hatbsXi_{\rm eff}^T
    -
    \bs\ell_1^T \bsS_{\rm eff}^{-1}
    (\Delta_c \hatbsS_{\rm eff})_{\poly}
    \right]
    \bsS_{\rm eff}^{-1}\bsp_{\sigma,1}.
\end{equation}
Using~\cref{eq:decomposition_DeltaSeff}, this amounts to
$
    \hat L
    =
    \frac{1}{\ntr}
    \sum_{i=1}^{\ntr}\hat L_i
$,
where
\begin{equation}
\begin{aligned}
    \hat L_i
    &=
    \left\{
    \frac{1}{\sqrt N}
    \left[
        y_i\hatbs\xi_{i,\rm eff}^T
        -
        \bs\ell_1^T \bsS_{\rm eff}^{-1}
        \left(
            \bs p_{i,1}\hatbs\xi_{i,\rm eff}^T
            +
            \hatbs\xi_{i,\rm eff}\bs p_{i,1}^T
        \right)
    \right]
    -
    \frac{1}{N}
    \bs\ell_1^T \bsS_{\rm eff}^{-1}
    \left(
        \hatbs\xi_{i,\rm eff}\hatbs\xi_{i,\rm eff}^T
        -
        \EE[\hatbs\xi_{i,\rm eff}\hatbs\xi_{i,\rm eff}^T]
    \right)
    \right\}
    \bsS_{\rm eff}^{-1}\bsp_{\sigma,1}
\label{eq:def_gi}
\end{aligned}
\end{equation}
The first linear term in~\cref{eq:def_gi} contains the coefficient
$y_i-\bs\ell_1^T\bsS_{\rm eff}^{-1}\bs p_{i,1}$.
This coefficient is itself $O(N^{-1})$.
Indeed,
\begin{equation}\label{eq:linear_coefficient_correction}
    \bsS_{\rm eff}^{-1}\bs\ell_1-U_1^T\bs w_\calO
    =
    -\frac{1}{N}\bsS_{\rm eff}^{-1}
    \overline{\bs\Sigma}_{\rm eff}U_1^T\bs w_\calO
    = O(N^{-1}).
\end{equation}
Thus this part of the nominally $N^{-1/2}$ fluctuation contributes only at subleading order.

The variables $\hat L_i$ are independent and centered, therefore
$\EE[\hat L]=0$ and
$\EE[\hat L^2]
=\frac{1}{\ntr^2}\sum_{i=1}^{\ntr}\EE[\hat L_i^2]$.
Let
$\bs w_{\calO}=\langle\bstildemu,\calO\rangle\in\im(\bsP)\subset \RR^{\nout}$ be the canonical dual-frame outcome weights for the observable $\calO$, and define
\begin{equation}\label{eq:def_weff}
    \bs w^{\rm eff}_{\calO}
    \equiv (\bs I_{\nout} - U_2 \overline{\bs\Sigma}_{22}^{-1}\overline{\bs\Sigma}_{21} U_1^T)\bs w_\calO.
\end{equation}
Interestingly, the two weight vectors $\bs w_\calO$ and $\bs w_\calO^{\rm eff}$ have parallel variational interpretations: $\bs w_\calO$ solves $\bsP^T\bs u=\langle\bs\rho,\calO\rangle$ with minimum Euclidean norm, whereas $\bs w^{\rm eff}_{\calO}$ solves the same linear system but minimizing the weighted quadratic form $\bs u^T\overline{\bs\Sigma}\bs u$.

Using $\bs w_\calO^T U_1 \hatbs\xi_{i,\rm eff}=(\bs w^{\rm eff}_{\calO})^T\hatbs\xi_i$, \cref{eq:def_gi} becomes
\begin{equation}
    \hat L_i
    =
    -\frac{1}{\sqrt N}
    (\bs p_{i,1}^T\bsS_{\rm eff}^{-1}\bsp_{\sigma,1})
    (\bs w^{\rm eff}_{\calO})^T\hatbs\xi_i
    +
    \hat r_i,
    \qquad
    \EE[\hat r_i^2]
    =
    O\left(\frac{1}{N^2}\right).
\end{equation}
When squaring~\cref{eq:def_gi}, one also encounters the cross term between
the $N^{-1/2}$ contribution linear in $\hatbs\xi_{i,\rm eff}$ and the
$N^{-1}$ centered quadratic contribution:
$
    -\frac{1}{N}
    \bs\ell_1^T\bsS_{\rm eff}^{-1}
    \left(
        \hatbs\xi_{i,\rm eff}\hatbs\xi_{i,\rm eff}^T
        -
        \EE[\hatbs\xi_{i,\rm eff}\hatbs\xi_{i,\rm eff}^T]
    \right)
    \bsS_{\rm eff}^{-1}\bsp_{\sigma,1}
$.
Its expectation is governed by the third moment of
$\hatbs\xi_{i,\rm eff}$.
To show that this is subleading, write $\hatbs\xi_i=N^{-1/2}\sum_{s=1}^N (\bs e_{Z_{i,s}}-\bs p_i)$, with independent categorical draws $Z_{i,s}$.
Independence and centering give the exact identity
$
    \EE[\hatbs\xi_i^{\otimes3}]
    =\frac{1}{\sqrt N}
    \EE[(\bs e_{Z_{i,1}}-\bs p_i)^{\otimes3}],
$
thus the apparent $N^{-3/2}$ prefactor in that cross term acquires a further factor $N^{-1/2}$ after averaging.
The same identity holds after applying the deterministic map $U_1^T-\overline{\bs\Sigma}_{12}\overline{\bs\Sigma}_{22}^{-1} U_2^T$ to each tensor factor, hence $\EE[\hatbs\xi_{i,\rm eff}^{\otimes 3}]=O(N^{-1/2})$ holds as well.
The fourth moments are uniformly bounded, as follows from the subgaussian estimate in~\cref{lem:multinomial_subgaussian_direct}.
Together with~\cref{eq:linear_coefficient_correction}, these observations give
\begin{equation}\label{eq:L_variance_leading}
\begin{gathered}
    \EE[\hat L^2]
    = \frac{V_N}{N \ntr} + O\Big(\frac{1}{N^2\ntr}\Big),
    \qquad
    V_N = \frac{1}{\ntr} \sum_{i=1}^{\ntr}
    (\bs p_{i,1}^T
    \bsS_{\rm eff}^{-1} \bsp_{\sigma,1})^2
    \Var_i(\bs w^{\rm eff}_{\calO}),
    \qquad
    \Var_i(\bs w^{\rm eff}_{\calO})
    \equiv
    (\bs w^{\rm eff}_{\calO})^T
    \bs\Sigma_i
    \bs w^{\rm eff}_{\calO}.
\end{gathered}
\end{equation}
The $O(N^{-2}\ntr^{-1})$ term is small relative to the $V_N/N\ntr$ only in the large-$N$ limit; at fixed $N$, both terms scale as $\ntr^{-1}$ and can be comparable.
For large $N$, we can furthermore write
\begin{equation}
    V_N = V + O(N^{-1}),
    \qquad
    V \equiv \frac{1}{\ntr} \sum_{i=1}^{\ntr}
    (\bs p_{i,1}^T
    \bsS_{11}^{-1} \bsp_{\sigma,1})^2
    \Var_i(\bs w^{\rm eff}_{\calO})
    =
    \ntr \sum_{i=1}^{\ntr}
    \langle\tilde\rho_i,\sigma\rangle^2
    \Var_i(\bs w^{\rm eff}_{\calO}),
\end{equation}
using the identity $\bs p_{i,1}^T\bsS_{11}^{-1}\bsp_{\sigma,1}=\ntr \langle\tilde\rho_i,\sigma\rangle$, where $\tilde\rho_i\equiv \frameOp_\rho^{-1}(\rho_i)$ (see \cref{app:frame} for the relevant frame notation).

\subsection{Deterministic bound on the remainder}
\label{subsec:direct_remainder_bound}

We derived in \cref{subsec:direct_bias,subsec:direct_linear_fluctuation} the explicit expressions for $B$ and $\hat L$, but their interpretation as leading bias and variance terms depends on the remainder being sufficiently small.
To control $\hat R$ we will here first provide a deterministic upper bound on $\lvert\hat R\rvert$ conditionally on an event in which empirical moments are sufficiently close to their concentration targets, for arbitrary $N\ge 1$.
We will show in \cref{sec:highprobbounds} that this event does indeed happen with high probability for large $\ntr$.
The main result of this subsection is \cref{eq:R_bound}.

\parTitle{Residual and exact reminder}
Define
\begin{equation}%
\label{eq:def_q}
\begin{gathered}
    \hatbs q \equiv
    \hatbs\ell_{\rm eff} - \hatbsS_{\rm eff}U_1^T\bs w_\calO
    =
    \bs q_0 + \hatbs q_{\poly} + \hatbs q_{\rm nl},
    \qquad
    \bs q_0 = \bs\ell_1 - \bsS_{\rm eff}U_1^T\bs w_\calO = -\frac{1}{N} \overline{\bs\Sigma}_{\rm eff}U_1^T\bs w_\calO,
    \\
    \hatbs q_{\poly} \equiv
    (\Delta_c\hatbs\ell_{\rm eff})_{\poly} - (\Delta_c\hatbsS_{\rm eff})_{\poly}U_1^T\bs w_\calO,
    \qquad
    \hatbs q_{\rm nl} \equiv
    (\Delta_c\hatbs\ell_{\rm eff})_{\rm nl} - (\Delta_c\hatbsS_{\rm eff})_{\rm nl}U_1^T\bs w_\calO.
\end{gathered}
\end{equation}
Since $\bs\ell_1=\bs S_{11} U_1^T\bs w_\calO$ and $\bs\ell_1^T\bsS_{11}^{-1}\bsp_{\sigma,1}
=\bs w_{\calO}^T U_1\bsp_{\sigma,1}$, we can rewrite \cref{eq:exact_Deltao_decomp} as
\begin{equation}\label{eq:Deltao_with_q}
    \Delta\hat o
    =
    \hatbs\ell_{\rm eff}^T \hatbs S_{\rm eff}^{-1} \bs p_{\sigma,1}
    - \bs\ell_1^T \bsS_{11}^{-1}\bs p_{\sigma,1}
    =
    (\hatbs\ell_{\rm eff}-\hatbsS_{\rm eff} U_1^T\bs w_\calO)^T
    \hatbsS_{\rm eff}^{-1}\bs p_{\sigma,1}
    =
    \hatbs q^T
    \hatbsS_{\rm eff}^{-1}
    \bsp_{\sigma,1}.
\end{equation}
This rewriting preserves the cancellation between $\hatbs\ell_{\rm eff}$ and $\hatbsS_{\rm eff}U_1^T\bs w_\calO$ instead of bounding these two
$O(1)$ quantities separately.
Their difference contains the deterministic $O(N^{-1})$ term $\bs q_0$ and the smaller fluctuation terms controlled below.
In particular, the residual need not tend to zero as $\ntr\to\infty$ at fixed $N$, because $\bs q_0$ need not vanish.
In this notation,
\begin{equation}
    B=\bs q_0^T\bsS_{\rm eff}^{-1}\bsp_{\sigma,1},
    \qquad
    \hat L=
    \left[
        \hatbs q_{\poly}
        -(\Delta_c\hatbsS_{\rm eff})_{\poly}
        \bsS_{\rm eff}^{-1}\bs q_0
    \right]^T
    \bsS_{\rm eff}^{-1}\bsp_{\sigma,1}.
\end{equation}
Applying the resolvent identity
$
    (\bs A+\bs F)^{-1}
    =\bs A^{-1}-\bs A^{-1}\bs F\bs A^{-1}
    +(\bs A^{-1}\bs F)^2(\bs A+\bs F)^{-1}
$
with $\bs A=\bsS_{\rm eff}$ and
$\bs F=\Delta_c\hatbsS_{\rm eff}$ then yields
\begin{equation}%
\label{eq:remainder_with_q}
    \hat R =
    - \bs q_0^T\bsS_{\rm eff}^{-1}
    (\Delta_c\hatbsS_{\rm eff})_{\rm nl}
    \bsS_{\rm eff}^{-1}\bsp_{\sigma,1}
    - (\hatbs q - \bs q_0)^T
    \bsS_{\rm eff}^{-1}
    \Delta_c\hatbsS_{\rm eff}
    \bsS_{\rm eff}^{-1}\bsp_{\sigma,1}
    + \hatbs q^T
    (\bsS_{\rm eff}^{-1}\Delta_c\hatbsS_{\rm eff})^2
    \hatbsS_{\rm eff}^{-1}\bsp_{\sigma,1}
    +
    \hatbs q_{\rm nl}^T
    \bsS_{\rm eff}^{-1}\bsp_{\sigma,1}.
\end{equation}
Each term now contains either two centered fluctuations or one nonlinear Schur-complement remainder.

\parTitle{Conditioning event}
The remainder identity in \cref{eq:remainder_with_q} is only useful if we can control the inverse factors in it.
By the standing assumptions in~\cref{eq:direct_uniform_conditioning,eq:direct_uniform_conditioning_2} and $\overline{\bs\Sigma}\ge0$, we have 
$
    \overline{\bs\Sigma}_{\rm eff}
    =
    \overline{\bs\Sigma}_{11}
    -
    \overline{\bs\Sigma}_{12}
    \overline{\bs\Sigma}_{22}^{-1}
    \overline{\bs\Sigma}_{21}
    \ge0
$ (because $\bs x^T\overline{\bs\Sigma}_{\rm eff}\bs x = (\bs x^T,-\bs x^T \bs A^T)\overline{\bs\Sigma}(\bs x,-\bs A\bs x)$ with $\bs A\equiv \overline{\bs\Sigma}_{22}^{-1}\overline{\bs\Sigma}_{21}$), and
$\bsS_{\rm eff}=\bsS_{11}+\frac{1}{N} \overline{\bs\Sigma}_{\rm eff}\ge\bsS_{11} \ge \kappa_{11}>0$.
It will be useful to also define the term
\begin{equation}
    \hatbs h\equiv
    \Delta\hatbsS_{21}U_1^T \bs w_{\calO}-\Delta\hatbs\ell_2
     =
    \frac{1}{N\ntr}
    \sum_{i=1}^{\ntr}
    \left(
        \hatbs\xi_{i,2}\hatbs\xi_{i,1}^T
        -
        \EE[
            \hatbs\xi_{i,2}\hatbs\xi_{i,1}^T
        ]
    \right)
    U_1^T\bs w_{\calO}.
    \label{eq:def_h}
\end{equation}
This identity holds because the $N^{-1/2}$ term in $\Delta\hatbsS_{21}U_1^T \bs w_{\calO}$ is $\hatbsXi_2\bsP_1^TU_1^T\bs w_\calO/(\ntr\sqrt N)$, which cancels with $\Delta\hatbs\ell_2$.
The remaining vector $\hatbs h$ has the smaller, quadratic fluctuation scale $N^{-1}$.
For $u>0$, define the ``good event'':
\begin{equation}\label{eq:def_good_event}
    \calE_u=\left\{
    \begin{aligned}
    &\|\Delta\hatbsS_{11}\|_{\rm op},\ \|\Delta\hatbsS_{12}\|_{\rm op},\ \|\Delta\hatbsS_{21}\|_{\rm op}\le u/\sqrt N,
    &&\|\Delta\hatbsS_{22}\|_{\rm op}\le u/N,\\
    &\|\Delta\hatbs\ell_a\|_2\le2B_yu/\sqrt N\quad(a=1,2),
    &&\|\hatbs h\|_2\le\|\bs w_\calO\|_2u/N
    \end{aligned}\right\}.
\end{equation}
We also introduce the two constants which will appear when controlling the terms of $\hat R$:
\begin{equation}\label{eq:schur_bound_constants}
    K_S^{\poly}=\left(1+\|\overline{\bs\Sigma}_{12}\overline{\bs\Sigma}_{22}^{-1}\|_{\rm op}\right)^2,
    \qquad K_S^{\rm nl}=2K_S^{\poly}/\kappa_{22},
\end{equation}
and remember that $\kappa_{22}>0$ is defined by the previously given definition $\kappa_{22}\le  \lambda_{\min}(\overline{\bs\Sigma}_{22})$.
In the remainder of this subsection, we will assume that $u$ satisfies
\begin{equation}\label{eq:deterministic_smallness}
    u\le1,\qquad u\le\kappa_{22}/2,\qquad
    K_S^{\poly}u/\sqrt N+K_S^{\rm nl}u^2\le\kappa_{11}/2.
\end{equation}
All subsequent bounds are deterministic on $\calE_u$. The associated constants may depend on the choice of POVM, $\calO$, $\sigma$, and the training states through the conditioning constants, but not on $N$ or $\ntr$.

\parTitle{Control of nonlinear terms in the remainder}
We start observing that $(\Delta_c\hatbsS_{\rm eff})_{\rm nl}$ can be written as
\begin{equation}\label{eq:fullexp_DeltaSeffNL}
    (\Delta_c\hatbsS_{\rm eff})_{\rm nl}
    =-\left(\Delta\hatbsS_{12}-\overline{\bs\Sigma}_{12}\overline{\bs\Sigma}_{22}^{-1}\Delta\hatbsS_{22}\right)
    \hatbsS_{22}^{-1}
    \left(\Delta\hatbsS_{21}-\Delta\hatbsS_{22}\overline{\bs\Sigma}_{22}^{-1}\overline{\bs\Sigma}_{21}\right).
\end{equation}
Remembering that $\hatbs S_{22}=N^{-1}\overline{\bs\Sigma}_{22}+\Delta\hatbs S_{22}$, Weyl's inequality on the event $\calE_u$ gives $\lambda_{\rm min}(\hatbs S_{22})\ge \frac{\kappa_{22}-u}{N}\ge \frac{\kappa_{22}}{2N}$, and thus $\|\hatbs S_{22}^{-1}\|_{\rm op}\le 2N/\kappa_{22}$.
Each outer factor in \cref{eq:fullexp_DeltaSeffNL} has instead norm at most $(1+\|\overline{\bs\Sigma}_{12}\overline{\bs\Sigma}_{22}^{-1}\|_{\rm op})u/\sqrt N$. Together with the first-order expression in \cref{eq:decomposition_DeltaSeff}, this proves
\begin{equation}\label{eq:bound_SeffFull}
    \|(\Delta_c\hatbsS_{\rm eff})_{\rm nl}\|_{\rm op}\le K_S^{\rm nl}u^2,
    \qquad
    \|\Delta_c\hatbsS_{\rm eff}\|_{\rm op}\le K_S^{\poly}u/\sqrt N+K_S^{\rm nl}u^2.
\end{equation}
The smallness conditions \cref{eq:deterministic_smallness} therefore imply
\begin{equation}\label{eq:good_event_invertibility}
    \lambda_{\min}(\hatbsS_{22})\ge\kappa_{22}/(2N),
    \qquad \lambda_{\min}(\hatbsS_{\rm eff})\ge\kappa_{11}/2.
\end{equation}
In particular, $\hatbsP$ has full row rank on this event.
The same factorization applied to $\hatbs\ell_{\rm eff}$ yields
\begin{equation}\label{eq:explicit_leff_remainder}
    (\Delta_c\hatbs\ell_{\rm eff})_{\rm nl}
    =-\left(\Delta\hatbsS_{12}-\overline{\bs\Sigma}_{12}\overline{\bs\Sigma}_{22}^{-1}\Delta\hatbsS_{22}\right)\hatbsS_{22}^{-1}\Delta\hatbs\ell_2,
\end{equation}
and using these in $\hatbs q_{\rm nl}$ gives
\begin{equation}\label{eq:bound_qnl}
\begin{aligned}
    \hatbs q_{\rm nl}
    =\left(\Delta\hatbsS_{12}-\overline{\bs\Sigma}_{12}\overline{\bs\Sigma}_{22}^{-1}\Delta\hatbsS_{22}\right)\hatbsS_{22}^{-1}
    \left(\hatbs h-\Delta\hatbsS_{22}\overline{\bs\Sigma}_{22}^{-1}\overline{\bs\Sigma}_{21}U_1^T\bs w_\calO\right)
    \implies
    \|\hatbs q_{\rm nl}\|_2\le K_S^{\rm nl}\|\bs w_\calO\|_2u^2/\sqrt N.
\end{aligned}
\end{equation}
Note that the three factors in this expansion for $\hatbs q_{\rm nl}$ have scales $u/\sqrt N$, $N$, and $u/N$, respectively, and bounding $\Delta\hatbsS_{21}U_1^T\bs w_\calO$ and $\Delta\hatbs\ell_2$ separately would lose the cancellation in the last factor.
The definitions of $\bs q_0$ and $\hatbs q_{\poly}$ now imply
\begin{equation}\label{eq:bounds_fullq}
    \|\bs q_0\|_2\le C/N,\qquad
    \|\hatbs q-\bs q_0\|_2\le Cu/\sqrt N,\qquad
    \|\hatbs q\|_2\le C/\sqrt N,
\end{equation}
where the constant may be taken as
\begin{equation}\label{eq:def_Cq}
    C = C_0^{(q)}+C_1^{(q)},
    \qquad
    C_0^{(q)}=\|\overline{\bs\Sigma}_{\rm eff}\|_{\rm op}\|\bs w_\calO\|_2,
    \qquad
    C_1^{(q)}=
    2B_y\left(
    1+\|\overline{\bs\Sigma}_{12}\overline{\bs\Sigma}_{22}^{-1}\|_{\rm op}
    \right)
    +(K_S^{\poly}+K_S^{\rm nl})\|\bs w_\calO\|_2,
\end{equation}
where $C_0^{(q)}/N$ bounds $\bs q_0$, while the two terms in $C_1^{(q)}$ bound $\hatbs q_p$ and $\hatbs q_{\rm nl}$ on $\calE_u$; the last inequality in \cref{eq:bounds_fullq} also uses $u\le 1$ and $N\ge 1$.

\parTitle{Remainder bound}
With these bounds we can now control the magnitude of $\hat R$ conditionally on $\calE_u$ and \cref{eq:deterministic_smallness}: for a constant $K_R$ independent of $N$ and $\ntr$:
\begin{equation}\label{eq:R_bound}
    \calE_u\ \text{and \cref{eq:deterministic_smallness}}
    \quad\Longrightarrow\quad |\hat R|\le K_Ru^2/\sqrt N.
\end{equation}
We will show in \cref{thm:remainder_high_probability} how to tie the factor $u$ to a confidence interval; this will be needed to make statements about the maximum magnitude of $\hat R$ being below a certain threshold with some given probability.

A concrete, conservative choice for the constant is obtained by bounding the four terms in \cref{eq:remainder_with_q} separately.
With $K_S\equiv K_S^{\poly}+K_S^{\rm nl}$ from \cref{eq:schur_bound_constants}, the constant $C$ from \cref{eq:def_Cq}, and the conditioning constant $\kappa_{11}$ from \cref{eq:direct_uniform_conditioning_2}, one may take
\begin{equation}\label{eq:def_KR}
    K_R=\|\bsp_{\sigma}\|_2\left[
    \frac{C K_S^{\rm nl}}{\kappa_{11}^2}
    +\frac{C K_S}{\kappa_{11}^2}
    +\frac{2C K_S^2}{\kappa_{11}^3}
    +\frac{K_S^{\rm nl}\|\bs w_\calO\|_2}{\kappa_{11}}
    \right].
\end{equation}
The four summands correspond, in order, to the four terms of \cref{eq:remainder_with_q}; \cref{eq:good_event_invertibility} supplies $\|\bsS_{\rm eff}^{-1}\|_{\rm op}\le\kappa_{11}^{-1}$ and $\|\hatbsS_{\rm eff}^{-1}\|_{\rm op}\le2\kappa_{11}^{-1}$. 
These constants are not optimized so tighter bounds might very well be possible.

\subsection{High-probability bounds}%
\label{sec:highprobbounds}

The deterministic argument in \cref{subsec:direct_remainder_bound} has reduced remainder control to one question: how often do all the fluctuations defining $\calE_u$ remain small?
We answer it for independent multinomial shots, conditional on the fixed training states and uniformly for $N\ge1$.
A scalar concentration bound first controls each noise column; vector and matrix versions then control the linear and quadratic moments together.
Substituting their common concentration radius into \cref{eq:R_bound} gives an explicit sample-size and confidence condition.
This will control typical realizations in which $\calE_u$ is realized with $u$ satisfying \cref{eq:deterministic_smallness}; we will discuss later on in \cref{subsec:direct_mse_bound} how to deal with ill-conditioned realizations, which is necessary to bound the MSE in expectation.

For each training state $i\in\{1,\ldots,\ntr\}$, let $Z_{i,1},\ldots,Z_{i,N}$ be independent categorical random variables with probability vector $\bs p_i\in\RR^{\nout}$. Assume also that the samples associated with distinct training states are independent.
Define
\begin{equation}
    \hatbs p_i=\frac{1}{N}\sum_{s=1}^N\bs e_{Z_{i,s}},
    \qquad
    \hatbs\xi_i=\sqrt N\left(\hatbs p_i-\bs p_i\right),
    \label{eq:multinomial_noise_definition}
\end{equation}
where $\bs e_{Z_{i,s}}$ denotes the canonical vectors.
The two block dimensions are $m_1=d^2$ and $m_2=\nout-d^2$, so that
$m_1+m_2=\nout$.  We denote by
$\hatbs\xi_{i,a}\in\RR^{m_a}$ the $a$-th block of $\hatbs\xi_i$.
Recall that $y_i=\bs p_{i,1}^T U_1^T\bs w_{\calO}$ and $B_y\equiv\max_i|\langle\calO,\rho_i\rangle|
\le\|\calO\|_\infty$.
For $0<\delta<1$, define
\begin{equation}
    q_\delta=2\nout\log 9+\log\!\left(\frac{32}{\delta}\right),
    \qquad
    r_\delta
    =\sqrt{\frac{q_\delta}{\ntr}}+\frac{q_\delta}{\ntr}.
    \label{eq:def_concentration_radius}
\end{equation}
The numerical factors in $q_\delta$ are chosen so that the vector and matrix bounds below hold simultaneously.
\begin{lemma}[Subgaussian multinomial error]
\label{lem:multinomial_subgaussian_direct}
For every $\bs v\in\RR^{\nout}$ and $\lambda\in\RR$,
\begin{equation}
    \EE\exp\!\left(\lambda\bs v^T(\hatbs p_i-\bs p_i)\right)
    \le
    \exp\!\left(\frac{\lambda^2\|\bs v\|_2^2}{4N}\right).
    \label{eq:multinomial_subgaussian_mgf}
\end{equation}
Equivalently, in the convention
$\EE e^{\lambda X}\le e^{\lambda^2\sigma^2/2}$, the random variable
$\bs v^T\hatbs\xi_i$ is subgaussian with parameter
$\sigma^2\le\|\bs v\|_2^2/2$.
\end{lemma}

\begin{proof}
Set $Y_s=\bs v^T(\bs e_{Z_{i,s}}-\bs p_i)$. Then $\EE[Y_s]=0$, and the range of $Y_s$ has length $\max_j v_j-\min_j v_j\le\sqrt2\|\bs v\|_2$.
Hoeffding's lemma therefore implies $\EE e^{tY_s}\le\exp(t^2\|\bs v\|_2^2/4)$. Independence and $\bs v^T(\hatbs p_i-\bs p_i)=N^{-1}\sum_{s=1}^NY_s$ now give
\begin{equation}
    \EE\exp\!\left(\lambda\bs v^T(\hatbs p_i-\bs p_i)\right)
    =\prod_{s=1}^N\EE\exp\!\left(\frac{\lambda}{N}Y_s\right)
    \le\exp\!\left(\frac{\lambda^2\|\bs v\|_2^2}{4N}\right).
\end{equation}
\end{proof}

\begin{lemma}[Concentration of $\Delta\hatbs\ell_a$]
\label{lem:ell_concentration}
With probability at least $1-\delta$, simultaneously for
$a\in\{1,2\}$,
\begin{equation}
    \|\Delta\hatbs\ell_a\|_2
    \le
    2B_y\sqrt{\frac{q_\delta}{N\ntr}}
    \le
    \frac{2B_yr_\delta}{\sqrt N}.
    \label{eq:ell_concentration_complete}
\end{equation}
\end{lemma}

\begin{proof}
Because
$\Delta\hatbs\ell_a=(\ntr\sqrt N)^{-1}\sum_{i=1}^{\ntr}y_i\hatbs\xi_{i,a}$,
\cref{lem:multinomial_subgaussian_direct} and independence across $i$ imply,
for every unit vector
$\bs v\in\RR^{m_a}$,
\begin{equation}
    \EE\exp\!\left(\lambda\bs v^T\Delta\hatbs\ell_a\right)
    \le
    \exp\!\left(
        \frac{\lambda^2}{4N\ntr^2}\sum_{i=1}^{\ntr}y_i^2
    \right)
    \le
    \exp\!\left(\frac{B_y^2\lambda^2}{4N\ntr}\right).
\end{equation}
If \(B_y=0\), then every \(y_i=0\) and
\(\Delta\hatbs\ell_a=0\), so the claim is immediate.  Assume henceforth
that \(B_y>0\).
The Chernoff bound therefore yields
\begin{equation}
    \PP\!\left(|\bs v^T\Delta\hatbs\ell_a|>t\right)
    \le 2\exp\!\left(-\frac{N\ntr t^2}{B_y^2}\right).
    \label{eq:ell_fixed_direction_tail}
\end{equation}
However, a bound in one direction is insufficient for $\calE_u$, which requires control over the full vector norm. To derive a bound in the other direction, let $\mathcal N_a$ be
a $1/4$-net of the unit sphere in $\RR^{m_a}$ with
$|\mathcal N_a|\le9^{m_a}$. For any nonzero $\bs z\in\RR^{m_a}$, choose
$\bs v\in\mathcal N_a$ within distance $1/4$ of
$\bs z/\|\bs z\|_2$; the case $\bs z=0$ is immediate.  Then
\begin{equation}
    \|\bs z\|_2
    =\frac{\bs z^T\bs z}{\|\bs z\|_2}
    \le |\bs v^T\bs z|+\frac14\|\bs z\|_2,
\end{equation}
so
$\|\bs z\|_2\le(4/3)\max_{\bs v\in\mathcal N_a}
|\bs v^T\bs z|$; we use the looser factor $2$ below.
Taking $t=B_y\sqrt{q_\delta/(N\ntr)}$ in
~\cref{eq:ell_fixed_direction_tail} and applying a union bound gives
\begin{equation}
    \PP\!\left(
        \|\Delta\hatbs\ell_a\|_2
        >2B_y\sqrt{\frac{q_\delta}{N\ntr}}
    \right)
    \le 2\,9^{m_a}e^{-q_\delta}
    \le\frac{\delta}{16}.
\end{equation}
A second union bound over $a\in\{1,2\}$ proves the simultaneous statement.
\end{proof}

\begin{lemma}[Quadratic concentration]
\label{lem:quadratic_concentration}
For $a,b\in\{1,2\}$, define
\begin{equation}
    \hatbs Q_{ab}
    =\frac{1}{\ntr}\sum_{i=1}^{\ntr}
    \left(
        \hatbs\xi_{i,a}\hatbs\xi_{i,b}^T
        -\EE[\hatbs\xi_{i,a}\hatbs\xi_{i,b}^T]
    \right).
\end{equation}
There is an absolute constant $K_{\rm B}>0$ such that, simultaneously for
all four pairs $(a,b)$, with probability at least $1-\delta$,
\begin{equation}
    \|\hatbs Q_{ab}\|_{\mathrm{op}}
    \le K_{\rm B}r_\delta.
    \label{eq:quadratic_concentration_complete}
\end{equation}
Equivalently,
\begin{equation}
    \left\|
        \frac{1}{N\ntr}\sum_{i=1}^{\ntr}
        \left(
            \hatbs\xi_{i,a}\hatbs\xi_{i,b}^T
            -\EE[\hatbs\xi_{i,a}\hatbs\xi_{i,b}^T]
        \right)
    \right\|_{\mathrm{op}}
    \le\frac{K_{\rm B}r_\delta}{N}.
\end{equation}
\end{lemma}

\begin{proof}
Fix unit vectors $\bs u\in\RR^{m_a}$ and $\bs v\in\RR^{m_b}$, and let
\begin{equation}
    X_i
    =(\bs u^T\hatbs\xi_{i,a})(\bs v^T\hatbs\xi_{i,b})
    -\EE[(\bs u^T\hatbs\xi_{i,a})(\bs v^T\hatbs\xi_{i,b})].
\end{equation}
Both factors have uniformly bounded subgaussian norm by
\cref{lem:multinomial_subgaussian_direct}.
They need not be independent: the standard product inequality for Orlicz
norms still gives a uniform subexponential bound for their product, and
centering changes that bound by at most a factor of two. Thus
$\|X_i\|_{\psi_1}\le K_X$ for one absolute constant $K_X$. Bernstein's
inequality then gives an absolute constant $K_{\rm B}$ such that, for every
$s>0$,
\begin{equation}
    \PP\!\left(
        \left|\frac{1}{\ntr}\sum_{i=1}^{\ntr}X_i\right|
        >K_{\rm B}\left(\sqrt{\frac{s}{\ntr}}+\frac{s}{\ntr}\right)
    \right)
    \le2e^{-s}.
    \label{eq:quadratic_fixed_direction_tail}
\end{equation}
Let $\mathcal N_a$ and $\mathcal N_b$ be $1/4$-nets of the corresponding
unit spheres. Their product has cardinality at most $9^{m_a+m_b}$, and any
$\bs B\in\RR^{m_a\times m_b}$ satisfies
$
    \|\bs B\|_{\mathrm{op}}
    \le2\max_{\bs u\in\mathcal N_a,\bs v\in\mathcal N_b}
    |\bs u^T\bs B\bs v|
$.
Apply~\cref{eq:quadratic_fixed_direction_tail} with $s=q_\delta$. Since
$m_a+m_b\le2\nout$, a union bound over the two nets and all four block pairs
has failure probability at most
$8\,9^{2\nout}e^{-q_\delta}=\delta/4$. Absorbing the net factor $2$ into
$K_{\rm B}$ proves~\cref{eq:quadratic_concentration_complete}.
\end{proof}

\begin{lemma}[Concentration of the linear matrix terms]
\label{lem:linear_matrix_concentration}
For $a\in\{1,2\}$, define
$\hatbs A_a=\ntr^{-1}\sum_{i=1}^{\ntr}\bs p_{i,1}\hatbs\xi_{i,a}^T$,
where $\bs p_{i,1}\in\RR^{m_1}$ is the first block of $\bs p_i$. Then,
simultaneously for $a\in\{1,2\}$, with probability at least $1-\delta$,
\begin{equation}
    \|\hatbs A_a\|_{\mathrm{op}}
    \le2\sqrt{\frac{q_\delta}{\ntr}}
    \le2r_\delta.
    \label{eq:linear_matrix_concentration_complete}
\end{equation}
Consequently,
\begin{equation}
    \left\|
        \frac{1}{\ntr\sqrt N}\sum_{i=1}^{\ntr}
        \bs p_{i,1}\hatbs\xi_{i,a}^T
    \right\|_{\mathrm{op}}
    \le\frac{2r_\delta}{\sqrt N}.
\end{equation}
\end{lemma}

\begin{proof}
Fix unit vectors $\bs u\in\RR^{m_1}$ and $\bs v\in\RR^{m_a}$. The
subgaussian moment-generating-function bound gives
\begin{equation}
    \EE\exp\!\left(\lambda\bs u^T\hatbs A_a\bs v\right)
    \le
    \exp\!\left(
        \frac{\lambda^2}{4\ntr^2}
        \sum_{i=1}^{\ntr}(\bs u^T\bs p_{i,1})^2
    \right)
    \le\exp\!\left(\frac{\lambda^2}{4\ntr}\right),
\end{equation}
where we used $\|\bs p_{i,1}\|_2\le\|\bs p_i\|_2\le1$. Hence $\PP(|\bs u^T\hatbs A_a\bs v|>t)\le2e^{-\ntr t^2}$.
Applying the same $1/4$-net argument as above with $t=\sqrt{q_\delta/\ntr}$ gives, for each $a$,
\begin{equation}
    \PP\!\left(
        \|\hatbs A_a\|_{\mathrm{op}}>2\sqrt{\frac{q_\delta}{\ntr}}
    \right)
    \le2\,9^{m_1+m_a}e^{-q_\delta}
    \le\frac{\delta}{16}.
\end{equation}
A union bound over the two values of $a$ completes the proof.
\end{proof}

We can now assemble the event needed for the deterministic proof:
\begin{corollary}[Simultaneous basic event]
\label{cor:prob_events}
There is a fixed absolute constant $K_{\rm S}>0$ such that, with
probability at least $1-\delta$, all of the following bounds hold
simultaneously, where $r_\delta$ is defined in
\cref{eq:def_concentration_radius}:
\begin{equation}
\begin{gathered}
    \|\Delta\hatbsS_{11}\|_{\mathrm{op}},\
    \|\Delta\hatbsS_{12}\|_{\mathrm{op}},\
    \|\Delta\hatbsS_{21}\|_{\mathrm{op}}
    \le\frac{K_{\rm S}r_\delta}{\sqrt N},
    \qquad
    \|\hatbs h\|_2
    \le
    K_{\rm B}\|\bs w_{\calO}\|_2\frac{r_\delta}{N},
    \\
    \|\Delta\hatbsS_{22}\|_{\mathrm{op}}
    \le\frac{K_{\rm S}r_\delta}{N},
    \qquad
    \|\Delta\hatbs\ell_a\|_2
    \le\frac{2B_yr_\delta}{\sqrt N},
    \quad a\in\{1,2\}.
    \label{eq:simultaneous_basic_event}
\end{gathered}
\end{equation}
Thus, for $C_0\ge\max\{1,K_{\rm S},K_{\rm B}\}$, we have $\PP(\calE_{C_0 r_\delta})\ge1-\delta$.
\end{corollary}
\begin{proof}
The block fluctuations have the exact decompositions
\begin{equation}
\begin{gathered}
    \Delta\hatbsS_{11}
    =\frac{\hatbs A_1+\hatbs A_1^T}{\sqrt N}
    +\frac{\hatbs Q_{11}}{N},
    \qquad
    \Delta\hatbsS_{12}
    =\frac{\hatbs A_2}{\sqrt N}
    +\frac{\hatbs Q_{12}}{N},
    \\
    \Delta\hatbsS_{21}
    =\frac{\hatbs A_2^T}{\sqrt N}
    +\frac{\hatbs Q_{21}}{N},
    \qquad
    \Delta\hatbsS_{22}
    =\frac{\hatbs Q_{22}}{N}.
    \label{eq:block_fluctuation_decompositions}
\end{gathered}
\end{equation}
Combining
\cref{lem:ell_concentration,lem:quadratic_concentration,lem:linear_matrix_concentration}
and using $N^{-1}\le N^{-1/2}$ proves the asserted block and vector
bounds.  The common choice of $q_\delta$ makes their total failure
probability smaller than $\delta$.
Finally, the bound on $\hatbs h$ follows from
\cref{eq:def_h,eq:quadratic_concentration_complete}
with $(a,b)=(2,1)$.
\end{proof}

It now remains to compare the concentration radius in \cref{cor:prob_events} with the radius for which the inverse bounds hold, i.e. with \cref{eq:deterministic_smallness}.
Under the standing uniform conditioning assumptions, a fixed threshold can enforce those bounds for every $N\ge1$.
Our main objective here is thus to bound the $u$ in \cref{eq:R_bound} into a factor that depends on the confidence interval $\delta$.
\begin{theorem}[High-probability remainder bound]
\label{thm:remainder_high_probability}
For fixed $d$ and $\nout$, there are deterministic constants $u_\star,C_R,C_0>0$, independent of $N$ and $\ntr$, such that for any $\delta$ satisfying $C_0 r_\delta\le u_\star$, with probability at least $1-\delta$, we have
$
    |\hat R|
    \le C_R\frac{r_\delta^2}{\sqrt N}
$.
\end{theorem}

\begin{proof}
Choose $C_0\ge\max\{1,K_{\rm S},K_{\rm B}\}$ as in~\cref{cor:prob_events}, and set
\begin{equation}\label{eq:direct_uniform_threshold}
    u_\star\equiv
    \min\left\{
        1,\,
        \frac{\kappa_{22}}{2},\,
        \frac{\kappa_{11}}{4K_S^{\poly}},\,
        \sqrt{\frac{\kappa_{11}}{4K_S^{\rm nl}}}
    \right\}.
\end{equation}
With the uniform constants chosen above, $u_\star$ is independent of $N$ and $\ntr$. For every $N\ge1$ and $u\le u_\star$, $K_S^{\poly}u/\sqrt N+K_S^{\rm nl}u^2 \le\kappa_{11}/2$.
Thus~\cref{eq:R_bound,eq:good_event_invertibility} applies on $\calE_u$.
If $C_0r_\delta\le u_\star$, the event $\calE_{C_0r_\delta}$ has probability at least $1-\delta$, and on that event $|\hat R|\le K_R C_0^2r_\delta^2/\sqrt N$.
The claim follows with $C_R\equiv K_R C_0^2$.
\end{proof}

The smallness condition on $r_\delta$ can also be written directly as a requirement on $\ntr$. With
$c_\star\equiv
(\sqrt{1+4u_\star/C_0}-1)^2/4$, the definition
in~\cref{eq:def_concentration_radius} gives
\begin{equation}\label{eq:fixed_confidence_sample_requirement}
    C_0r_\delta\le u_\star
    \quad\Longleftrightarrow\quad
    \ntr\ge\frac{q_\delta}{c_\star}
    =\frac{2\nout\log9+\log(32/\delta)}{c_\star}.
\end{equation}
Thus for fixed dimensions, training states, and $\delta$, if $\ntr \ge q_\delta/c_\star$, then
\begin{equation}\label{eq:remainder_op_scaling}
    \PP\!\left(|\hat R|\le\frac{4C_Rq_\delta}{\ntr\sqrt N}\right)\ge1-\delta.
\end{equation}
For $\hat R$ to be small relative to $|B|=|b_N|/N$ and the fluctuation, we also require $4C_Rq_\delta\sqrt N\ll |b_N|\ntr$ and $4C_Rq_\delta\ll \ntr\sqrt{N\EE[\hat L^2]}$, respectively, which are satisfied for sufficiently large $\ntr$.

\subsection{Mean-squared error bound}
\label{subsec:direct_mse_bound}

\Cref{sec:highprobbounds}, and in particular \cref{thm:remainder_high_probability}, controls the magnitude of $\hat R$ on typical training realizations.
However, to bound the MSE we need control on $\hat R$ \textit{on average}, which means we must make sure that rare realizations with very ill-conditioned matrices do not blow up $\hat R$ enough to make the MSE expansion invalid.
We therefore split the remainder's second moment over the fixed event $\calE_{u_\star}$ and its complement.  The threshold $u_\star$ is the deterministic inverse-stability radius defined in \cref{eq:direct_uniform_threshold}; unlike the variable radius $u$ used to state the family of events $\calE_u$, this is independent of the confidence level, $N$, and $\ntr$.  On $\calE_{u_\star}$ we integrate the concentration tail, while on its complement we use a global bound on $\|\hatbsP^+\|_{\rm op}$.

\begin{lemma}[Remainder on the conditioning event]
\label{lem:good_event_remainder_moment}
    Under \cref{eq:direct_uniform_conditioning_2}, for $\ntr\ge\nout$ and every $N\ge1$,
    \begin{equation}\label{eq:remainder_good_event_second_moment}
        \EE[\hat R^2\boldsymbol1_{\calE_{u_\star}}]
        \le K_{\rm good}\frac{\nout^2}{N\ntr^2},
    \end{equation}
    where $K_{\rm good}$ is independent of $N,\ntr$.
\end{lemma}
\begin{proof}
The deterministic bound in \cref{eq:R_bound} makes $\hat R$ quadratic in the moment fluctuations, so a second-moment estimate for $\hat R$ requires a fourth-moment estimate for their size.
To integrate all componentwise bounds at once, define the realized fluctuation radius
\begin{equation}\label{eq:def_realized_radius}
    \mathfrak e\equiv\inf\{u\ge0:\calE_u\text{ occurs}\}.
\end{equation}
The events $\calE_u$ are nested, so $\mathfrak e$ is simply the smallest radius that bounds every quantity in \cref{eq:def_good_event}; equivalently, $\{\mathfrak e\le u\}=\calE_u$.

By \cref{cor:prob_events}, $\PP(\mathfrak e>C_0r_\delta)\le\delta$.  For $t\ge\log16$, set $\delta=16e^{-t}$ and define
\begin{equation}\label{eq:def_tail_constant}
    C_q\equiv2\log9+\log2+1,
    \qquad K_{\rm tail}\equiv C_0\max\{\sqrt{C_q},C_q\}.
\end{equation}
Since $q_\delta=2\nout\log9+t+\log2\le C_q(\nout+t)$ for $\nout\ge1$, the definition of $r_\delta$ yields
\begin{equation}\label{eq:normalized_fluctuation_tail}
    \PP\!\left(\mathfrak e>K_{\rm tail}\left[\sqrt{\frac{\nout+t}{\ntr}}+\frac{\nout+t}{\ntr}\right]\right)\le16e^{-t},\qquad t\ge0.
\end{equation}
For $0\le t<\log16$, the right-hand side $16e^{-t}$ is at least one, so the bound is trivial.
To convert the tail bound into a fourth moment, put
$g(s)=K_{\rm tail}[\sqrt{(\nout+\log16+s)/\ntr}+(\nout+\log16+s)/\ntr]$.
Taking $t=\log16+s$ in \cref{eq:normalized_fluctuation_tail} gives $\PP(\mathfrak e>g(s))\le e^{-s}$.
For any nonnegative random variable, its fourth moment is the area under the tail of its fourth power.
Splitting that area at $g(0)^4$ and using the preceding bound gives
\begin{equation}
    \EE[\mathfrak e^4]\le g(0)^4+\int_0^\infty e^{-s}\,\rmd[g(s)^4]
    =\int_0^\infty e^{-s}g(s)^4\,\rmd s
    \le C\left[\frac{\nout^2}{\ntr^2}+\frac{\nout^4}{\ntr^4}\right].
\end{equation}
Here, the equality is obtained integrating by parts; the boundary at infinity vanishes because $g$ grows only linearly.  Since $\ntr\ge\nout$, the last expression is at most $C\nout^2/\ntr^2$ for a universal constant $C$.  On the fixed conditioning event $\calE_{u_\star}=\{\mathfrak e\le u_\star\}$, \cref{eq:R_bound} gives $\hat R^2\le K_R^2\mathfrak e^4/N$, which proves \cref{eq:remainder_good_event_second_moment}.
\end{proof}

\Cref{lem:good_event_remainder_moment} handles the contributions from $\calE_{u^\star}$. On $\calE_{u^\star}^c$, however, the Schur inverse formulas used in the previous sections do not apply, so we must return to working directly with the pseudoinverse $\hatbsP^+_N$, which remains defined rank-deficient realizations.
We will leverage that the discreteness of the shot counts gives a coarse bound on its largest possible singular value, and then exploit the exponentially small probability of these unlikely events to conclude that $\EE[\hat R^2 \bs1_{\calE_{u^\star}^c}]$ is also negligible.

\begin{lemma}[Integer-frequency pseudoinverse bound]
\label{lem:integer_pseudoinverse_bound}
For every realization of a multinomial frequency matrix,
\begin{equation}\label{eq:integer_pseudoinverse_bound}
    \|\hatbsP^+\|_{\rm op}\le\nout^{\nout-1}N^{\nout}.
\end{equation}
\end{lemma}
\begin{proof}
Let $s=\rank(\hatbsP)$.  If $s=0$, the pseudoinverse vanishes.  Otherwise, $\bs A=N\hatbsP$ has an invertible $s\times s$ integer minor $\bs B$, so $|\det\bs B|\ge1$ and $\|\bs B\|_{\rm op}\le sN$.  The product of the largest $s-1$ singular values is at most $(sN)^{s-1}$, hence $\sigma_s(\bs B)\ge(sN)^{-(s-1)}$.  Singular-value interlacing under row and column restriction gives $\sigma_s(\bs A)\ge\sigma_s(\bs B)$.  Therefore $\|\hatbsP^+\|_{\rm op}=N/\sigma_s(\bs A)\le s^{s-1}N^s\le\nout^{\nout-1}N^{\nout}$.
\end{proof}

\begin{lemma}[Exceptional realizations]
\label{lem:bad_event_remainder}
For every realization there is a constant $K_{\rm bad}$, independent of $N$ and $\ntr$, such that
\begin{equation}\label{eq:bad_event_deterministic_bound}
    |\hat R|^2\le K_{\rm bad}\ntr(N\nout)^{2\nout}.
\end{equation}
Define
\begin{equation}\label{eq:def_u0}
    u_0\equiv\min\left\{1,\left[\frac{u_\star}{2K_{\rm tail}}\right]^2,
    \frac{u_\star}{2K_{\rm tail}}\right\}.
\end{equation}
The role of $u_0$ is only to choose a tail parameter for which the radius in \cref{eq:normalized_fluctuation_tail} is no larger than the fixed stability threshold $u_\star$.  If $u_0\ntr\ge\nout$, then
\begin{equation}\label{eq:bad_event_second_moment}
    \PP(\calE_{u_\star}^c)\le16e^{-u_0\ntr+\nout},
    \qquad
    \EE[\hat R^2\boldsymbol1_{\calE_{u_\star}^c}]
    \le16K_{\rm bad}\ntr(N\nout)^{2\nout}e^{-u_0\ntr+\nout}.
\end{equation}
\end{lemma}
\begin{proof}
By \cref{lem:integer_pseudoinverse_bound}, the empirical prediction obeys
\begin{equation}\label{eq:global_prediction_bound}
    |\bs y^T\hatbsP^+\bsp_\sigma|
    \le\|\bs y\|_2\|\hatbsP^+\|_{\rm op}\|\bsp_\sigma\|_2
    \le B_y\sqrt{\ntr}\,\nout^{\nout-1}N^{\nout}.
\end{equation}
The ideal term satisfies $|\bs y^T\bsP^+\bsp_\sigma|\le B_y/\sqrt{\kappa_{11}}$, because $\|\bs y\|_2\le B_y\sqrt{\ntr}$, $\|\bsp_\sigma\|_2\le1$, and $\|\bsP^+\|_{\rm op}\le(\kappa_{11}\ntr)^{-1/2}$.  Combining this with \cref{eq:global_prediction_bound} gives $|\Delta\hat o|\le C_\Delta\sqrt{\ntr}(N\nout)^{\nout}$.  In addition, $|B|\le C_B$ follows directly from \cref{eq:bias_exact} and uniform conditioning.  Finally, inserting the deterministic bound $\|\hatbs\xi_i\|_2\le\sqrt{2N}$ into the explicit column formula \cref{eq:def_gi} cancels its $N^{-1/2}$ and $N^{-1}$ prefactors and gives $|\hat L_i|\le C_L$; hence $|\hat L|\le C_L$.  Therefore
\begin{equation}
    |\hat R|^2=|\Delta\hat o-B-\hat L|^2
    \le3(C_\Delta^2\ntr(N\nout)^{2\nout}+C_B^2+C_L^2),
\end{equation}
which is \cref{eq:bad_event_deterministic_bound} after enlarging $K_{\rm bad}$.  Multiplying this deterministic bound by $\PP(\calE_{u_\star}^c)$ is the precise step that controls the exceptional contribution to $\EE[\hat R^2]$.
The deterministic estimate alone grows with $N$; its usefulness comes from the rarity of the event on which it is needed.
To combine the two, take $t=u_0\ntr-\nout$ in \cref{eq:normalized_fluctuation_tail}.
Its threshold is at most $u_\star$, proving the probability bound; multiplying by the deterministic bound proves the second claim.
\end{proof}

Combining \cref{lem:good_event_remainder_moment,lem:bad_event_remainder} controls the second moment of $\hat R$ over the entire multinomial experiment, including nearly singular empirical matrices.
The following theorem states the resulting finite-parameter control before taking any asymptotic limit.
It deliberately keeps $\EE[\hat L^2]$ exact: simplifying that variance at large $N$ is a separate step.

\begin{theorem}[Finite-$N$ control of the mean, variance, and MSE]
\label{prop:direct_moments}
Assume \cref{eq:direct_uniform_conditioning}, $N\ge1$, and $\ntr\ge\nout/u_0$.
Define the remainder bound
\begin{equation}\label{eq:total_remainder_second_moment}
    \EE[\hat R^2]\le\mathcal R^{(2)}_{N,\ntr}
    \equiv\frac{K_{\rm good}\nout^2}{N\ntr^2}
    +16K_{\rm bad}\ntr(N\nout)^{2\nout}e^{-u_0\ntr+\nout}.
\end{equation}
With $B=-b_N/N$ and the exact $\EE[\hat L^2]=\ntr^{-2}\sum_i\EE[\hat L_i^2]$ from \cref{eq:def_gi},
\begin{equation}\label{eq:direct_mean_variance_control}
\begin{aligned}
    |\EE[\Delta\hat o]+b_N/N|&\le\sqrt{\mathcal R^{(2)}_{N,\ntr}},
    \qquad
    |\Var(\Delta\hat o)-\EE[\hat L^2]|&\le2\sqrt{\EE[\hat L^2]\mathcal R^{(2)}_{N,\ntr}}+\mathcal R^{(2)}_{N,\ntr}.
\end{aligned}
\end{equation}
Moreover,
\begin{equation}\label{eq:MSE_remainder_bound}
    |\EE[\Delta\hat o^2]-B^2-\EE[\hat L^2]|
    \le2\bigl(|B|+\sqrt{\EE[\hat L^2]}\bigr)\sqrt{\mathcal R^{(2)}_{N,\ntr}}+\mathcal R^{(2)}_{N,\ntr}.
\end{equation}
\end{theorem}
\begin{proof}
Adding the good-event and exceptional-event estimates from \cref{lem:good_event_remainder_moment,lem:bad_event_remainder} gives \cref{eq:total_remainder_second_moment}.
We can therefore average the exact decomposition without excluding nearly singular realizations.
Centering removes the $B\hat L$ cross term, but not the correlations with the remainder:
$\EE[\hat L]=0$ gives
$\EE[\Delta\hat o]=B+\EE[\hat R]$ and
\begin{equation}\label{eq:mse_cross_terms}
    \EE[\Delta\hat o^2]-B^2-\EE[\hat L^2]
    =2B\EE[\hat R]+2\EE[\hat L\hat R]+\EE[\hat R^2].
\end{equation}
The inequalities $|\EE[\hat R]|\le\sqrt{\EE[\hat R^2]}$ and $|\EE[\hat L\hat R]|\le\sqrt{\EE[\hat L^2]\EE[\hat R^2]}$ applied to \cref{eq:mse_cross_terms} prove the first line of \cref{eq:direct_mean_variance_control} and \cref{eq:MSE_remainder_bound}.
For the variance, the deterministic $B$ drops out; expanding $\Var(\hat L+\hat R)$ and using $\Var(\hat R)\le\EE[\hat R^2]$ gives the stated bound.
Thus the mean is controlled by comparing $\sqrt{\mathcal R^{(2)}_{N,\ntr}}$ with $|B|$, whereas the variance is controlled by comparing $\mathcal R^{(2)}_{N,\ntr}$ with $\EE[\hat L^2]$.  These are different relative-error tests even though both use the same absolute remainder bound.
\end{proof}

\parTitle{Growth conditions and relative errors}
\Cref{prop:direct_moments} gives finite-parameter absolute errors for the mean, variance, and MSE whenever the uniform conditioning assumption \cref{eq:direct_uniform_conditioning} holds, $N\ge1$, and $\ntr\ge\nout/u_0$.
To obtain a scaling for the MSE, we must ensure that the exceptional contributions are negligible, which means that the exponential term in \cref{eq:total_remainder_second_moment} must be small compared to the scale of the other MSE contributions.
For example, to make it negligible relative to $1/(N\ntr^2)$, a sufficient condition is
\begin{equation}\label{eq:exceptional_relative_condition}
    u_0\ntr-(2\nout+1)\log N-3\log\ntr-2\nout\log\nout-\nout\longrightarrow+\infty.
\end{equation}
At fixed dimensions and uniform conditioning, the simpler sufficient condition
\begin{equation}\label{eq:direct_growth_condition}
    \ntr/\log(2+N)\longrightarrow\infty
\end{equation}
makes that term smaller than every fixed inverse polynomial in $N,\ntr$.

This implies in particular that $\ntr\to\infty$ at any fixed $N\ge1$ is a sufficient condition to make $\EE[\hat R^2]$ negligible, although it might not be necessary due to the very coarse bound we used for \cref{lem:bad_event_remainder}.

\begin{theorem}[Asymptotic mean and variance at finite shot count]\label{thm:direct_mean_variance_asymptotics}
Assume the setting of \cref{prop:direct_moments}, with fixed $d,\nout$, uniform conditioning, and the growth condition \cref{eq:direct_growth_condition}.  Then
\begin{equation}\label{eq:direct_mean_variance_scaling}
\begin{aligned}
    \EE[\Delta\hat o]&=-b_N/N+O\!\left(\frac{\nout}{\sqrt N\,\ntr}\right),
    \\
    \Var(\Delta\hat o)&=\EE[\hat L^2]+O\!\left(\frac{\nout}{N\ntr^{3/2}}+\frac{\nout^2}{N\ntr^2}\right)
    =\frac{V_N}{N\ntr}+O\!\left(\frac{1}{N^2\ntr}+\frac{\nout}{N\ntr^{3/2}}+\frac{\nout^2}{N\ntr^2}\right).
\end{aligned}
\end{equation}
\end{theorem}
The factors of $\nout$ shown here come from the concentration estimate; the implicit constants also depend on the conditioning data and are not uniform in $d$ or $\nout$.
For nonzero coefficients of fixed order, the mean has relative error tending to zero if $\ntr\gg\nout\sqrt N$, whereas the variance is asymptotic to $\EE[\hat L^2]$ if $\ntr\gg\nout^2$ and $N\ntr\,\EE[\hat L^2]$ stays bounded away from zero.
Replacing this exact fluctuation variance by $V_N/(N\ntr)$ additionally requires $N\to\infty$ and nonvanishing $V_N$.
More generally, the finite-parameter tests are $\sqrt{\mathcal R^{(2)}_{N,\ntr}}\ll|b_N|/N$ for the mean and $\mathcal R^{(2)}_{N,\ntr}\ll\EE[\hat L^2]$ for the variance.
These comparisons remain meaningful when the coefficients are small; the simpler power-counting statements do not.

\parTitle{Fixed-$N$ and joint limits}
The comparisons above lead to two distinct conclusions from the same theorem.
At fixed $N$, \cref{prop:direct_moments,eq:def_gi} control the exact first-variation variance and the deterministic floor $b_N^2/N^2$.  The replacement of that exact variance by $V_N/(N\ntr)$ uses the large-$N$ estimate \cref{eq:L_variance_leading}, and the two-term MSE further requires the joint-limit conditions that will be stated in \cref{thm:direct_joint_mse}.
If $\ntr\,\EE[\hat L^2]$ has a nonzero limit, it gives the leading $1/\ntr$ training variance.
The MSE satisfies $\EE[\Delta\hat o^2]=b_N^2/N^2+O(\ntr^{-1})$.
This identifies a nonzero bias floor when $b_N$ has a nonzero limit, but does not identify the coefficient of the $1/\ntr$ correction: the term $2B\EE[\hat R]$ can contribute at that same order.
Even replacing $V_N/(N\ntr)$ by the exact $\EE[\hat L^2]$ does not remove this limitation.
If $b_N=0$, that cross term vanishes, and \cref{eq:MSE_remainder_bound} gives the sharper comparison directly.

When both resources grow, every term omitted from $B^2+\EE[\hat L^2]$ in the exact identity \cref{eq:mse_cross_terms} becomes small on the scale of the total leading error.  In particular, the potentially largest omitted correlation is $2B\EE[\hat R]$; the other two are $2\EE[\hat L\hat R]$ and $\EE[\hat R^2]$.
To check this without assuming either bias or variance dominates, we collect the four error scales as
\begin{equation}\label{eq:def_MSE_error}
    \mathrm{Err}_{N,\ntr}
    =\frac{1}{N^2\ntr}+\frac{\nout}{N^{3/2}\ntr}
    +\frac{\nout}{N\ntr^{3/2}}+\frac{\nout^2}{N\ntr^2}.
\end{equation}
The first term is the variance replacement error in \cref{eq:L_variance_leading}; the remaining terms bound $2B\EE[\hat R]$, $2\EE[\hat L\hat R]$, and $\EE[\hat R^2]$, respectively.

\begin{theorem}[Joint-limit MSE]\label{thm:direct_joint_mse}
Assume fixed $d,\nout$, the uniform conditioning in \cref{eq:direct_uniform_conditioning}, $N,\ntr\to\infty$, and \cref{eq:direct_growth_condition}.  Then
\begin{equation}\label{eq:final_MSE_expansion}
    \EE[\Delta\hat o^2]=\frac{b_N^2}{N^2}+\frac{V_N}{N\ntr}+O(\mathrm{Err}_{N,\ntr}),
\end{equation}
and, because every term in $\mathrm{Err}_{N,\ntr}$ is $o(N^{-2}+(N\ntr)^{-1})$,
\begin{equation}\label{eq:direct_leading_joint_MSE}
    \EE[\Delta\hat o^2]=\frac{b_N^2}{N^2}+\frac{V_N}{N\ntr}
    +o\!\left(N^{-2}+(N\ntr)^{-1}\right).
\end{equation}
If $|b_N|$ and $V_N$ stay bounded above and away from zero, this is a relative approximation to the total MSE, and the two displayed contributions cross at $\ntr\asymp NV_N/b_N^2$.
\end{theorem}
The theorem controls the sum of squared bias and variance.  It does not, by itself, show that $\EE[\Delta\hat o]=-b_N/N+o(b_N/N)$: that separate identification of the statistical bias requires $\sqrt{\mathcal R^{(2)}_{N,\ntr}}\ll|b_N|/N$, as stated after \cref{eq:direct_mean_variance_scaling}.

This completes the direct proof.
The finite-$N$ theorem remains available when the simple resource fit is not controlled, and the next section expresses its coefficients and conditioning requirements in terms of the training and measurement frames.

\subsection{Comparison with the homogeneous Gaussian model}
\label{subsec:gaussian_multinomial_comparison}

We can now compare the coefficients extracted from the homogeneous Gaussian model in \cref{eq:gaussian_b_V_identification}, with those obtained directly for the multinomial model.
This comparison concerns the leading large-$N$ coefficients $b$ and $V$, rather than the finite-$N$ coefficients $b_N,V_N$ or the subleading terms in the two expansions.
For the multinomial model, \cref{eq:frame_coefficients_infinite_N} gives
\begin{equation}
\begin{aligned}
    b
    =\frac{1}{\ntr}\sum_{i=1}^{\ntr}
      \Cov_i\!\left(
        \bs w^{\rm eff}_{\calO},
        \ntr\langle\bstildemu,\tilde\sigma\rangle
      \right),
    \qquad V=\frac{1}{\ntr}\sum_{i=1}^{\ntr}
      \bigl(\ntr\langle\tilde\rho_i,\sigma\rangle\bigr)^2
      \Var_i(\bs w^{\rm eff}_{\calO}).
\end{aligned}
\label{eq:multinomial_b_V_for_comparison}
\end{equation}
Here every covariance is computed with the state-dependent multinomial
covariance $\bs\Sigma_i$, while $\bs w^{\rm eff}_{\calO}$ is defined from
their average $\overline{\bs\Sigma}$ as in \cref{eq:def_weff}.

The bias comparison does not in fact require homogeneous column
covariances.  Since the second argument of the covariance in
\cref{eq:multinomial_b_V_for_comparison} is independent of $i$, the sum
over the training states gives
\begin{equation}
\label{eq:multinomial_gaussian_b_match}
    b=(\bs w^{\rm eff}_{\calO})^T
      \overline{\bs\Sigma}\,
      U_1\bsS_{11}^{-1}\bsp_{\sigma,1}=\ntr\,\bs w_{\calO}^TU_1
      \overline{\bs\Sigma}_{\rm eff}
      \bsP_{11}^{-T}V_1^T
      \langle\tilde{\bs\rho},\sigma\rangle.
\end{equation}
The second equality follows from two elementary block
identities:
\begin{enumerate}
    \item $
    (\bs w^{\rm eff}_{\calO})^T
    \overline{\bs\Sigma}U_1
    =\bs w_{\calO}^TU_1
      \overline{\bs\Sigma}_{\rm eff}
    $
    holds because in the $(U_1,U_2)$ basis the two blocks of
    $\bs w^{\rm eff}_{\calO}$ are
    $U_1^T\bs w_{\calO}$ and
    $-\overline{\bs\Sigma}_{22}^{-1} \overline{\bs\Sigma}_{21}U_1^T\bs w_{\calO}$;
    multiplying them by the two block rows of $\overline{\bs\Sigma}U_1$ produces precisely its Schur complement.
    \item $
        \bsS_{11}^{-1}\bsp_{\sigma,1}
        =\ntr\,\bsP_{11}^{-T}V_1^T
          \langle\tilde{\bs\rho},\sigma\rangle$
    holds because $\bsS_{11}=\bsP_{11}\bsP_{11}^T/\ntr$ and
    $\bsP_{11}^{-1}\bsp_{\sigma,1} =V_1^T\langle\tilde{\bs\rho},\sigma\rangle$.
\end{enumerate}
The right-hand side of \cref{eq:multinomial_gaussian_b_match} is exactly
the Gaussian coefficient in \cref{eq:gaussian_b_V_identification}.
Thus replacing the state-dependent multinomial covariances by their
average leaves the leading bias coefficient unchanged.

For the variance, suppose now that the multinomial column covariances are homogeneous, that is,
$\bs\Sigma_i=\overline{\bs\Sigma}$, $i=1,\ldots,\ntr$. Under this approximation,
\begin{equation}
    \Var_i(\bs w^{\rm eff}_{\calO})
    =(\bs w^{\rm eff}_{\calO})^T
      \overline{\bs\Sigma}\bs w^{\rm eff}_{\calO}
    =\bs w_{\calO}^TU_1
      \overline{\bs\Sigma}_{\rm eff}U_1^T\bs w_{\calO},
    \label{eq:weff_variance_identity}
\end{equation}
independently of $i$.  Moreover,
$\langle\tilde{\bs\rho},\sigma\rangle$ belongs to the row space of
$\bsP$, and hence
\begin{equation}
    \frac{1}{\ntr}\sum_{i=1}^{\ntr}
      \bigl(\ntr\langle\tilde\rho_i,\sigma\rangle\bigr)^2
    =\ntr\,
      \|V_1^T\langle\tilde{\bs\rho},\sigma\rangle\|_2^2.
    \label{eq:dual_norm_identity_comparison}
\end{equation}
Substitution of
\cref{eq:weff_variance_identity,eq:dual_norm_identity_comparison} into the
second line of \cref{eq:multinomial_b_V_for_comparison} gives
\begin{equation}
    V
    =\ntr\,
      \|V_1^T\langle\tilde{\bs\rho},\sigma\rangle\|_2^2
      \bs w_{\calO}^TU_1
      \overline{\bs\Sigma}_{\rm eff}U_1^T\bs w_{\calO},
    \label{eq:multinomial_gaussian_V_match}
\end{equation}
which is exactly the Gaussian expression in
\cref{eq:gaussian_b_V_identification}.

Consequently, the homogeneous Gaussian model reproduces both leading multinomial coefficients when the multinomial covariances are homogeneous.

\section{Cost scaling of training a QELM}%
\label{app:asymptotic}

\Cref{app:direct_remainder} has established which parts of the prediction error can be controlled and in which limits.
We now translate these results into a resource-cost statement for the MSE.
In particular, \cref{sec:costscaling_infteststat} rewrites the exact bias and variance coefficients in frame language, \cref{sec:costscaling_conclusions} converts
their asymptotic limits into the dataset-size/shot-budget tradeoff, and \cref{sec:costscaling_averages} evaluates the resulting leading coefficients for symmetric tight designs.

\parTitle{Problem setting and notation}
Suppose that every state in the training set is measured \(N\) times.
Because shot noise generally makes the target rows inconsistent with the
row space of the observed frequency matrix, QELM training is naturally
formulated as the least-squares problem
\begin{equation}\label{eq:training_linear_problem}
    \hatbsW_{\bs \calO}
    \in
    \operatorname*{arg\,min}_{\bs W}
    \left\|
        \bsW\hatbsP_N-\bsv_\calO^T\bsMrho
    \right\|_{\rm F}^2 .
\end{equation}
Here $\hatbsP_N$ is the noisy probability matrix, whose $k$-th column
contains the frequencies obtained from \(N\) measurements of $\rho_k$,
and
$\bsP=\EE[\hatbsP_N]=\langle\bs\mu,\bs\rho\rangle
=\bsMmu^T\bsMrho$.
As explained in \cref{app:notation}, $\bsv_\calO=\on{vec}(\calO)$ denotes the vectorization of the target observable $\calO$ and $\bsMrho=\sum_{i=1}^{\ntr} \on{vec}(\rho_i)\bra i$ collects the (vectorized) training states $\bs\rho\equiv\{\rho_i\}_{i=1}^{\ntr}$.
The canonical least-squares solution is
$\hatbsW_{\bs\calO}
=\langle\bs\calO,\bs\rho\rangle\hatbsP_N^+$.

\parTitle{Training with infinite statistics}
At infinite training statistics, $\hatbsP_N$ is replaced by its
expectation $\bsP$.  Under the completeness assumptions stated below,
the noiseless least-squares residual vanishes, and the canonical solution is
\begin{equation}\label{eq:WO_solution}
    \bsW_{\calO}
    = (\bsv_\calO^T\bsMrho)(\bsMmu^T\bsMrho)^+
    = \bsv_\calO^T \bsMmu^{T+}
    = \langle\calO, \bstildemu\rangle.
\end{equation}
This result assumes that both $\bsMrho$ and $\bsMmu$ have full row rank: physically this amounts to the training states being complete --- i.e., spanning the space of all Hermitian operators --- and the POVM being informationally complete.
Under this assumption, the pseudoinverse of the product also factorises as $(\bsMmu^T\bsMrho)^+=\bsMrho^+\bsMmu^{T+}$.
Furthermore, it is a general fact that the pseudoinverse of a surjective matrix can be interpreted as the transpose of the matrix representing the \textit{canonical dual frame} of the frame given by the columns of the original matrix.
Formally, this means that $\bsMmu^{+}= (\bsMtildemu)^T$.
For more details on this perspective, we refer to the discussion in~\cite{innocenti2023potential,innocenti2023shadow}.
In conclusion, training with infinite statistics amounts to recovering a description of the POVM $\bs\mu$ through its dual $\bstildemu$.


\subsection{Explicit expressions for bias and variance coefficients}%
\label{sec:costscaling_infteststat}

For any $\bs u,\bs v\in \RR^{\nout}$, define the conditional covariance and variance as
$\Cov_i(\bs u,\bs v)\equiv \bs u^T\bs\Sigma_i\bs v$, and $\Var_i(\bs u)\equiv \Cov_i(\bs u,\bs u)$.
The finite-$N$ coefficients in \cref{eq:final_MSE_expansion} can then be written as
\begin{equation}\label{eq:def_bN_VN}
\begin{aligned}
    b_N &=
    \frac{1}{\ntr} \sum_{i=1}^{\ntr}
    \Cov_i(
    \bs w^{\rm eff}_{\calO},
    U_1 \bsS_{\rm eff}^{-1}\bsp_{\sigma,1}
    )
    =
    (\bs w^{\rm eff}_{\calO})^T
    \overline{\bs\Sigma}
    U_1\bsS_{\rm eff}^{-1}\bsp_{\sigma,1}
    =
    \bs w_\calO^T
    \overline{\bs\Sigma}
    \left(\bsS+\frac{1}{N}\overline{\bs\Sigma}\right)^{-1}
    \!\bsp_\sigma,
    \\
    V_N &=
    \frac{1}{\ntr} \sum_{i=1}^{\ntr}
    (\bs p_{i,1}^T\bsS_{\rm eff}^{-1}\bsp_{\sigma,1})^2
    \Var_i(\bs w^{\rm eff}_{\calO})
    =
    \frac{1}{\ntr}\sum_{i=1}^{\ntr}
    \left(
    \bs p_i^T
    \Big(
    \bsS+\frac{1}{N}\overline{\bs\Sigma}
    \Big)^{-1}
    \bsp_\sigma
    \right)^2
    \Var_i(\bs w^{\rm eff}_{\calO}).
\end{aligned}
\end{equation}
Notice in particular how $b_N$ only depends on the averaged covariance matrix $\overline{\bs\Sigma}$, whereas this is not exactly the case for $V_N$.
The term $\Var_i(\bs w^{\rm eff}_{\calO})$ in $V_N$ can be written more explicitly but does not simplify as much:
\begin{equation}
    \Var_i(\bs w^{\rm eff}_{\calO})
    =
    (\bs w^{\rm eff}_{\calO})^T
    \bs\Sigma_i
    \bs w^{\rm eff}_{\calO}
    =
    \bs w_\calO^T
    (\bs I - U_2 \overline{\bs\Sigma}_{22}^{-1}\overline{\bs\Sigma}_{21} U_1^T)^T
    \bs\Sigma_i
    (\bs I - U_2 \overline{\bs\Sigma}_{22}^{-1}\overline{\bs\Sigma}_{21} U_1^T)
    \bs w_\calO.
\end{equation}
Note that $\bsS_{\rm eff}=\bsS_{11}+O(N^{-1})$, and
\begin{equation}
    U_1\bsS_{11}^{-1}\bsp_{\sigma,1}
    = \ntr
    \langle\bstildemu,\tilde\sigma\rangle,
    \qquad
    \bs p_{i,1}^T\bsS_{11}^{-1}\bsp_{\sigma,1}
    =\ntr\langle\tilde\rho_i,\sigma\rangle,
    \label{eq:dual_frame_coefficients}
\end{equation}
where $\tilde\rho_i=\frameOp_\rho^{-1}(\rho_i)$ and
$\tilde\sigma=\frameOp_\rho^{-1}(\sigma)$.
Substitution in \cref{eq:def_bN_VN} therefore yields
$b_N=b+O(N^{-1})$ and $V_N=V+O(N^{-1})$, with
\begin{equation}
\begin{gathered}
    b
    =\frac1{\ntr}\sum_{i=1}^{\ntr}
    \Cov_i(\bs w^{\rm eff}_{\calO},
    \ntr\langle\bstildemu,\tilde\sigma\rangle),
    \qquad
    V
    =\frac1{\ntr}\sum_{i=1}^{\ntr}
    \bigl(\ntr\langle\tilde\rho_i,\sigma\rangle\bigr)^2
    \Var_i(\bs w^{\rm eff}_{\calO}).
    \label{eq:frame_coefficients_infinite_N}
\end{gathered}
\end{equation}
The coefficients $b$ and $V$ are independent of $N$, but generally
retain dependence on the training ensemble. Under the uniform
conditioning assumptions, they remain bounded as $\ntr$ varies;
convergence as $\ntr\to\infty$ requires convergence of the relevant
empirical training moments.
By~\cref{eq:bias_exact}, the deterministic bias component is
$B=-b_N/N$, whereas the full mean is $B+\EE[\hat R]$.
The estimate in~\cref{eq:direct_mean_variance_scaling} gives $\EE[\hat R]=O(\nout/(\sqrt N\,\ntr))$ in the stated growth regime.
For nonvanishing $b_N$, this is subleading relative to $B$ when $\ntr\gg\nout\sqrt N$.

\subsection{Resource scaling}%
\label{sec:costscaling_conclusions}

\Cref{sec:costscaling_infteststat} identifies the finite-$N$ frame coefficients and their large $N$ limits.
To obtain the two-term resource fit, assume the fixed-dimensional conditioning hypotheses of \cref{thm:direct_joint_mse}, let $N,\ntr\to\infty$, and require $\ntr/\log(2+N)\to\infty$.  Substituting \cref{eq:def_bN_VN} into \cref{eq:final_MSE_expansion}, followed by the controlled replacements $b_N=b+O(N^{-1})$ and $V_N=V+O(N^{-1})$, gives
\begin{equation}%
\label{eq:leading_bias2_terms}
    \EE_{\hatbsXi}[\Delta\hat o^2]
    =
    \frac{b_N^2}{N^2}
    +
    \frac{V_N}{N\ntr}
    +O\!\left(\mathrm{Err}_{N,\ntr}\right)
    =
    \frac{b^2}{N^2}
    +
    \frac{V}{N\ntr}
    +O\!\left(\mathrm{Err}_{N,\ntr}+\frac1{N^3}\right),
\end{equation}
with $\mathrm{Err}_{N,\ntr}$ defined in
\cref{eq:def_MSE_error}.
The additional \(N^{-3}\) term comes from replacing \(b_N^2\) by \(b^2\); the \(N\)-dependence of \(V_N\) is already absorbed by \(\mathrm{Err}_{N,\ntr}\).
Note that the $N^{-3}$ correction is retained because $\mathrm{Err}_{N,\ntr}$ is $O(\ntr^{-1})$ for large $\ntr$, and $N^{-3}$ need not be smaller than terms like $N^{-2}\ntr^{-1}$ when $\ntr > N$.
At fixed $d$ and $\nout$, the leading variance is controlled by the total number $N_{\rm tot}=N\ntr$ of measured training copies, whereas the squared bias is controlled by the number of shots per state.
If $|b|$ and $V$ remain bounded above and away from zero, increasing $\ntr$ at fixed $N$ suppresses the variance but leaves the bias floor $b_N^2/N^2$, and the crossover occurs at $\ntr\asymp N$.

Note that the bias is not due to the rank increase alone.
Indeed, when $\nout=d^2$ there is no such rank increase, but $\bsS+N^{-1}\barSigma$ still differs from $\bsS$.
The ordinary inverse argument then gives the deterministic shift $-N^{-1}\bs w_\calO^T\barSigma(\bsS+N^{-1}\barSigma)^{-1}\bsp_\sigma$, which need not vanish.
Overcompleteness changes that shift through the Schur complement and creates additional noise directions, but it is not the only source of bias.

\parTitle{Quantitative sample requirements}
\Cref{eq:leading_bias2_terms}, on its own, is still not sufficient to conclusively determine whether a given experiment and noise realization lies in the controlled regime.
For this, we require the empirical inverses to be stable, then compare the remainder and exceptional-event contribution with the accuracy sought.
For a fixed confidence $\delta$, a sufficient condition valid for every $N\ge1$ is given by \cref{eq:fixed_confidence_sample_requirement}.
Alternatively, one may directly use the $N$-dependent conditions \cref{eq:deterministic_smallness} with $u=C_0r_\delta$.
For example, a sufficient asymptotic condition for those inequalities to hold with a margin is
\begin{equation}\label{eq:sample_condition_with_geometry}
    \ntr\gg C_0^2q_\delta\max\left\{1,\kappa_{22}^{-2},
    \frac{(K_S^{\poly})^2}{N\kappa_{11}^2},\frac{K_S^{\rm nl}}{\kappa_{11}}\right\}.
\end{equation}
To derive \cref{eq:sample_condition_with_geometry}, first take $\ntr\ge q_\delta$, so that \cref{eq:def_concentration_radius} gives $r_\delta\le2\sqrt{q_\delta/\ntr}$.  Insert $u=C_0r_\delta$ into the three requirements of \cref{eq:deterministic_smallness}.  The bounds $u\ll1$ and $u\ll\kappa_{22}$ give, respectively, $\ntr\gg C_0^2q_\delta$ and $\ntr\gg C_0^2q_\delta\kappa_{22}^{-2}$.  Splitting the last requirement into $K_S^{\poly}u/\sqrt N\ll\kappa_{11}$ and $K_S^{\rm nl}u^2\ll\kappa_{11}$ gives the last two terms in the maximum.  Numerical factors are absorbed by $\gg$; no finite-size threshold is claimed.
This condition controls the inverses and the absolute high-probability remainder, not its relative size or the exceptional-event contribution to the MSE.
Those additional checks are given in \cref{eq:direct_mean_variance_control,eq:exceptional_relative_condition,eq:MSE_remainder_bound}.
At fixed small $N$, increasing $\ntr$ controls the exact variance and the bias floor, but does not justify dropping $O(N^{-2}\ntr^{-1})$ from the variance or resolving the $1/\ntr$ correction to the MSE.
The constants in the concentration estimates have not been numerically optimized, so these are sufficient analytical criteria, not calibrated finite-size thresholds.

For illustration, the balanced tight designs below have $\kappa_{22}=1/\nout$, $\kappa_{11}=1/[\nout(d+1)^2]$, and $K_S^{\poly}=1$, $K_S^{\rm nl}=2\nout$.
With these exact values, \cref{eq:sample_condition_with_geometry} becomes the conservative requirement $\ntr\gg C_0^2q_\delta\nout^2\max\{(d+1)^2,(d+1)^4/N\}$.
The uniform threshold $u_\star=1/[4\nout(d+1)^2]$ is more restrictive when $N$ is large.
These explicit conditioning penalties explain why the present proof does not supply a dimension-independent rule such as $\ntr\gg\nout$.

\subsection{Leading coefficients for symmetric states and POVMs}%
\label{sec:costscaling_averages}

Here we specialize the results of \cref{sec:costscaling_conclusions} to a symmetric reference example in which training states are pure and form a complex projective 2-design, and the POVM is a fixed balanced rank-1 tight measurement frame:
$\mu_a=\frac d{\nout}\PP(\psi_a)$, $a=1,\cdots,\nout$,
with $(\psi_a)_a$ a complex projective two-design.
When averaging over test state $\sigma$ and target observable $\calO$ these are also taken as independent Haar-random rank-one projectors.
In this case we can compute explicitly the terms contributing to bias and variance.

\parTitle{Some useful identities}
Before stating the main results we start by showing what some of the quantities involved in the expressions for $b_N$ and $V_N$ look like in this symmetric model.
The duals of effects and training states are
\begin{equation}
    \tilde\mu_a=(d+1)\PP(\psi_a)-\bsI,
    \qquad
    \tilde\rho_i
    =\frac d{\ntr}\bigl[(d+1)\rho_i-\bsI\bigr].
    \label{eq:tight_dual_frames}
\end{equation}
For the averaged covariance matrix, remember that $\bs\Sigma(\rho)=\on{diag}(\bs p(\rho))-\bs p(\rho)\bs p(\rho)^T$, hence
\begin{equation}
\begin{gathered}
    \overline{\bs\Sigma} = \frac{1}{\ntr}\sum_{i=1}^{\ntr}
    [\on{diag}(\bs p(\rho_i))-\bs p(\rho_i)\bs p(\rho_i)^T]
    =
    \on{diag}(\langle\bs\mu,\bar\rho\rangle)
    -\bs S,
    \qquad
    \bs S\equiv \frac{1}{\ntr}\bs P\bs P^T=
    \frac{1}{\ntr}\sum_{i=1}^{\ntr}\langle\bs\mu,\rho_i\rangle\langle\rho_i,\bs\mu\rangle.
\end{gathered}
\end{equation}
Furthermore, for all states, $\bs\Sigma(\rho)\bs1=0$, and for all training states ensembles $\overline{\bs\Sigma}_{12}=U_1^T \on{diag}(\langle\bs\mu,\bar\rho\rangle) U_2$, as by construction $U_1^T \bs S U_2=0$.
In the symmetric case we further have $\langle\bs\mu,\bar\rho\rangle=\bsI_{\nout}/\nout$, and thus $\overline{\bs\Sigma}_{12}=0$.
Again in the symmetric case we can decompose $\bs S$ exactly into its action in the traceless subspace and its orthogonal, similarly to what we showed in \cref{app:frame}:
\begin{equation}\label{eq:decomposition_S_Pis}
    \bs S = \frac{1}{\nout} \Pi_0 + \frac{1}{\nout(d+1)^2} \Pi_{\rm tr},
    \qquad
    \Pi_0 \equiv \frac{\bs1\bs1^T}{\nout},
    \quad
    \Pi_{\rm tr} = U_1 U_1^T - \Pi_0,
\end{equation}
where $\Pi_0$ is the projection onto the identity subspace, and $\Pi_{\rm tr}$ is the projection onto the remainder of the ``signal subspace'', specified by $U_1U_1^T$.
From $\bsI_{\nout}=\Pi_0+\Pi_{\rm tr}+\Pi_2$, with $\Pi_2=U_2U_2^T-U_1 U_1^T$ projection on the noisy subspace, we then get
\begin{equation}\label{eq:barSigma_decomposed}
    \overline{\bs\Sigma} = \frac{d(d+2)}{\nout(d+1)^2}\Pi_{\rm tr} + \frac{1}{\nout}\Pi_2,
    \qquad
    \overline{\bs\Sigma}_{11}
    =
    \frac{d(d+2)}{\nout(d+1)^2} \Pi_{\rm tr}.
\end{equation}
For $\bs S_{\rm eff}=\bs S_{11}+\frac{1}{N}\overline{\bs\Sigma}_{\rm eff}$ we have instead, using the same symmetry assumptions, $\overline{\bs\Sigma}_{\rm eff}=\overline{\bs\Sigma}_{11}$, and
\begin{equation}\label{eq:Seff_decomposed}
    \bs S_{\rm eff} = \frac{1}{\nout} \Pi_0 + \frac{1+d(d+2)/N}{\nout (d+1)^2} \Pi_{\rm tr},
    \qquad
    \bs S_{\rm eff}^{-1} = \nout \Pi_0 + \frac{\nout (d+1)^2}{1+d(d+2)/N} \Pi_{\rm tr},
\end{equation}
which differs from \cref{eq:decomposition_S_Pis} only in the eigenvalue in the traceless signal subspace.

We can now move on to state the main results of this section:
\begin{lemma}
    If the training states form a complex projective 2-design, and the POVM is a balanced rank-1 tight measurement frame, the large-$N$ bias and variance coefficients $b,V$ as given in \cref{eq:frame_coefficients_infinite_N}, averaged over Haar-random pure test states and observables, equal
    \begin{equation}
        \EE_{\sigma,\calO}[b^2] =
        \frac{(d-1)(d+2)^2}{d+1},
        \qquad
        \EE_{\sigma,\calO}[V] =
        \frac{d(d-1)(d+2)}{d+1}.
    \end{equation}
\end{lemma}
\begin{proof}
    Because $\overline{\bs\Sigma}_{12}=0$ we have
    $\bs w_\calO^{\rm eff} = (\bs I - U_2 \overline{\bs\Sigma}_{22}^{-1}\overline{\bs\Sigma}_{21} U_1^T)
    \bs w_\calO = \bs w_\calO$.
    Furthermore, standard design identities give
    \begin{equation}
    \begin{gathered}
        \ntr\langle\tilde\rho_i,\sigma\rangle
        =
        d\left[
            (d+1)\trace(\rho_i\sigma)-1
        \right],
        \qquad
        \EE_\sigma\!\left[
            \bigl(
                \ntr\langle\tilde\rho_i,\sigma\rangle
            \bigr)^2
        \right]
        =
        d^2,
        \\
        \EE_{\calO}
        \Var_i(\bs w_{\calO})
        =
        1-
        \EE_{\calO}
        [\trace(\calO\rho_i)^2]
        =
        \frac{(d-1)(d+2)}{d(d+1)},
    \end{gathered}
    \label{eq:tight_dual_coefficient_second_moment}
    \end{equation}
    and defining $\bs w_A\equiv\langle\bstildemu,A\rangle\in\RR^{\nout}$, also
    \begin{equation}
        \frac{1}{\ntr}
        \sum_{i=1}^{\ntr}
        \Cov_i(\bs w_{\calO},\bs w_A)
        =
        \frac{d+2}{d+1}
        \left[
            \trace(\calO A)
            -
            \frac{\trace(\calO)\trace(A)}{d}
        \right].
        \label{eq:tight_covariance_identity}
    \end{equation}
    Taking
    $A=\ntr\tilde\sigma=d[(d+1)\sigma-\bsI]$
    in \cref{eq:frame_coefficients_infinite_N} then yields
    \begin{equation}
        b
        =
        (d+2)
        \left[
            d\trace(\calO\sigma)-\trace(\calO)
        \right]
        =
        (d+2)
        \left[
            d\trace(\calO\sigma)-1
        \right],
        \label{eq:tight_bias_amplitude}
    \end{equation}
    where the second equality uses that $\calO$ is a rank-one projector.
    For independent Haar-random pure states,
    \begin{equation}
        \EE[\trace(\calO\sigma)]
        =
        \frac{1}{d},
        \qquad
        \EE[\trace(\calO\sigma)^2]
        =
        \frac{2}{d(d+1)}.
    \end{equation}
    These prove the statement.
\end{proof}

\begin{lemma}\label{lem:bN_vs_b}
    If the training states form a complex projective 2-design, and the POVM is a balanced rank-1 tight measurement frame, then the coefficients $b_N$ and $b$, as given in \cref{eq:def_bN_VN,eq:frame_coefficients_infinite_N}, satisfy
    \begin{equation}
        b_N = \frac{b}{1+d(d+2)/N},
        \qquad
        \frac{b_N^2}{N^2} =
        \frac{b^2}{[N+d(d+2)]^2}.
    \end{equation}
\end{lemma}

\begin{proof}
From \cref{eq:barSigma_decomposed,eq:Seff_decomposed} we get
\begin{equation}
    \overline{\bs\Sigma}_{11}\bs S_{\rm eff}^{-1}
    =\frac{d(d+2)}{1+d(d+2)/N}\Pi_{\rm tr}
    =\frac{1}{1+d(d+2)/N}
    \overline{\bs\Sigma}_{11}\bs S_{11}^{-1} ,
\end{equation}
and thus
\begin{equation}
    b_N = \bs w_\calO^T U_1 \overline{\bs\Sigma}_{11}\bs S_{\rm eff}^{-1}\bs p_{\sigma,1}
    =
    \frac{1}{1+d(d+2)/N}
    \bs w_\calO^T U_1 \overline{\bs\Sigma}_{11}\bs S_{11}^{-1}\bs p_{\sigma,1}
    =
    \frac{b}{1+d(d+2)/N}.
\end{equation}
\end{proof}

\begin{lemma}\label{lem:VN_vs_V}
    If the training states form a complex projective 2-design, and the POVM is a balanced rank-1 tight measurement frame, then the coefficients $V_N$ and $Vb$, as given in \cref{eq:def_bN_VN,eq:frame_coefficients_infinite_N}, satisfy
    \begin{equation}
    \begin{gathered}
        V_N = 
        \frac{1}{\ntr} \sum_{i=1}^{\ntr}
        [1+\beta_{N,d} (q_i-1)]^2
        \Var_i(\bs w_\calO)
        =
        \beta_{N,d}^2 V
        + 2\beta_{N,d}(1-\beta_{N,d})
        C + (1-\beta_{N,d})^2A,
        \\
        q_i \equiv \ntr\langle\tilde\rho_i,\sigma\rangle,
        \qquad
        \beta_{N,d} \equiv \frac{1}{1+d(d+2)/N},
        \qquad
        C \equiv \frac{1}{\ntr} \sum_{i=1}^{\ntr}
        q_i \Var_i(\bs w_\calO),
        \qquad
        A \equiv \frac{1}{\ntr} \sum_{i=1}^{\ntr}
        \Var_i(\bs w_\calO).
    \end{gathered}
    \end{equation}
    Averaging over test state and observables the relation simplifies to
    \begin{equation}
        \EE_{\sigma,\calO}[V_N] =
        \EE_{\sigma,\calO}[V]
        \left(
        \beta_{N,d}^2
        + \frac{1-\beta_{N,d}^2}{d^2}
        \right)
        =
        \frac{(d-1)(d+2)}{d(d+1)}
        \left[
        1+\frac{d^2-1}{[1+d(d+2)/N]^2}
        \right]
    \end{equation}
\end{lemma}
\begin{proof}
    Note that with the definitions given in the lemma, and remembering that $\bs w_\calO=\bs w_\calO^{\rm eff}$, we have $V=\ntr^{-1}\sum_{i=1}^{\ntr} q_i^2 \Var_i(\bs w_\calO)$.
    The statement follows from the same identities for $\bs S_{\rm eff}$ used in \cref{lem:bN_vs_b}.
\end{proof}

Thus, as a function of $N$, we have
\begin{equation}
    \EE_{\sigma,\calO}[V_N] \sim
    \begin{cases}
    1, & N \ll d, \\
    N^2/d^2, & d \ll N \ll d^2, \\
    d^2, & N \gg d^2.
    \end{cases}
    \qquad
    \mathbb{E}_{\sigma,\mathcal O}[b_N^2]
    \sim
    \begin{cases}
    N^2/d^2,
    & N \ll d^2, \\
    d^2,
    & N \gg d^2.
    \end{cases}
\end{equation}

\section{Numerical results}
\label{app:numerics}

In this section, we present several numerical simulations supporting the theoretical discussion in the manuscript. All simulations estimate the training bias and variance numerically using a finite number of realizations of the shot noise arising during training or testing. Achieving full agreement between the simulations and the theoretical results therefore requires unbiased estimators of the bias and variance. This subtlety is discussed in detail in \cref{subsec:bias_variance_estimators}.

The numerical simulations serve two distinct purposes. First, we explore a range of settings to demonstrate the broad applicability of the manuscript’s theoretical findings. For example, in \cref{subsec:random_povm}, we reproduce the predicted scaling of the MSE, bias, and variance for randomly generated POVMs. Second, the simulations in \cref{subsec:gaussian_test} provide empirical evidence supporting the validity of the Gaussian approximation discussed in the main text.

\subsection{Unbiased estimators for the training bias and variance}
\label{subsec:bias_variance_estimators}

The bias--variance decomposition in \cref{eq:bias_var_decomposition_train,eq:test_bias_var} is defined in terms of exact averages over the shot noise. In the simulations, the average over $\hatbsp_\sigma$ is performed analytically using \cref{eq:test_bias_var} while the average over the training noise is replaced by a finite sample mean. We derive the resulting estimators at fixed training set, POVM, observable $\calO$, and test state $\sigma$; any further average over these objects is performed only afterwards.

For each realization of the training shot-noise, let
\begin{equation}
    \hat y=\hatbsW_N^{(r)}\bsp_\sigma,
    \qquad
    y=\bsW_\calO\bsp_\sigma=\langle\calO,\sigma\rangle.
    \label{eq:numerical_training_predictions}
\end{equation}
Thus $\hat y$ is the infinite-test-statistics prediction obtained from a realization training shot-noise $\hatbsXi$, whereas $y$ is the exact target. If $\bar y_{\rm tr}\equiv\EE_{\hatbsXi}[\hat y]$, the quantities defined in \cref{eq:bias_var_decomposition_train} may also written as:
\begin{equation} 
\bias^2(N,\sigma)=(\bar y_{\rm tr}-y)^2\,,\quad\Var_{\rm tr}(N,\sigma)=\EE_{\hatbsXi}[(\hat y-\bar y_{\rm tr})^2]\,.
\end{equation}
Let $R\geq2$ be the number of independent training-noise realizations and let introduce the empirical mean and sample variance
\begin{equation}
    \hat y_R\equiv\frac{1}{R}\sum_{r=1}^R \hat y_r,
    \qquad
    \hat S_{R}^2\equiv\frac{1}{R-1}\sum_{r=1}^R(\hat y_r-\hat y_R)^2.
    \label{eq:numerical_training_sample_moments}
\end{equation}
where $\{\hat y_r\}_{r=1}^R$ are $R$ independent realization of $\hat y$ over the shot-noise. Using the independence of the $R$ training realization, it is simple to show that:
\begin{equation}
\label{eq:some_estimator_mean}
    \EE[\hat y_R]=\bar y_{\rm tr}\,,\quad \Var[\hat y_{R}]=\frac{\Var_{\rm tr}(N,\sigma)}{R}\,,\quad \EE[\hat S^2_R]=\Var_{\rm tr}(N,\sigma)\,.
\end{equation}
In particular, using the last identity in \cref{eq:some_estimator_mean}, it is immediate to recognize that the sample variance is actually an unbiased estimator of the true training variance:
\begin{equation}
    \widehat{\Var}_{\rm tr}(N,\sigma)\equiv \hat S^2_{R} = \frac{1}{R-1}\sum_{r=1}^R(\hat y_r-\hat y_R)^2
\end{equation}
is actually an unbiased estimator of the training variance.

On the other hand, the estimator of the square bias is more subtle. A naive guess would be to use $(\hat y_R-y)^2$ to estimate the actual bias. However, using again the independence of the noise realizations, it is easy to check that this would not be an unbiased estimator. In fact:
\begin{equation}
    \EE[(\hat y_R-y)^2] = (y-\EE[\hat y_R])^2 + \Var[\hat y_R] = \bias(N,\sigma)^2+\frac{1}{R}\Var(N,\sigma)
    \label{eq:numerical_bias_finite_R}
\end{equation}
The term on the left is therefore positively biased at finite $R$. Unbiased estimator of the squared bias is instead
\begin{equation}
    \widehat{\bias^2}(N,\sigma)\equiv(\hat y_R-y)^2-\frac{1}{R}\hat S_{R}^2.
    \label{eq:unbiased_training_bias_variance_estimators}
\end{equation}

\subsection{Numerical tests of the Gaussian approximation}
\label{subsec:gaussian_test}

We compare the prediction errors obtained with exact multinomial sampling and with a Gaussian approximation having the same conditional mean and covariance. For each training state, the empirical probability vector is generated according to either
\begin{equation}
\widehat{\boldsymbol p}
=\frac{1}{N}\operatorname{Mult}(N,\boldsymbol p),
\qquad\text{or}\qquad
\widehat{\boldsymbol p}
=\boldsymbol p+\frac{\boldsymbol\xi}{\sqrt N},
\quad
\boldsymbol\xi\sim\mathcal N\!\left(
0,\operatorname{diag}(\boldsymbol p)-\boldsymbol p\boldsymbol p^\top
\right).
\end{equation}
The Gaussian model retains the state dependence of the categorical covariance and its normalization constraint. It therefore tests the replacement of multinomial fluctuations by Gaussian fluctuations, separately from the additional homogeneity assumption introduced in section \cref{app:concentration}.

\Cref{fig:multi_vs_gauss_1} presents this comparison for a qubit ($d=2$) measured with either the MUB POVM ($n_{\text{out}}=6$) or a random POVM with $n_{\text{out}}=16$ outcomes. Each realization of the training noise defines a fitted regressor, whose squared prediction error is averaged over 50 Haar-random test states and 20 rank-one observables. We report both the mean and the median over noise realizations, together with the corresponding spread, to distinguish typical performance from sensitivity to rare events.

For $n_{\mathrm{tr}}<d^2=4$, the training states do not span the full operator space. Increasing $N$ then suppresses measurement noise but does not remove the lack of information in the training set, resulting in a nonzero large-$N$ error. Once the training states span the operator space, the error decreases with increasing measurement statistics. The multinomial and Gaussian models exhibit the same qualitative behavior, with close agreement between their medians and percentile ranges across the settings shown.

Near the interpolation threshold, the distribution of the single-experiment MSE becomes broad, and its mean is particularly sensitive to rare, poorly conditioned training matrices. The resulting fluctuations occur for both noise models and make the empirical mean a less stable indicator of agreement than the median and quantiles. These comparisons support the Gaussian approximation for the prediction statistics considered here. They do not, by themselves, establish the existence of inverse moments near the interpolation threshold or justify replacing the state-dependent covariance by a homogeneous one.
\begin{figure}[h!]
    \centering 
    \includegraphics[width=0.75\linewidth]{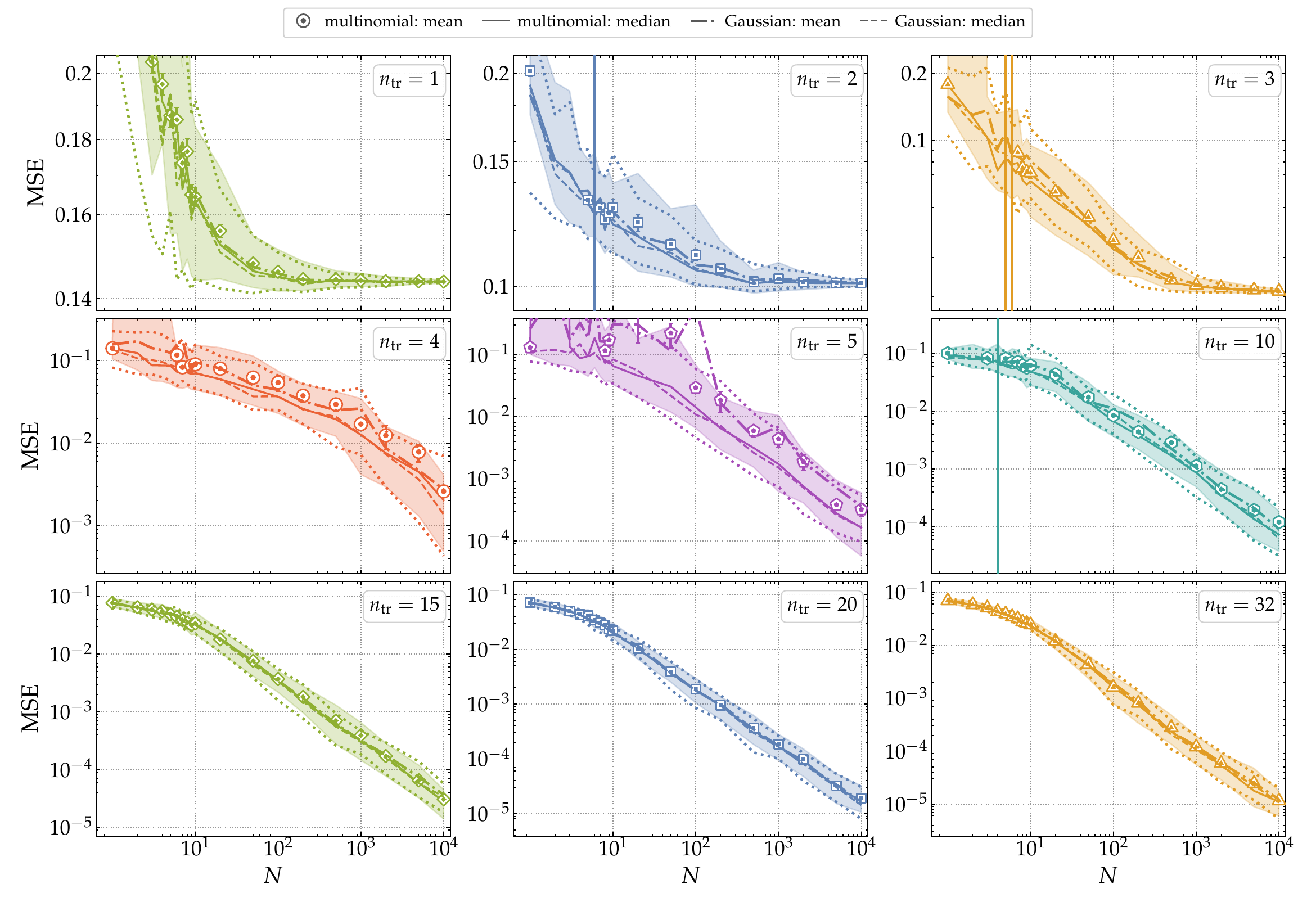}
    \includegraphics[width=0.75\linewidth]{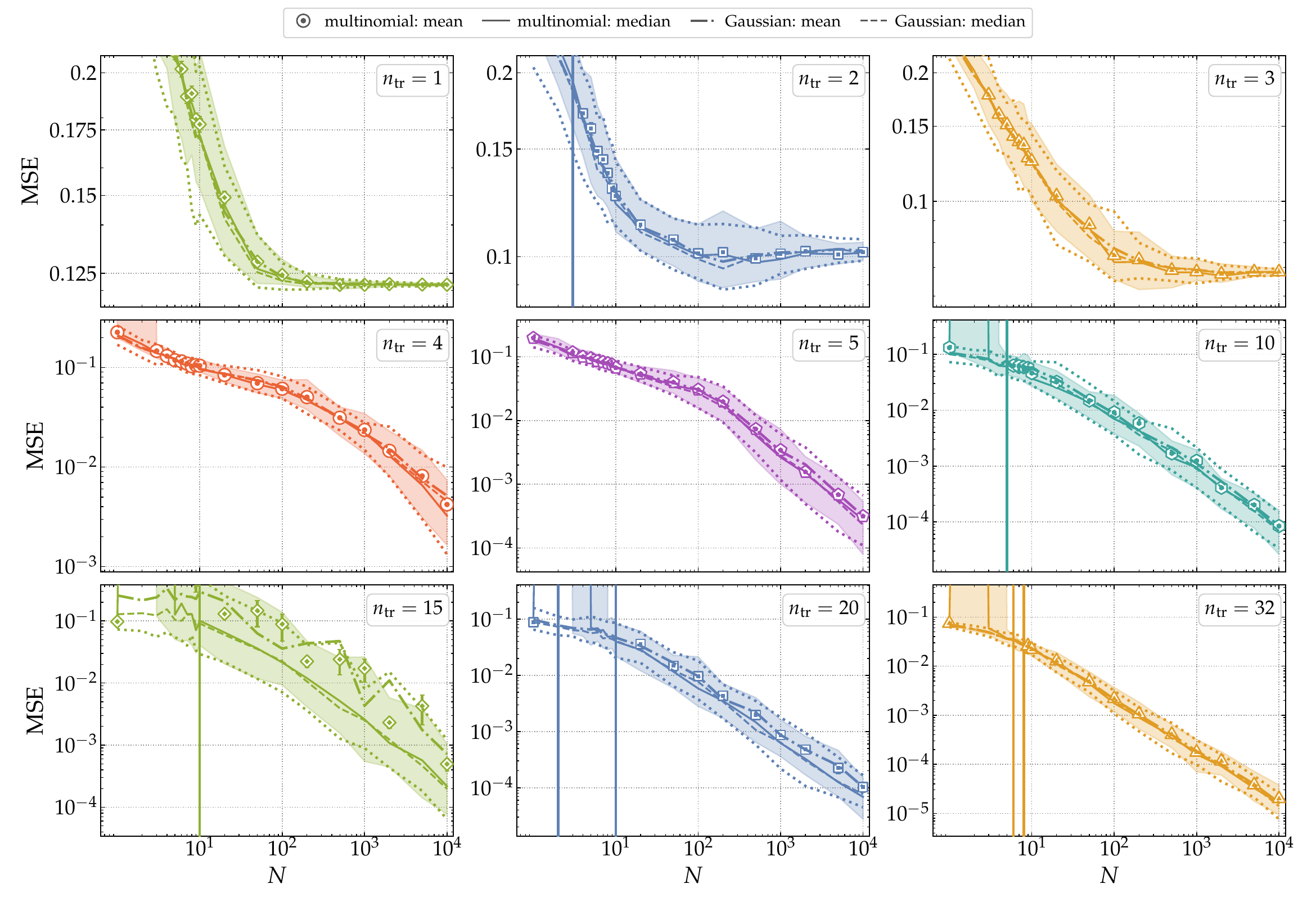}
    \caption{Single-experiment MSE of pseudoinverse regression under multinomial shot noise and under its Gaussian approximation, for a qubit ($d=2$) measured with the MUB (upper panel) and a random POVM ($n_{\rm out}=6$ outcomes) (lower panel), as a function of the number $N$ of shots per training state. An experiment is one draw of the shot noise on the $n_{\rm tr}$ training states; its MSE is the squared prediction error averaged over $50$ Haar-random test states and $20$ rank-one observables. 
    Multinomial noise, $\hat{p}=\mathrm{Mult}(N,p)/N$: markers show the mean, with error bars giving its standard error; the solid line shows the median; the shaded band spans the 10th--90th percentiles. Gaussian approximation, $\hat{p}=p+\xi/\sqrt{N}$ with $\xi\sim\mathcal{N}\big(0,\mathrm{diag}(p)-pp^{\top}\big)$: the dash-dotted line shows the mean, the dashed line the median, and the dotted lines the 10th and 90th percentiles. For $n_{\rm tr}<d^2=4$ the training states do not span the space of qubit operators and the MSE saturates at large $N$; for $n_{\rm tr}\geq 4$ it keeps decreasing. Near the interpolation threshold $n_{\rm tr}\simeq n_{\rm out}$ the single-experiment MSE is heavy-tailed for both noise models, so the mean is dominated by rare experiments and fluctuates strongly; means outside the plotted range are not shown.}
    \label{fig:multi_vs_gauss_1}
\end{figure}

\newpage
\subsection{Simulations with random POVMs}
\label{subsec:random_povm}

We next investigate whether the resource scalings derived above persist for randomly generated POVMs, beyond the symmetric measurement designs for which explicit ensemble averages are available. \Cref{fig:povm_random_M_inf_vs_N} and \cref{fig:povm_random_M_1000_vs_N} consider a qubit with $n_{\mathrm{out}}\in\{8,16,32,64\}$, trained on $n_{\mathrm{tr}}=256$ states and evaluated on $200$ test states. The plotted markers and bands summarize the median and the $10$th-$90$th percentiles over $100$ random target observables.

With exact test probabilities,  \cref{fig:povm_random_M_inf_vs_N} displays the two regimes predicted by section \cref{app:asymptotic}. At lower training statistics, the systematic bias dominates and the MSE decreases approximately as $N^{-2}$. At higher statistics, the training fluctuations dominate, giving the slower decay $(Nn_{\mathrm{tr}})^{-1}$. The symmetric-design expressions provide reference curves for these contributions. The random POVM results reproduce their scaling and show only a weak dependence on the number of outcomes over the range considered, although the smallest POVM exhibits a larger error and a broader spread.

Finite test statistics change the large-$N$ behavior. In \cref{fig:povm_random_M_1000_vs_N}, where $M=10^3$, the initial decrease is arrested by a plateau arising from the test-sampling contribution
\begin{equation}
\frac{1}{M}\,
\mathbb E_{\mathrm{tr}}\!\left[
\widehat{\boldsymbol w}_N
\Sigma_\sigma
\widehat{\boldsymbol w}_N^\top
\right].
\end{equation}
The symmetric qubit value $2/(3M)$ is shown as a reference. The observed saturation illustrates that increasing the training statistics alone cannot eliminate the error caused by a fixed test-measurement budget.

\Cref{fig:povm_random_single_experiment} examines a single fixed regression instance, allowing the analytical prediction to be compared with simulations without averaging over different measurement geometries. We vary the training-target noise, the test-target noise, and the test budget, and compare the empirical MSE with
\begin{equation}
f(N)=C_0+\frac{C_1}{N}+\frac{C_2}{N^2}+C_3N,
\end{equation}
using coefficients evaluated for that instance. The decreasing terms describe the training fluctuation and squared-bias contributions, while the constant term includes the noise floors. When training-target noise and finite test sampling are both present, the term proportional to $N$ produces an eventual increase of the error: more precise training features amplify target noise through the increasingly small singular values of the empirical feature matrix. With exact test probabilities, this growing contribution is absent. Independent test-target noise instead adds an irreducible constant contribution. The agreement shown in the figure supports this separation of the different error mechanisms.

Finally, \cref{fig:povm_random_M_inf_vs_nout} and \cref{fig:povm_random_M_1000_vs_nout} vary $n_{\mathrm{out}}$ at fixed training set size. The error initially decreases as the POVM becomes more redundant, remains comparatively flat at intermediate outcome numbers, and rises sharply near the interpolation threshold $n_{\mathrm{out}}\simeq n_{\mathrm{tr}}$. Beyond this region, it decreases again. Finite test sampling strongly enhances the peak because poorly conditioned regressors also amplify fluctuations in the test probabilities. Away from the peak, the finite-$M$ curves become largely insensitive to $N$ once test sampling dominates. These scans show that additional measurement outcomes can improve performance, but their benefit depends on the conditioning of the empirical regression problem and on the available training and test statistics.



\begin{figure}[htbp]
    \centering
    \begin{minipage}[t]{0.48\textwidth}
        \centering
        \includegraphics[width=\linewidth]{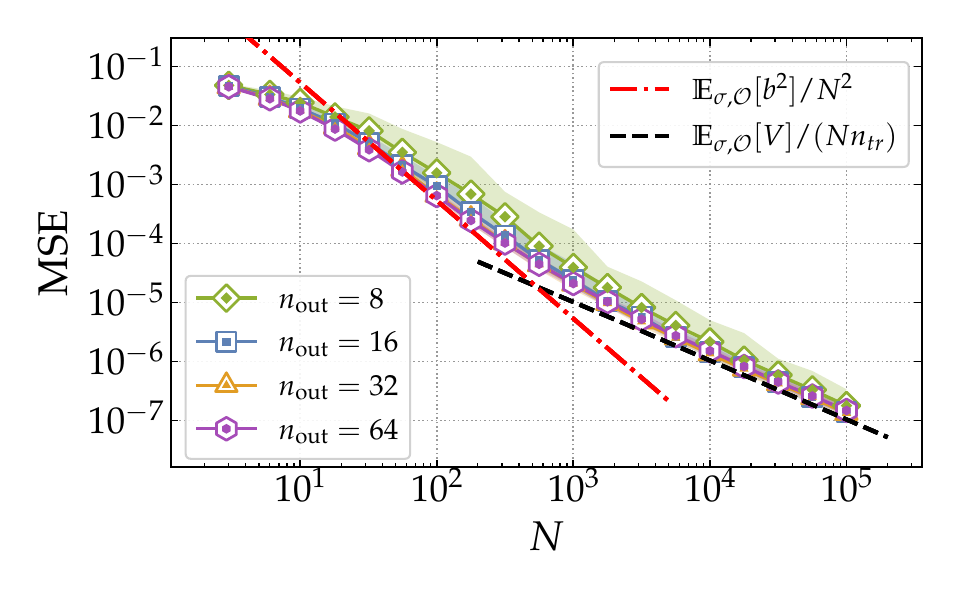}
    \caption{Test MSE as a function of $N$ for a single qubit ($d=2$) measured with random POVMs with $n_{\mathrm{out}} \in \{8, 16, 32, 64\}$ outcomes, with exact outcome probabilities ($M = \infty$). Markers show the median over $n_{\mathrm{obs}} = 100$ random observables $\mathcal{O}$, and shaded bands show the 10th--90th percentile range. The model is trained on $n_{\mathrm{tr}} = 256$ states and tested on $200$ states. Two regimes appear. For small $N$, the error is dominated by the bias and decays as $N^{-2}$, following $\mathbb{E}_{\sigma,\mathcal{O}}[b^2]/N^2$ (red dash-dotted line). For large $N$, the variance takes over and the error decays as $N^{-1}$, following $\mathbb{E}_{\sigma,\mathcal{O}}[V]/(N n_{\mathrm{tr}})$ (black dashed line). Unlike the finite-$M$ case, there is no plateau. The curves depend only weakly on $n_{\mathrm{out}}$: the $n_{\mathrm{out}} = 8$ curve sits slightly higher and has a wider spread.}
    \label{fig:povm_random_M_inf_vs_N}
    \end{minipage}
    \hfill
            \begin{minipage}[t]{0.48\textwidth}
        \centering
        \includegraphics[width=\linewidth]{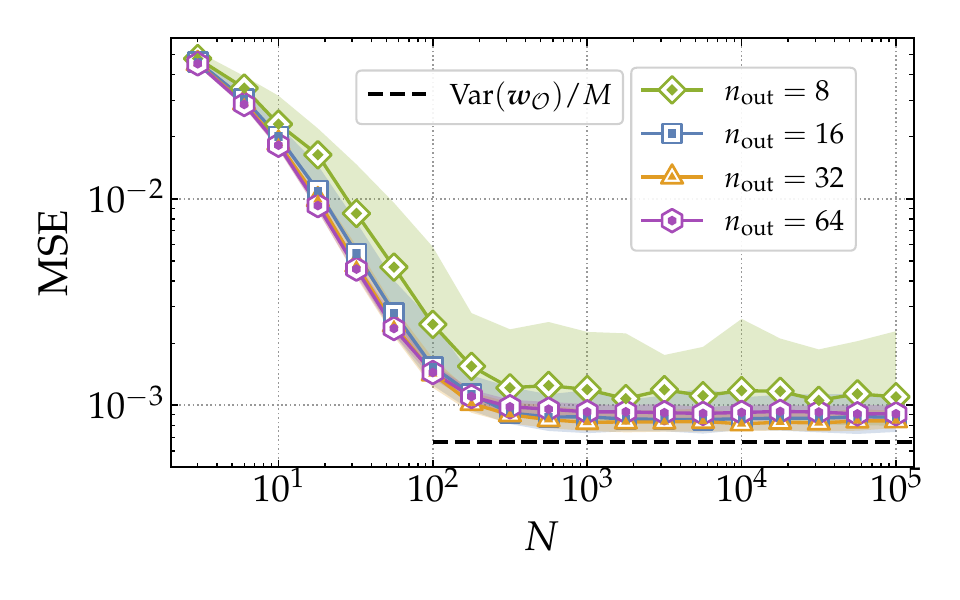}
        \caption{Same as Fig.~\ref{fig:povm_random_M_inf_vs_N}, but with finite statistics of $M = 10^3$ shots per state. As $N$ grows, the MSE drops and then levels off at a plateau set by shot noise, instead of continuing to decay. The dashed line is the limit $\mathrm{Var}(\boldsymbol{w}_{\mathcal{O}})/M = 2/(3M)$. As in the $M=\infty$ case, the curves depend only weakly on $n_{\mathrm{out}}$: the $n_{\mathrm{out}} = 8$ curve sits slightly higher and has a wider spread.}
        \label{fig:povm_random_M_1000_vs_N}
    \end{minipage}
\end{figure}


\begin{figure}[t]
    \centering
    \includegraphics[width=0.75\linewidth]{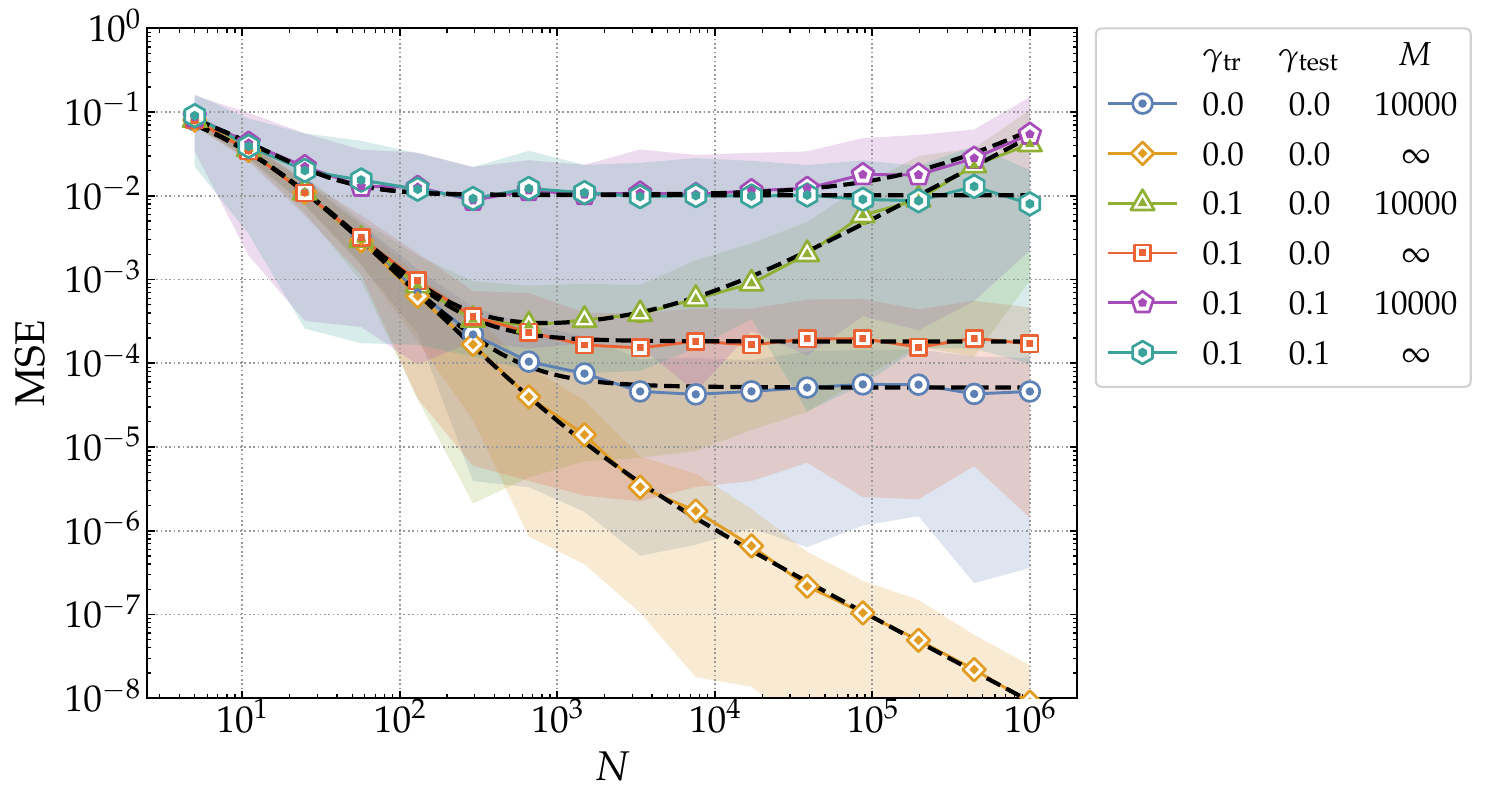}
\caption{Empirical MSE of a single fixed regression instance compared with the prediction $f(N) = C_0 + C_1/N + C_2/N^2 + C_3 N$, as a function of the number $N$ of training shots per training state. The instance has $d=2$, a Haar-random rank-one POVM with $n_{\rm out}=16$ outcomes, $n_{\rm tr}=256$ training states, one test state $\sigma$ and one target observable. Each curve corresponds to one choice of the training label-noise level $\gamma_{\rm tr}$, the test label-noise level $\gamma_{\rm test}$ and the number of test shots $M$, as given in the legend. For each $N$, all noise sources are sampled jointly over $100$ independent realizations. Markers show the mean squared error and shaded bands span the 10th--90th percentiles. Black dashed lines show $f(N)$, with coefficients evaluated in closed form for this instance.}
\label{fig:povm_random_single_experiment}
\end{figure}

\begin{figure}[htbp]
    \centering
        \begin{minipage}[t]{0.48\textwidth}
        \centering
        \includegraphics[width=\linewidth]{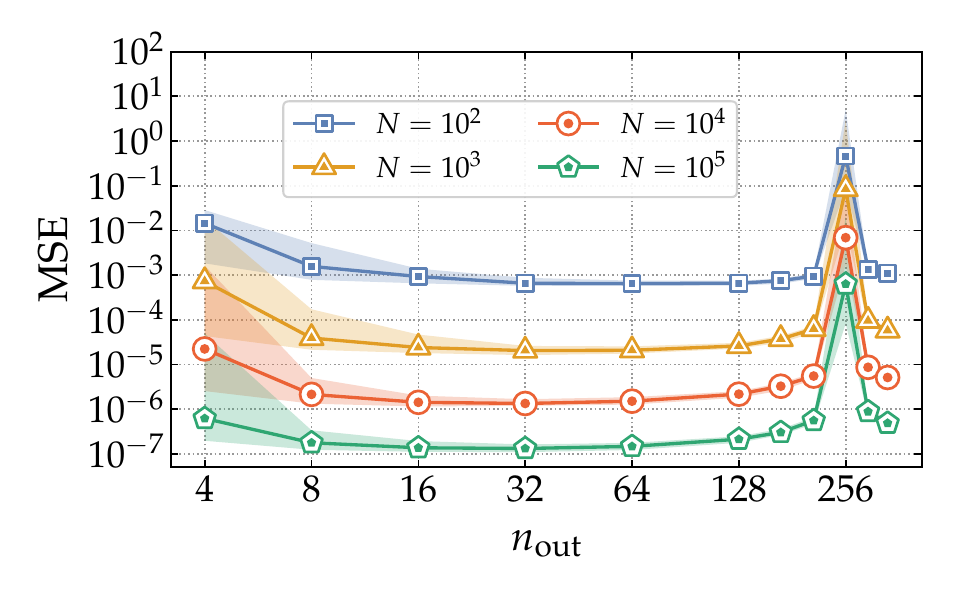}
        \caption{Test MSE as a function of the number of POVM outcomes $n_{\mathrm{out}}$ for a single qubit ($d=2$) measured with random POVMs and exact outcome probabilities ($M = \infty$), for $N \in \{10^2, 10^3, 10^4, 10^5\}$. Markers show the median over $n_{\mathrm{obs}} = 100$ random observables $\mathcal{O}$, and shaded bands show the 10th--90th percentile range. The model is trained on $n_{\mathrm{tr}} = 256$ states and tested on $200$ states. The MSE falls quickly from $n_{\mathrm{out}} = 4$ and stays nearly flat between $n_{\mathrm{out}} \approx 16$ and $64$. It then rises and peaks sharply at $n_{\mathrm{out}} = n_{\mathrm{tr}}$, where the linear system is square and ill-conditioned (interpolation threshold), before decreasing again for $n_{\mathrm{out}} > n_{\mathrm{tr}}$. At every $n_{\mathrm{out}}$, increasing $N$ lowers the error by roughly an order of magnitude per decade, in line with the $N^{-1}$ scaling of Fig.~\ref{fig:povm_random_M_inf_vs_N}.}
        \label{fig:povm_random_M_inf_vs_nout}
    \end{minipage}
    \hfill
    \begin{minipage}[t]{0.48\textwidth}
        \centering
        \includegraphics[width=\linewidth]{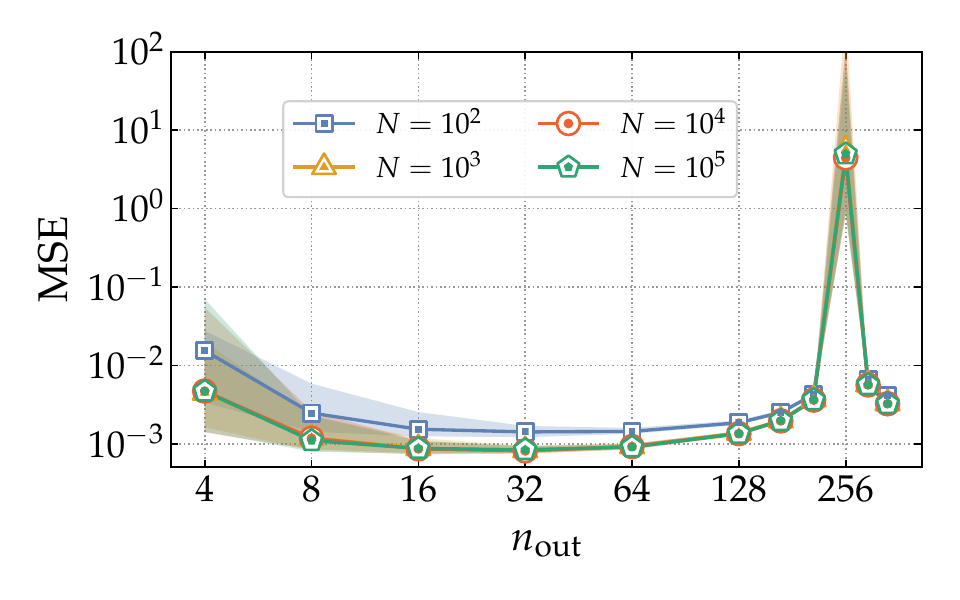}
        \caption{Same as Fig.~\ref{fig:povm_random_M_inf_vs_nout}, but with finite statistics of $M = 10^3$ shots per state. For $N \gtrsim 10^3$ the curves collapse onto one another, because the error is set by the shot-noise floor $\mathrm{Var}(\boldsymbol{w}_{\mathcal{O}})/M$ rather than by $N$ (cf.\ Fig.~\ref{fig:povm_random_M_1000_vs_N}). The dependence on $n_{\mathrm{out}}$ is the same as for $M=\infty$: the error is lowest for $16 \lesssim n_{\mathrm{out}} \lesssim 64$ and peaks at $n_{\mathrm{out}} = n_{\mathrm{tr}} = 256$. Shot noise makes the peak much stronger, with an MSE of order unity or higher.}
        \label{fig:povm_random_M_1000_vs_nout}
    \end{minipage}
\end{figure}

\end{widetext}

\bibliography{bibliography}

@article{jaeger2004harnessing,
  title={Harnessing nonlinearity: Predicting chaotic systems and saving energy in wireless communication},
  author={Jaeger, Herbert and Haas, Harald},
  journal={science},
  volume={304},
  number={5667},
  pages={78--80},
  year={2004},
  publisher={American Association for the Advancement of Science}
}

@article{huang2006extreme,
  title={Extreme learning machine: theory and applications},
  author={Huang, Guang-Bin and Zhu, Qin-Yu and Siew, Chee-Kheong},
  journal={Neurocomputing},
  volume={70},
  number={1-3},
  pages={489--501},
  year={2006},
  publisher={Elsevier}
}

@book{scholkopf2001learning,
  title={Learning with kernels: support vector machines, regularization, optimization, and beyond},
  author={Sch{\"o}lkopf, Bernhard and Smola, Alexander J},
  year={2001},
  publisher={the MIT Press}
}

@article{tripledescend,
  author        = {Lo Monaco, Gabriele and others},
  title         = {Triple descent in quantum linear regression},
  journal       = {},
  year          = {},
  note          = {to appear},
  eprint        = {},
  archivePrefix = {},
}

@article{jager2023universal,
  author = {J{\"a}ger, Jonas and Krems, Roman V},
  title = {Universal expressiveness of variational quantum classifiers and quantum kernels for support vector machines},
  journal = {Nature Communications},
  volume = {14},
  number = {1},
  pages = {576},
  year = {2023},
  doi = {10.1038/s41467-023-36144-5},
  eprint = {2207.05865},
  archiveprefix = {arXiv}
}

@misc{jager2026provable,
  title = {Provable and Scalable Quantum {{Gaussian}} Processes for Quantum Learning},
  author = {J{\"a}ger, Jonas and Braccia, Paolo and Bermejo, Pablo and Algaba, Manuel G. and {Garc{\'i}a-Mart{\'i}n}, Diego and Cerezo, M.},
  year = 2026,
  month = apr,
  number = {arXiv:2605.00099},
  eprint = {2605.00099},
  primaryclass = {quant-ph},
  publisher = {arXiv},
  doi = {10.48550/arXiv.2605.00099},
  archiveprefix = {arXiv}
}

@misc{jager2026QuantumGaussianProcesses,
  title = {Quantum {{Gaussian}} Processes for Prediction of Channel Observations},
  author = {J{\"a}ger, Jonas and Khmelnitskiy, Yaroslav and Braccia, Paolo and Miroszewski, Artur and {Garc{\'i}a-Mart{\'i}n}, Diego and Cerezo, M. and Czarnik, Piotr},
  year = 2026,
  month = aug,
  number = {arXiv:2608.19306},
  eprint = {2608.19306},
  primaryclass = {quant-ph},
  publisher = {arXiv},
  doi = {10.48550/arXiv.2608.19306},
  archiveprefix = {arXiv}
}

@article{wedin1973perturbation,
  author = {Wedin, Per-{\AA}ke},
  title = {Perturbation theory for pseudo-inverses},
  journal = {BIT Numerical Mathematics},
  volume = {13},
  number = {2},
  pages = {217--232},
  year = {1973},
  doi = {10.1007/BF01933494}
}

@article{zia2025quantum,
  author = {Zia, Danilo and Innocenti, Luca and Minati, Giorgio and Lorenzo, Salvatore and Suprano, Alessia and Di Bartolo, Rosario and Spagnolo, Nicol{\`o} and Giordani, Taira and Cimini, Valeria and Palma, G. Massimo and Ferraro, Alessandro and Sciarrino, Fabio and Paternostro, Mauro},
  title = {Quantum Reservoir Computing for Photonic Entanglement Witnessing},
  journal = {Science Advances},
  volume = {11},
  number = {50},
  pages = {eady7987},
  year = {2025},
  month = dec,
  doi = {10.1126/sciadv.ady7987},
  eprint = {2502.18361},
  archiveprefix = {arXiv},
  primaryclass = {quant-ph}
}

@article{suprano2024experimental,
  author = {Suprano, Alessia and Zia, Danilo and Innocenti, Luca and Lorenzo, Salvatore and Cimini, Valeria and Giordani, Taira and Palmisano, Ivan and Polino, Emanuele and Spagnolo, Nicol{\`o} and Sciarrino, Fabio and Palma, G. Massimo and Ferraro, Alessandro and Paternostro, Mauro},
  title = {Experimental {{Property Reconstruction}} in a {{Photonic Quantum Extreme Learning Machine}}},
  journal = {Physical Review Letters},
  volume = {132},
  number = {16},
  pages = {160802},
  year = {2024},
  month = apr,
  issn = {0031-9007},
  doi = {10.1103/PhysRevLett.132.160802},
  eprint = {2308.04543},
  archiveprefix = {arXiv},
  primaryclass = {quant-ph},
  langid = {english}
}

@article{lomonaco2024quantum,
  author = {Lo Monaco, Gabriele and Bertini, Marco and Lorenzo, Salvatore and Palma, G. Massimo},
  title = {Quantum Extreme Learning of Molecular Potential Energy Surfaces and Force Fields},
  journal = {Machine Learning: Science and Technology},
  volume = {5},
  number = {3},
  pages = {035014},
  year = {2024},
  issn = {2632-2153},
  doi = {10.1088/2632-2153/ad6120},
  eprint = {2406.14607},
  archiveprefix = {arXiv},
  primaryclass = {quant-ph}
}

@article{khan2024practical,
  author = {Khan, Saeed A and Kaufman, Ryan and Mesits, Boris and Hatridge, Michael and T{\"u}reci, Hakan E},
  title = {Practical Trainable Temporal Postprocessor for Multistate Quantum Measurement},
  journal = {PRX Quantum},
  volume = {5},
  number = {2},
  pages = {020364},
  year = {2024},
  doi = {10.1103/PRXQuantum.5.020364},
  eprint = {2310.18519},
  archiveprefix = {arXiv}
}

@article{hu2023tackling,
  author = {Hu, Fangjun and Angelatos, Gerasimos and Khan, Saeed A and Vives, Marti and T{\"u}reci, Esin and Bello, Leon and Rowlands, Graham E and Ribeill, Guilhem J and T{\"u}reci, Hakan E},
  title = {Tackling sampling noise in physical systems for machine learning applications: Fundamental limits and eigentasks},
  journal = {Physical Review X},
  volume = {13},
  number = {4},
  pages = {041020},
  year = {2023},
  doi = {10.1103/PhysRevX.13.041020},
  eprint = {2307.16083},
  archiveprefix = {arXiv}
}

@article{kempkes2026double,
  author = {Kempkes, Marie and Ijaz, Aroosa and {Gil-Fuster}, Elies and {Bravo-Prieto}, Carlos and Spiegelberg, Jakob and Van Nieuwenburg, Evert and Dunjko, Vedran},
  title = {Double {{Descent}} in {{Quantum Kernel Methods}}},
  journal = {PRX Quantum},
  volume = {7},
  number = {1},
  pages = {010312},
  year = {2026},
  month = jan,
  issn = {2691-3399},
  doi = {10.1103/cn64-gs6b},
  eprint = {2501.10077},
  archiveprefix = {arXiv},
  primaryclass = {quant-ph},
  urldate = {2026-08-01},
  langid = {english}
}

@article{liang2020just,
  author = {Liang, Tengyuan and Rakhlin, Alexander},
  title = {Just interpolate: Kernel “ridgeless” regression can generalize},
  journal = {The Annals of Statistics},
  volume = {48},
  number = {3},
  pages = {1329--1347},
  year = {2020},
  doi = {10.1214/19-AOS1849},
  eprint = {1808.00387},
  archiveprefix = {arXiv}
}

@article{dobriban2018high,
  author = {Dobriban, Edgar and Wager, Stefan},
  title = {High-dimensional asymptotics of prediction: Ridge regression and classification},
  journal = {The Annals of Statistics},
  volume = {46},
  number = {1},
  pages = {247--279},
  year = {2018},
  doi = {10.1214/17-AOS1549},
  eprint = {1507.03003},
  archiveprefix = {arXiv}
}

@article{bishop1995training,
  author = {Bishop, Chris M},
  title = {Training with noise is equivalent to Tikhonov regularization},
  journal = {Neural Computation},
  volume = {7},
  number = {1},
  pages = {108--116},
  year = {1995},
  doi = {10.1162/neco.1995.7.1.108}
}

@inproceedings{belkin2018understand,
  author = {Belkin, Mikhail and Ma, Siyuan and Mandal, Soumik},
  editor = {Dy, Jennifer and Krause, Andreas},
  title = {To understand deep learning we need to understand kernel learning},
  booktitle = {Proceedings of the 35th International Conference on Machine Learning},
  series = {Proceedings of Machine Learning Research},
  volume = {80},
  pages = {541--549},
  year = {2018},
  publisher = {PMLR},
  eprint = {1802.01396},
  archiveprefix = {arXiv}
}

@article{belkin2019ReconcilingModernMachine,
  title = {Reconciling Modern Machine Learning Practice and the Bias-Variance Trade-Off},
  author = {Belkin, Mikhail and Hsu, Daniel and Ma, Siyuan and Mandal, Soumik},
  year = 2019,
  month = aug,
  journal = {Proceedings of the National Academy of Sciences},
  volume = {116},
  number = {32},
  eprint = {1812.11118},
  primaryclass = {stat.ML},
  pages = {15849--15854},
  issn = {0027-8424, 1091-6490},
  doi = {10.1073/pnas.1903070116},
  urldate = {2026-09-26},
  archiveprefix = {arXiv}
}

@article{muthukumar2020harmless,
  author = {Muthukumar, Vidya and Vodrahalli, Kailas and Subramanian, Vignesh and Sahai, Anant},
  title = {Harmless interpolation of noisy data in regression},
  journal = {IEEE Journal on Selected Areas in Information Theory},
  volume = {1},
  number = {1},
  pages = {67--83},
  year = {2020},
  doi = {10.1109/JSAIT.2020.2984716},
  eprint = {1903.09139},
  archiveprefix = {arXiv}
}

@inproceedings{derezinski2020exact,
  author = {Derezinski, Michal and Liang, Feynman T and Mahoney, Michael W},
  title = {Exact expressions for double descent and implicit regularization via surrogate random design},
  booktitle = {Advances in Neural Information Processing Systems},
  volume = {33},
  pages = {5152--5164},
  year = {2020},
  eprint = {1912.04533},
  archiveprefix = {arXiv}
}

@article{bartlett2020benign,
  author = {Bartlett, Peter L and Long, Philip M and Lugosi, G{\'a}bor and Tsigler, Alexander},
  title = {Benign overfitting in linear regression},
  journal = {Proceedings of the National Academy of Sciences},
  volume = {117},
  number = {48},
  pages = {30063--30070},
  year = {2020},
  doi = {10.1073/pnas.1907378117},
  eprint = {1906.11300},
  archiveprefix = {arXiv}
}

@inproceedings{bibas2019new,
  author = {Bibas, Koby and Fogel, Yaniv and Feder, Meir},
  title = {A new look at an old problem: A universal learning approach to linear regression},
  booktitle = {2019 IEEE International Symposium on Information Theory (ISIT)},
  pages = {2304--2308},
  year = {2019},
  organization = {IEEE},
  doi = {10.1109/ISIT.2019.8849398},
  eprint = {1905.04708},
  archiveprefix = {arXiv}
}

@article{mei2022generalization,
  author = {Mei, Song and Montanari, Andrea},
  title = {The generalization error of random features regression: Precise asymptotics and the double descent curve},
  journal = {Communications on Pure and Applied Mathematics},
  volume = {75},
  number = {4},
  pages = {667--766},
  year = {2022},
  doi = {10.1002/cpa.22008},
  eprint = {1908.05355},
  archiveprefix = {arXiv}
}

@article{hastie2022surprises,
  author = {Hastie, Trevor and Montanari, Andrea and Rosset, Saharon and Tibshirani, Ryan J},
  title = {Surprises in high-dimensional ridgeless least squares interpolation},
  journal = {The Annals of Statistics},
  volume = {50},
  number = {2},
  pages = {949--986},
  year = {2022},
  doi = {10.1214/21-AOS2133},
  eprint = {1903.08560},
  eprinttype = {arXiv},
  eprintclass = {math}
}

@article{innocenti2023shadow,
  author = {Innocenti, Luca and Lorenzo, S. and Palmisano, I. and Albarelli, F. and Ferraro, A. and Paternostro, M. and Palma, G. M.},
  title = {Shadow {{Tomography}} on {{General Measurement Frames}}},
  journal = {PRX Quantum},
  volume = {4},
  number = {4},
  pages = {040328},
  year = {2023},
  month = nov,
  issn = {2691-3399},
  doi = {10.1103/PRXQuantum.4.040328},
  eprint = {2301.13229},
  archiveprefix = {arXiv},
  primaryclass = {quant-ph},
  urldate = {2025-04-22},
  langid = {english}
}

@book{christensen2003introduction,
  author = {Christensen, Ole},
  title = {An {{Introduction}} to {{Frames}} and {{Riesz Bases}}},
  series = {Applied and {{Numerical Harmonic Analysis}}},
  edition = {2},
  year = {2016},
  publisher = {Birkh\"auser Cham},
  address = {Cham},
  isbn = {978-3-319-25611-5},
  doi = {10.1007/978-3-319-25613-9},
  urldate = {2025-10-15}
}

@incollection{casazza2016brief,
  author = {Casazza, Peter G. and Lynch, Richard G.},
  editor = {Okoudjou, Kasso A.},
  title = {A Brief Introduction to {{Hilbert}} Space Frame Theory and Its Applications},
  booktitle = {Finite {{Frame Theory}}: {{A Complete Introduction}} to {{Overcompleteness}}},
  series = {Proceedings of {{Symposia}} in {{Applied Mathematics}}},
  volume = {73},
  pages = {1--51},
  year = {2016},
  publisher = {American Mathematical Society},
  address = {Providence, Rhode Island},
  doi = {10.1090/psapm/073/00627},
  eprint = {1509.07347},
  archiveprefix = {arXiv},
  primaryclass = {math.FA},
  urldate = {2021-05-27}
}

@incollection{casazza2013introduction,
  author = {Casazza, Peter G. and Kutyniok, Gitta and Philipp, Friedrich},
  editor = {Casazza, Peter G. and Kutyniok, Gitta},
  title = {Introduction to {{Finite Frame Theory}}},
  booktitle = {Finite {{Frames}}: {{Theory}} and {{Applications}}},
  series = {Applied and {{Numerical Harmonic Analysis}}},
  pages = {1--53},
  year = {2013},
  publisher = {Birkh\"auser},
  address = {Boston},
  doi = {10.1007/978-0-8176-8373-3_1},
  urldate = {2025-09-22},
  langid = {english}
}

@book{gupta2018matrix,
  author = {Gupta, Arjun K and Nagar, Daya K},
  title = {Matrix variate distributions},
  series = {Monographs and Surveys in Pure and Applied Mathematics},
  volume = {104},
  edition = {1},
  year = {1999},
  publisher = {Chapman and Hall CRC},
  isbn = {978-1-58488-046-2},
  doi = {10.1201/9780203749289}
}

@book{muirhead2009aspects,
  author = {Muirhead, Robb J},
  title = {Aspects of multivariate statistical theory},
  series = {Wiley Series in Probability and Statistics},
  year = {2009},
  publisher = {John Wiley \& Sons},
  isbn = {978-0-470-31670-2},
  doi = {10.1002/9780470316559}
}

@inbook{eaton2007wishart,
  author = {Eaton, Morris L},
  title = {Chapter 8: The Wishart Distribution},
  booktitle = {Multivariate Statistics: A Vector Space Approach},
  series = {Lecture Notes--Monograph Series},
  volume = {53},
  pages = {302--333},
  year = {2007},
  publisher = {Institute of Mathematical Statistics},
  doi = {10.1214/lnms/1196285114}
}

@book{zhang2005SchurComplementIts,
  title = {The {{Schur Complement}} and {{Its Applications}}},
  editor = {Zhang, Fuzhen},
  year = 2005,
  series = {Numerical {{Methods}} and {{Algorithms}}},
  volume = {4},
  publisher = {Springer-Verlag},
  address = {New York},
  doi = {10.1007/b105056},
  urldate = {2025-05-21},
  copyright = {http://www.springer.com/tdm},
  isbn = {978-0-387-24271-2},
  langid = {english}
}

@article{gonzalez2015laurent,
  author = {Gonzalez-Rodriguez, Pedro and Moscoso, Miguel and Kindelan, Manuel},
  title = {Laurent expansion of the inverse of perturbed, singular matrices},
  journal = {Journal of Computational Physics},
  volume = {299},
  pages = {307--319},
  year = {2015},
  doi = {10.1016/j.jcp.2015.07.006}
}

@book{stewart1990Matrix,
  author = {Stewart, G. W. and Sun, Ji-guang},
  title = {Matrix {{Perturbation Theory}}},
  edition = {1},
  year = {1990},
  month = jul,
  publisher = {Academic Press},
  address = {Boston},
  isbn = {978-0-12-670230-9},
  langid = {english}
}

@article{stewart1977perturbation,
  author = {Stewart, G. W.},
  title = {On the {{Perturbation}} of {{Pseudo-Inverses}}, {{Projections}} and {{Linear Least Squares Problems}}},
  journal = {SIAM Review},
  volume = {19},
  number = {4},
  pages = {634--662},
  year = {1977},
  month = oct,
  issn = {0036-1445},
  doi = {10.1137/1019104},
  langid = {english}
}

@online{nakkiran2019MoreDataCan,
  author = {Nakkiran, Preetum},
  title = {More {{Data Can Hurt}} for {{Linear Regression}}: {{Sample-wise Double Descent}}},
  shorttitle = {More {{Data Can Hurt}} for {{Linear Regression}}},
  date = {2019-12-16},
  doi = {10.48550/arXiv.1912.07242},
  eprint = {1912.07242},
  eprinttype = {arXiv},
  eprintclass = {stat},
  url = {https://arxiv.org/abs/1912.07242},
  urldate = {2025-06-13},
  pubstate = {prepublished}
}

@book{hastie2009Elements,
  author = {Hastie, Trevor and Tibshirani, Robert and Friedman, Jerome},
  title = {The {{Elements}} of {{Statistical Learning}}},
  series = {Springer Series in Statistics},
  edition = {2},
  date = {2009},
  publisher = {Springer},
  location = {New York, NY},
  isbn = {978-0-387-84857-0},
  doi = {10.1007/978-0-387-84858-7}
}

@inproceedings{kubler2021inductive,
  author = {K{\"u}bler, Jonas and Buchholz, Simon and Sch{\"o}lkopf, Bernhard},
  title = {The inductive bias of quantum kernels},
  booktitle = {Advances in Neural Information Processing Systems},
  volume = {34},
  pages = {12661--12673},
  year = {2021},
  eprint = {2106.03747},
  archiveprefix = {arXiv}
}

@article{gentinetta2024complexity,
  author = {Gentinetta, Gian and Thomsen, Arne and Sutter, David and Woerner, Stefan},
  title = {The complexity of quantum support vector machines},
  journal = {Quantum},
  volume = {8},
  pages = {1225},
  year = {2024},
  doi = {10.22331/q-2024-01-11-1225},
  eprint = {2203.00031},
  archiveprefix = {arXiv}
}

@article{xiong2025fundamental,
  author = {Xiong, Weijie and Facelli, Giorgio and Sahebi, Mehrad and Agnel, Owen and Chotibut, Thiparat and Thanasilp, Supanut and Holmes, Zo{\"e}},
  title = {On fundamental aspects of quantum extreme learning machines},
  journal = {Quantum Machine Intelligence},
  volume = {7},
  number = {1},
  pages = {20},
  year = {2025},
  doi = {10.1007/s42484-025-00239-7},
  eprint = {2312.15124},
  archiveprefix = {arXiv}
}

@article{vetrano2025state,
  author = {Vetrano, Marco and Lo Monaco, Gabriele and Innocenti, Luca and Lorenzo, Salvatore and Palma, G. Massimo},
  title = {State Estimation with Quantum Extreme Learning Machines beyond the Scrambling Time},
  journal = {npj Quantum Information},
  volume = {11},
  number = {1},
  pages = {20},
  year = {2025},
  month = feb,
  issn = {2056-6387},
  doi = {10.1038/s41534-024-00927-5},
  eprint = {2409.06782},
  archiveprefix = {arXiv},
  primaryclass = {quant-ph},
  urldate = {2025-04-22},
  langid = {english}
}

@article{innocenti2023potential,
  author = {Innocenti, Luca and Lorenzo, Salvatore and Palmisano, Ivan and Ferraro, Alessandro and Paternostro, Mauro and Palma, G. Massimo},
  title = {Potential and Limitations of Quantum Extreme Learning Machines},
  journal = {Communications Physics},
  volume = {6},
  number = {1},
  pages = {118},
  year = {2023},
  month = may,
  issn = {2399-3650},
  doi = {10.1038/s42005-023-01233-w},
  eprint = {2210.00780},
  archiveprefix = {arXiv},
  primaryclass = {quant-ph},
  urldate = {2023-08-05}
}

@inproceedings{li2022concentration,
  author = {Li, Guangxi and Ye, Ruilin and Zhao, Xuanqiang and Wang, Xin},
  title = {Concentration of data encoding in parameterized quantum circuits},
  booktitle = {Advances in Neural Information Processing Systems},
  volume = {35},
  pages = {19456--19469},
  year = {2022},
  eprint = {2206.08273},
  archiveprefix = {arXiv}
}

@article{stilck2021limitations,
  author = {Stilck Fran{\c{c}}a, Daniel and Garcia-Patron, Raul},
  title = {Limitations of optimization algorithms on noisy quantum devices},
  journal = {Nature Physics},
  volume = {17},
  number = {11},
  pages = {1221--1227},
  year = {2021},
  doi = {10.1038/s41567-021-01356-3},
  eprint = {2009.05532},
  archiveprefix = {arXiv}
}

@article{tangpanitanon2020expressibility,
  author = {Tangpanitanon, Jirawat and Thanasilp, Supanut and Dangniam, Ninnat and Lemonde, Marc-Antoine and Angelakis, Dimitris G},
  title = {Expressibility and trainability of parametrized analog quantum systems for machine learning applications},
  journal = {Physical Review Research},
  volume = {2},
  number = {4},
  pages = {043364},
  year = {2020},
  doi = {10.1103/PhysRevResearch.2.043364},
  eprint = {2005.11222},
  archiveprefix = {arXiv}
}

@article{holmes2022connecting,
  author = {Holmes, Zo{\"e} and Sharma, Kunal and Cerezo, Marco and Coles, Patrick J},
  title = {Connecting ansatz expressibility to gradient magnitudes and barren plateaus},
  journal = {PRX Quantum},
  volume = {3},
  number = {1},
  pages = {010313},
  year = {2022},
  doi = {10.1103/PRXQuantum.3.010313},
  eprint = {2101.02138},
  archiveprefix = {arXiv}
}

@article{thanasilp2024exponential,
  author = {Thanasilp, Supanut and Wang, Samson and Cerezo, M. and Holmes, Zo{\"e}},
  title = {Exponential Concentration in Quantum Kernel Methods},
  journal = {Nature Communications},
  volume = {15},
  number = {1},
  pages = {5200},
  year = {2024},
  month = jun,
  issn = {2041-1723},
  doi = {10.1038/s41467-024-49287-w},
  eprint = {2208.11060},
  archiveprefix = {arXiv},
  urldate = {2026-03-27},
  langid = {english}
}

@article{kairon2025equivalence,
  author = {Kairon, Pranav and J{\"a}ger, Jonas and Krems, Roman V},
  title = {Equivalence between exponential concentration in quantum machine learning kernels and barren plateaus in variational algorithms},
  journal = {Physical Review A},
  volume = {114},
  number = {2},
  pages = {022417},
  year = {2026},
  month = aug,
  doi = {10.1103/6hbr-kj64},
  eprint = {2501.07433},
  archiveprefix = {arXiv}
}

@article{cerezo2025does,
  author = {Cerezo, M. and Larocca, Martin and Garc{\'i}a-Mart{\'i}n, Diego and Diaz, N. L. and Braccia, Paolo and Fontana, Enrico and Rudolph, Manuel S. and Bermejo, Pablo and Ijaz, Aroosa and Thanasilp, Supanut and Anschuetz, Eric R. and Holmes, Zo{\"e}},
  title = {Does provable absence of barren plateaus imply classical simulability?},
  journal = {Nature Communications},
  volume = {16},
  number = {1},
  pages = {7907},
  year = {2025},
  doi = {10.1038/s41467-025-63099-6},
  eprint = {2312.09121},
  archiveprefix = {arXiv}
}

@article{cerezo2021VariationalQuantumAlgorithms,
  title = {Variational Quantum Algorithms},
  author = {Cerezo, M. and Arrasmith, Andrew and Babbush, Ryan and Benjamin, Simon C. and Endo, Suguru and Fujii, Keisuke and McClean, Jarrod R. and Mitarai, Kosuke and Yuan, Xiao and Cincio, Lukasz and Coles, Patrick J.},
  year = 2021,
  month = aug,
  journal = {Nature Reviews Physics},
  volume = {3},
  number = {9},
  eprint = {2012.09265},
  primaryclass = {quant-ph},
  pages = {625--644},
  issn = {2522-5820},
  doi = {10.1038/s42254-021-00348-9},
  archiveprefix = {arXiv},
  langid = {english}
}

@article{larocca2025barren,
  author = {Larocca, Martin and Thanasilp, Supanut and Wang, Samson and Sharma, Kunal and Biamonte, Jacob and Coles, Patrick J and Cincio, Lukasz and McClean, Jarrod R and Holmes, Zo{\"e} and Cerezo, Marco},
  title = {Barren plateaus in variational quantum computing},
  journal = {Nature Reviews Physics},
  volume = {7},
  number = {4},
  pages = {174--189},
  year = {2025},
  doi = {10.1038/s42254-025-00813-9},
  eprint = {2405.00781},
  archiveprefix = {arXiv}
}

@article{arrasmith2022equivalence,
  author = {Arrasmith, Andrew and Holmes, Zo{\"e} and Cerezo, Marco and Coles, Patrick J},
  title = {Equivalence of quantum barren plateaus to cost concentration and narrow gorges},
  journal = {Quantum Science and Technology},
  volume = {7},
  number = {4},
  pages = {045015},
  year = {2022},
  doi = {10.1088/2058-9565/ac7d06},
  eprint = {2104.05868},
  archiveprefix = {arXiv}
}

@article{mcclean2018barren,
  author = {McClean, Jarrod R and Boixo, Sergio and Smelyanskiy, Vadim N and Babbush, Ryan and Neven, Hartmut},
  title = {Barren plateaus in quantum neural network training landscapes},
  journal = {Nature Communications},
  volume = {9},
  number = {1},
  pages = {4812},
  year = {2018},
  doi = {10.1038/s41467-018-07090-4},
  eprint = {1803.11173},
  archiveprefix = {arXiv}
}

@article{cerezo2022ChallengesOpportunitiesQuantum,
  title = {Challenges and {{Opportunities}} in {{Quantum Machine Learning}}},
  author = {Cerezo, M. and Verdon, Guillaume and Huang, Hsin-Yuan and Cincio, Lukasz and Coles, Patrick J.},
  year = 2022,
  month = sep,
  journal = {Nature Computational Science},
  volume = {2},
  number = {9},
  eprint = {2303.09491},
  primaryclass = {quant-ph, stat},
  pages = {567--576},
  issn = {2662-8457},
  doi = {10.1038/s43588-022-00311-3},
  archiveprefix = {arXiv}
}

@article{wang2024ComprehensiveReviewQuantum,
  title = {A Comprehensive Review of Quantum Machine Learning: From {{NISQ}} to Fault Tolerance},
  shorttitle = {A Comprehensive Review of Quantum Machine Learning},
  author = {Wang, Yunfei and Liu, Junyu},
  year = 2024,
  month = nov,
  journal = {Reports on Progress in Physics},
  volume = {87},
  number = {11},
  eprint = {2401.11351},
  primaryclass = {quant-ph},
  pages = {116402},
  issn = {0034-4885, 1361-6633},
  doi = {10.1088/1361-6633/ad7f69},
  archiveprefix = {arXiv}
}

@article{jerbi2023QuantumMachineLearning,
  title = {Quantum Machine Learning beyond Kernel Methods},
  author = {Jerbi, Sofiene and Fiderer, Lukas J. and Poulsen Nautrup, Hendrik and K{\"u}bler, Jonas M. and Briegel, Hans J. and Dunjko, Vedran},
  year = 2023,
  month = jan,
  journal = {Nature Communications},
  volume = {14},
  number = {1},
  eprint = {2110.13162},
  primaryclass = {quant-ph},
  pages = {517},
  issn = {2041-1723},
  doi = {10.1038/s41467-023-36159-y},
  archiveprefix = {arXiv},
  langid = {english}
}

@misc{schuld2021SupervisedQuantumMachine,
  title = {Supervised Quantum Machine Learning Models Are Kernel Methods},
  author = {Schuld, Maria},
  year = 2021,
  month = apr,
  number = {arXiv:2101.11020},
  eprint = {2101.11020},
  primaryclass = {quant-ph, stat},
  publisher = {arXiv},
  archiveprefix = {arXiv}
}

@article{schuld2019QuantumMachineLearning,
  title = {Quantum {{Machine Learning}} in {{Feature Hilbert Spaces}}},
  author = {Schuld, Maria and Killoran, Nathan},
  year = 2019,
  month = feb,
  journal = {Physical Review Letters},
  volume = {122},
  number = {4},
  issn = {0031-9007, 1079-7114},
  doi = {10.1103/PhysRevLett.122.040504},
  urldate = {2020-11-15},
  langid = {english}
}

@article{fujii2017HarnessingDisorderedQuantum,
  title = {Harnessing Disordered Quantum Dynamics for Machine Learning},
  author = {Fujii, Keisuke and Nakajima, Kohei},
  year = 2017,
  month = aug,
  journal = {Physical Review Applied},
  volume = {8},
  number = {2},
  eprint = {1602.08159},
  primaryclass = {quant-ph},
  pages = {024030},
  issn = {2331-7019},
  doi = {10.1103/PhysRevApplied.8.024030},
  archiveprefix = {arXiv}
}

@article{mujal2021OpportunitiesQuantumReservoir,
  title = {Opportunities in {{Quantum Reservoir Computing}} and {{Extreme Learning Machines}}},
  author = {Mujal, Pere and {Mart{\'i}nez-Pe{\~n}a}, Rodrigo and Nokkala, Johannes and {Garc{\'i}a-Beni}, Jorge and Giorgi, Gian Luca and Soriano, Miguel C. and Zambrini, Roberta},
  year = 2021,
  journal = {Advanced Quantum Technologies},
  volume = {4},
  number = {8},
  pages = {2100027},
  issn = {2511-9044},
  doi = {10.1002/qute.202100027},
  copyright = {\copyright{} 2021 The Authors. Advanced Quantum Technologies published by Wiley-VCH GmbH},
  langid = {english}
}

@article{havlicek2019SupervisedLearningQuantum,
  title = {Supervised Learning with Quantum Enhanced Feature Spaces},
  author = {Havlicek, Vojtech and C{\'o}rcoles, Antonio D. and Temme, Kristan and Harrow, Aram W. and Kandala, Abhinav and Chow, Jerry M. and Gambetta, Jay M.},
  year = 2019,
  month = mar,
  journal = {Nature},
  volume = {567},
  number = {7747},
  eprint = {1804.11326},
  primaryclass = {quant-ph, stat},
  pages = {209--212},
  issn = {0028-0836, 1476-4687},
  doi = {10.1038/s41586-019-0980-2},
  archiveprefix = {arXiv}
}

@misc{hu2024GeneralizationErrorQuantum,
  title = {Generalization {{Error}} in {{Quantum Machine Learning}} in the {{Presence}} of {{Sampling Noise}}},
  author = {Hu, Fangjun and Gao, Xun},
  year = 2024,
  publisher = {arXiv},
  doi = {10.48550/ARXIV.2410.14654},
  urldate = {2026-09-27},
  copyright = {Creative Commons Attribution 4.0 International}
}

@article{sweke2020StochasticGradientDescent,
  title = {Stochastic Gradient Descent for Hybrid Quantum-Classical Optimization},
  author = {Sweke, Ryan and Wilde, Frederik and Meyer, Johannes and Schuld, Maria and Faehrmann, Paul K. and {Meynard-Piganeau}, Barth{\'e}l{\'e}my and Eisert, Jens},
  year = 2020,
  month = aug,
  journal = {Quantum},
  volume = {4},
  eprint = {1910.01155},
  primaryclass = {quant-ph},
  pages = {314},
  issn = {2521-327X},
  doi = {10.22331/q-2020-08-31-314},
  urldate = {2026-10-03},
  archiveprefix = {arXiv},
  langid = {english}
}

@article{bravyi2021MitigatingMeasurementErrors,
  title = {Mitigating Measurement Errors in Multiqubit Experiments},
  author = {Bravyi, Sergey and Sheldon, Sarah and Kandala, Abhinav and Mckay, David C. and Gambetta, Jay M.},
  year = 2021,
  month = apr,
  journal = {Physical Review A},
  volume = {103},
  number = {4},
  eprint = {2006.14044},
  primaryclass = {quant-ph},
  pages = {042605},
  issn = {2469-9926, 2469-9934},
  doi = {10.1103/PhysRevA.103.042605},
  urldate = {2026-05-14},
  archiveprefix = {arXiv},
  langid = {english}
}

@article{maciejewski2020MitigationReadoutNoise,
  title = {Mitigation of Readout Noise in Near-Term Quantum Devices by Classical Post-Processing Based on Detector Tomography},
  author = {Maciejewski, Filip B. and Zimbor{\'a}s, Zolt{\'a}n and Oszmaniec, Micha{\l}},
  year = 2020,
  month = apr,
  journal = {Quantum},
  volume = {4},
  eprint = {1907.08518},
  primaryclass = {quant-ph},
  pages = {257},
  issn = {2521-327X},
  doi = {10.22331/q-2020-04-24-257},
  urldate = {2026-10-03},
  archiveprefix = {arXiv},
  copyright = {https://creativecommons.org/licenses/by/4.0/},
  langid = {english}
}

@book{fuller1987MeasurementErrorModels,
  title = {Measurement {{Error Models}}},
  author = {Fuller, Wayne A.},
  year = 1987,
  month = jun,
  series = {Wiley {{Series}} in {{Probability}} and {{Statistics}}},
  edition = {1},
  publisher = {Wiley},
  doi = {10.1002/9780470316665},
  urldate = {2026-09-26},
  copyright = {http://doi.wiley.com/10.1002/tdm\_license\_1.1},
  isbn = {978-0-471-86187-4 978-0-470-31666-5},
  langid = {english}
}

@book{carroll2006MeasurementErrorNonlinear,
  title = {Measurement {{Error}} in {{Nonlinear Models}}},
  author = {Carroll, Raymond J. and Ruppert, David and Stefanski, Leonard A. and Crainiceanu, Ciprian M.},
  year = 2006,
  month = jun,
  edition = {0},
  publisher = {{Chapman and Hall/CRC}},
  doi = {10.1201/9781420010138},
  urldate = {2026-10-03},
  isbn = {978-1-4200-1013-8},
  langid = {english}
}

@article{golub1980AnalysisTotalLeast,
  title = {An {{Analysis}} of the {{Total Least Squares Problem}}},
  author = {Golub, Gene H. and Van Loan, Charles F.},
  year = 1980,
  month = dec,
  journal = {SIAM Journal on Numerical Analysis},
  volume = {17},
  number = {6},
  pages = {883--893},
  issn = {0036-1429, 1095-7170},
  doi = {10.1137/0717073},
  urldate = {2026-10-03},
  langid = {english}
}

@book{vanhuffel1991TotalLeastSquares,
  title = {The {{Total Least Squares Problem}}: {{Computational Aspects}} and {{Analysis}}},
  shorttitle = {The {{Total Least Squares Problem}}},
  author = {Van Huffel, Sabine and Vandewalle, Joos},
  year = 1991,
  month = jan,
  publisher = {{Society for Industrial and Applied Mathematics}},
  doi = {10.1137/1.9781611971002},
  urldate = {2026-10-03},
  isbn = {978-0-89871-275-9 978-1-61197-100-2},
  langid = {english}
}

@article{wang2021UnderstandingPowerQuantum,
  title = {Towards Understanding the Power of Quantum Kernels in the {{NISQ}} Era},
  author = {Wang, Xinbiao and Du, Yuxuan and Luo, Yong and Tao, Dacheng},
  year = 2021,
  month = aug,
  journal = {Quantum},
  volume = {5},
  eprint = {2103.16774},
  primaryclass = {quant-ph},
  pages = {531},
  issn = {2521-327X},
  doi = {10.22331/q-2021-08-30-531},
  urldate = {2026-10-03},
  archiveprefix = {arXiv},
  langid = {english}
}

@article{hubregtsen2022TrainingQuantumEmbedding,
  title = {Training Quantum Embedding Kernels on Near-Term Quantum Computers},
  author = {Hubregtsen, Thomas and Wierichs, David and {Gil-Fuster}, Elies and Derks, Peter-Jan H. S. and Faehrmann, Paul K. and Meyer, Johannes Jakob},
  year = 2022,
  month = oct,
  journal = {Physical Review A},
  volume = {106},
  number = {4},
  eprint = {2105.02276},
  primaryclass = {quant-ph},
  pages = {042431},
  issn = {2469-9926, 2469-9934},
  doi = {10.1103/PhysRevA.106.042431},
  urldate = {2026-10-03},
  archiveprefix = {arXiv},
  langid = {english}
}

@article{peters2023GeneralizationOverfittingQuantum,
  title = {Generalization despite Overfitting in Quantum Machine Learning Models},
  author = {Peters, Evan and Schuld, Maria},
  year = 2023,
  month = dec,
  journal = {Quantum},
  volume = {7},
  eprint = {2209.05523},
  primaryclass = {quant-ph},
  pages = {1210},
  issn = {2521-327X},
  doi = {10.22331/q-2023-12-20-1210},
  urldate = {2026-03-28},
  archiveprefix = {arXiv},
  langid = {english}
}

@article{huang2020PredictingManyProperties,
  title = {Predicting Many Properties of a Quantum System from Very Few Measurements},
  author = {Huang, Hsin-Yuan and Kueng, Richard and Preskill, John},
  year = 2020,
  month = oct,
  journal = {Nature Physics},
  volume = {16},
  number = {10},
  eprint = {2002.08953},
  primaryclass = {quant-ph},
  pages = {1050--1057},
  publisher = {Nature Publishing Group},
  issn = {1745-2481},
  doi = {10.1038/s41567-020-0932-7},
  urldate = {2020-11-30},
  archiveprefix = {arXiv},
  copyright = {2020 The Author(s), under exclusive licence to Springer Nature Limited},
  langid = {english}
}

@article{zhu2014QuantumStateEstimation,
  title = {Quantum State Estimation with Informationally Overcomplete Measurements},
  author = {Zhu, Huangjun},
  year = 2014,
  month = jul,
  journal = {Physical Review A},
  volume = {90},
  number = {1},
  eprint = {1404.3453},
  primaryclass = {quant-ph},
  pages = {012115},
  issn = {1050-2947, 1094-1622},
  doi = {10.1103/PhysRevA.90.012115},
  urldate = {2025-10-12},
  archiveprefix = {arXiv},
  copyright = {http://link.aps.org/licenses/aps-default-license},
  langid = {english}
}

@article{garciaperez2021LearningMeasureAdaptive,
  title = {Learning to {{Measure}}: {{Adaptive Informationally Complete Generalized Measurements}} for {{Quantum Algorithms}}},
  shorttitle = {Learning to {{Measure}}},
  author = {{Garc{\'i}a-P{\'e}rez}, Guillermo and Rossi, Matteo A.C. and Sokolov, Boris and Tacchino, Francesco and Barkoutsos, Panagiotis Kl. and Mazzola, Guglielmo and Tavernelli, Ivano and Maniscalco, Sabrina},
  year = 2021,
  month = nov,
  journal = {PRX Quantum},
  volume = {2},
  number = {4},
  eprint = {2104.00569},
  primaryclass = {quant-ph},
  pages = {040342},
  issn = {2691-3399},
  doi = {10.1103/PRXQuantum.2.040342},
  urldate = {2025-10-22},
  archiveprefix = {arXiv},
  langid = {english}
}

@article{scott2006TightInformationallyComplete,
  title = {Tight Informationally Complete Quantum Measurements},
  author = {Scott, A. J.},
  year = 2006,
  month = oct,
  journal = {Journal of Physics A: Mathematical and General},
  volume = {39},
  number = {43},
  eprint = {quant-ph/0604049},
  pages = {13507--13530},
  issn = {0305-4470, 1361-6447},
  doi = {10.1088/0305-4470/39/43/009},
  urldate = {2021-10-06},
  archiveprefix = {arXiv}
}

@article{moussa2023ResourceFrugalOptimizer,
  title = {Resource Frugal Optimizer for Quantum Machine Learning},
  author = {Moussa, Charles and Hunter Gordon, Max and Baczyk, Michal and Cerezo, M and Cincio, Lukasz and Coles, Patrick J},
  year = 2023,
  month = oct,
  journal = {Quantum Science and Technology},
  volume = {8},
  number = {4},
  eprint = {2211.04965},
  primaryclass = {quant-ph},
  pages = {045019},
  issn = {2058-9565},
  doi = {10.1088/2058-9565/acef55},
  urldate = {2026-10-03},
  archiveprefix = {arXiv}
}

@article{hu2024OvercomingCoherenceTime,
  title = {Overcoming the Coherence Time Barrier in Quantum Machine Learning on Temporal Data},
  author = {Hu, Fangjun and Khan, Saeed A. and Bronn, Nicholas T. and Angelatos, Gerasimos and Rowlands, Graham E. and Ribeill, Guilhem J. and T{\"u}reci, Hakan E.},
  year = 2024,
  month = aug,
  journal = {Nature Communications},
  volume = {15},
  number = {1},
  eprint = {2312.16165},
  primaryclass = {quant-ph},
  pages = {7491},
  publisher = {Nature Publishing Group},
  issn = {2041-1723},
  doi = {10.1038/s41467-024-51162-7},
  urldate = {2026-02-06},
  archiveprefix = {arXiv},
  copyright = {2024 The Author(s)},
  langid = {english}
}

@article{ahmed2025OptimalTrainingFinitely,
  title = {Optimal Training of Finitely Sampled Quantum Reservoir Computers for Forecasting of Chaotic Dynamics},
  author = {Ahmed, Osama and Tennie, Felix and Magri, Luca},
  year = 2025,
  month = jun,
  journal = {Quantum Machine Intelligence},
  volume = {7},
  number = {1},
  eprint = {2409.01394},
  primaryclass = {quant-ph},
  pages = {31},
  issn = {2524-4906, 2524-4914},
  doi = {10.1007/s42484-025-00261-9},
  urldate = {2026-01-21},
  archiveprefix = {arXiv},
  langid = {english}
}

\end{document}